\documentclass[10pt, superscriptaddress,
aps,
prx,
floatfix,
]{revtex4-2}

\usepackage{fontspec}
\usepackage[english]{babel}
\babelprovide[import=en]{en}

\usepackage[hidelinks,colorlinks=true,linkcolor=teal, citecolor=blue, bookmarks=true,breaklinks=true]{hyperref}
\usepackage{amsmath,amssymb}
\usepackage{amsthm}
\usepackage{tikz}
\usepackage{graphicx, overpic}\usepackage{dcolumn}\usepackage{bm, color,colortbl,xcolor}\usepackage[inline]{enumitem}

\usepackage{mathtools}
\usepackage{braket}
\usepackage{soul}
\usepackage{cancel}
\usepackage[skins,breakable]{tcolorbox}
\usepackage{comment}

\setul{0.5ex}{0.3ex}
\setulcolor{red}
\usepackage{pifont}

\newtheorem{theorem}{Theorem}
\newtheorem{corollary}[theorem]{Corollary}
\newtheorem{lemma}[theorem]{Lemma}

\newtheorem{proposition}[theorem]{Proposition}

\newtheorem{definition}[theorem]{Definition}

\newtheorem{remark}[theorem]{Remark}

\newtheorem{result}{Result}
\newtheorem{preliminaries}[theorem]{Preliminaries}
\newtheorem{assumption}[theorem]{Assumption}
\newtheorem{prob}[theorem]{Problem}

\usepackage[capitalize,nameinlink]{cleveref}

\newcommand{\fnm}[1]{#1}
\newcommand{\sur}[1]{#1}
\newcommand{\EV}[1]{\mathbb{E}\left[#1\right]}
\newcommand{\T}{\mathcal{T}}
\newcommand{\E}{\mathbb{E}}

\newcommand{\Prob}{\mathbb{P}}
\newcommand{\EE}{\mathbb{E}}

\newcommand{\PP}{\mathbb{P}}

\newcommand{\1}{\mathbf{1}}

\newcommand{\eps}{\varepsilon}
\def\polylog{\mathrm{polylog}}

\newcommand{\R}{\mathbb{R}}
\newcommand{\C}{\mathbb{C}}
\newcommand{\xst}{\mathbf{x}_{\text{st}}}

\renewcommand{\d}{\mathrm{d}}
\renewcommand{\Re}{\mathop{\mathrm{Re}}}

\newcommand{\range}[1]{[{#1}]}
\newcommand{\norm}[1]{\|{#1}\|}
\newcommand{\ud}{\,\mathrm{d}}

\usepackage{afterpage}

\graphicspath{{./fig/}}

\begin{document}

\preprint{APS/123-QED}

\title{Efficient Quantum Simulation for Nonlinear Stochastic Differential Equations
}
\author{\fnm{Xiangyu} \sur{Li}}\email{xiangyu.li@pnnl.gov}

\affiliation{{Pacific Northwest National Laboratory}, {Richland}, {WA}, {USA}}

\author{\fnm{Ahmet Burak} \sur{Catli}} 
\author{\fnm{Ho Kiat} \sur{Lim}} 
\author{\fnm{Matthew} \sur{Pocrnic}} 
\affiliation{{Department of Physics}, {University of Toronto}, {{Toronto}, {ON}, {Canada}}}

\author{\fnm{Dong} \sur{An}}
\affiliation{{Beijing International Center for Mathematical Research, Peking University}, {Beijing}, {China}}

\author{\fnm{Jin-Peng} \sur{Liu}}
\affiliation{{Yau Mathematical Sciences Center, Tsinghua University}, {Beijing}, {China}}
\affiliation{{Yanqi Lake Beijing Institute of Mathematical Sciences and Applications}, {Beijing}, {China}}

\author{\fnm{Nathan} \sur{Wiebe}}
\email{nathan.wiebe@utoronto.ca}
\affiliation{{Department of Computer Science}, {University of Toronto}, {{Toronto}, {ON}, {Canada}}}
\affiliation{{Pacific Northwest National Laboratory}, {Richland}, {WA}, {USA}}
\affiliation{{Canadian Institute for Advanced Research}, {Toronto}, {ON}, {Canada}}

\date{\today}

\begin{abstract}
  Nonlinear stochastic differential equations (NSDEs) are a pillar of mathematical modeling for scientific and engineering applications. Accurate and efficient simulation of large-scale NSDEs is prohibitive on classical computers due to the large number of degrees of freedom, and it is challenging on quantum computers due to the linear and unitary nature of quantum mechanics. We develop a quantum algorithm to tackle nonlinear differential equations driven by the Ornstein--Uhlenbeck (OU) stochastic process. The query complexity of our algorithm scales logarithmically with the error tolerance and nearly quadratically with the simulation time.
Our algorithmic framework comprises probabilistic Carleman linearization (PCL) to tackle nonlinearity coupled with stochasticity, and stochastic linear combination of Hamiltonian simulations (SLCHS) to simulate stochastic non-unitary dynamics. We obtain probabilistic exponential convergence for the Carleman linearization of \citet{liu_efficient_2021}, provided the NSDE is stable and reaches a steady state. We extend deterministic LCHS to stochastic linear differential equations, retaining near-optimal parameter scaling from \citet{an_quantum_2025} except for the nearly quadratic time scaling. This is achieved by using Monte Carlo integration for time discretization of both the stochastic inhomogeneous term in LCHS and the truncated Dyson series for each Hamiltonian simulation.
\end{abstract}

\maketitle

\tableofcontents

\clearpage

\section{Introduction}\label{sec:introduction}

Nonlinear stochastic differential equations (NSDEs) provide a unifying mathematical framework for complex spatiotemporal phenomena in physics, engineering, finance, and the life sciences. Hallmark examples include the Langevin equation, Kardar--Parisi--Zhang (KPZ) equation for interface growth \cite{kardar_dynamic_1986}, stochastic Navier--Stokes models of turbulent flows \cite{flandoli_martingale_1995}, stochastic reaction--diffusion systems, and financial asset price dynamics via log-normal random walks \cite{black1973pricing, merton1973theory}. Their defining attributes---nonlinearity, randomness, and high-dimensional state spaces---produce multi-scale dynamics and intermittency that are notoriously difficult to approximate and to analyze. From a computational perspective, classical discretizations (finite elements, spectral methods, stochastic Galerkin or collocation, and sampling-based Monte Carlo) experience severe bottlenecks: the curse of dimensionality, stiffness and low regularity, and high variance in statistical estimators. As a result, 
the cost of resolving multiscale features in high-dimensional stochastic landscapes scales exponentially with system size, rendering direct simulation of realistic turbulent flows or high-frequency financial markets intractable \cite{kloeden_numerical_1977}. While significant progress has been made, existing algorithms remain insufficient for the precision and scale required to unlock new scientific frontiers.

Quantum computing offers a fundamentally new paradigm for tackling high-dimensional computational problems. The ability of quantum systems to exist in a superposition of states provides a natural framework for encoding high-dimensional vectors in a number of qubits that scales only logarithmically with the vector size, which has long promised to break the curse of dimensionality; however, identifying problems in this space that are feasible on a quantum computer but not on a randomized Turing machine remains an active research question. Building on this, quantum differential system solver subroutines---including quantum linear system solvers \cite{Harrow2009quantum,Childs2017quantum,Costa2022optimal}, quantum signal processing and singular value transformation (QSVT) \cite{gilyen_quantum_2019}, linear combination of Hamiltonian simulation (LCHS) \cite{an_linear_2023,an_quantum_2025,Low2025optimal}, time marching \cite{fang_time-marching_2023}, and Schrödingerisation \cite{Jin2024quantum,Jin2025on}---provide asymptotically improved scaling for structured linear operations. For expectation estimation, quantum amplitude estimation gives a quadratic speedup over classical Monte Carlo in sample complexity \cite{brassard_quantum_2000}, which is particularly compelling for SDEs where observables are statistical (e.g., moments, correlation functions, risk measures) and classical sampling dominates the runtime. This is particularly important for quantum differential-equation applications because solvers output a quantum state encoding of the solution $\ket{x(t)} = x(t)/\|x(t)\|$, which cannot be learned efficiently when the dimension of the Hilbert space is exponentially large. For this reason, prior resource estimates pertaining to fluid dynamics as a quantum application have involved estimating expectation values by encoding the drag force as an observable \cite{penuel_feasibility_2024}.

Recent advancements in quantum algorithms present efficient methods for solving nonlinear systems arising throughout science and engineering. \citet{liu_efficient_2021} proposed the first provably efficient quantum algorithm for dissipative nonlinear differential equations, which has been improved to achieve higher accuracy \cite{krovi2023improved,costa2025further} and extended to a broader regime \cite{liu_efficient_2023,wu2025quantum,jennings2025quantum}. Other linear representation approaches have also been explored, such as Koopman--von Neumann methods~\cite{joseph2020koopman,dodin2021applications,engel2021linear}, non-Hermitian Hamiltonian methods \cite{lloyd2020quantum}, homotopy methods \cite{xue2021quantum}, and level set methods \cite{jin2022quantum}.

However, the application of quantum algorithms to the full complexity of \emph{nonlinear stochastic} systems remains a largely unexplored frontier. Some previous quantum approaches have predominantly focused on linear stochastic equations \cite{jin_quantum_2025} or on randomized compilation implementations of linear dynamics \cite{yang2025circuit}. Recent work developed quantum algorithms for simulating the Kolmogorov or Fokker--Planck equations that govern the probability evolution of nonlinear SDEs \cite{bravyi2025quantum,gnanasekaran_efficient_2023}. Aside from the fact that these approaches do not directly address nonlinear dynamics, the quantum computer is used only to produce expectations of the solution rather than to store and describe the full state. Moreover, the classes of stochastic equations studied have certain limitations, such as divergence-free conditions. The crucial interplay between nonlinearity and stochasticity, which generates rich and complex behavior, still requires further investigation.

To overcome the dual challenges of nonlinearity and randomness, in this work, we develop a novel quantum algorithmic framework that integrates \emph{probabilistic Carleman linearization} (PCL) for nonlinear dynamics with \emph{Stochastic Linear Combination of Hamiltonian Simulations} (SLCHS) to efficiently simulate NSDEs, providing a pathway to scalable quantum simulations of complex stochastic systems. We have the following main contributions:
\begin{enumerate}
    \item We formulate a general class of NSDEs amenable to Carleman linearization. We extend the standard dissipative condition for Carleman linearization to a broad class of forced nonlinear systems. We further analyze the stability and convergence properties of the Carleman linearization in the presence of stochastic forcing.
    \item We extend the deterministic LCHS to stochastic systems retaining most optimal scaling properties under reasonable assumptions.
    \item Finally, we develop a quantum algorithm for NSDEs by combining our improved Carleman linearization for forced nonlinear systems with the stochastic LCHS algorithm.
\end{enumerate}

\subsection{Problem statement}
We focus on nonlinear systems described by the quadratic nonlinear ODE:
\begin{equation}
\label{eq:main_system_intro}
\frac{d\mathbf{x}}{dt} = \mathbf{F}_2 \mathbf{x}^{\otimes 2} + \mathbf{F}_1 \mathbf{x} + \mathbf{F}_0, \qquad \mathbf{x}(0) = \mathbf{x}_{\mathrm{in}}
\end{equation}
where the solution vector $\mathbf{x} = [x_1, \ldots, x_n]^T \in \mathbb{R}^n$ with each $x_j = x_j(t)$ being a function of $t$ on the interval $[0,T]$ for $j\in\range{n}\coloneqq\{1,\ldots,n\}$,
$\mathbf{x}^{\otimes 2} = [x_1^2, x_1x_2, \ldots, x_1x_n, x_2x_1, \ldots, x_nx_{n-1}, x_n^2]^T \in \mathbb{R}^{n^2}$, $\mathbf{F}_2 \in \mathbb{R}^{n \times n^2}$ is a bounded time-independent linear operator from $\mathbb{R}^{n^2}$ to $\mathbb{R}^n$, and $\mathbf{F}_1 \in \mathbb{R}^{n \times n}$ is a time-independent bounded linear operator such that there exists $\alpha>0$ for
$\mathbf{F}_1 \preceq - \alpha I$.  
The inhomogeneous term $\mathbf{F}_0(t) \in \mathbb{R}^n$ 
is a general $n$-dimensional Ornstein--Uhlenbeck (OU) process:
\begin{equation}
d\mathbf{F}_0(t) = -\boldsymbol{\Theta} \mathbf{F}_0(t) \, dt + \boldsymbol{\Sigma} \, d\mathbf{W}(t)
\label{eq:general_ou_intro}
\end{equation}
where
$\boldsymbol{\Theta} \in \mathbb{R}^{n \times n}$ is the mean-reversion drift matrix with $\text{Re}(\lambda_i(\boldsymbol{\Theta})) > 0$,
$\boldsymbol{\Sigma} \in \mathbb{R}^{n \times n}$ is the diffusion matrix,
and $\mathbf{W}(t) \in \mathbb{R}^n$ is $n$-dimensional Brownian motion with independent components. We consider the following quantum computational problem:
\begin{prob}\label{prob:snode}
In the \emph{quantum stochastic quadratic ODE problem}, we consider an $n$-dimensional quadratic system as in \cref{eq:main_system_intro} driven by the OU process in \cref{eq:general_ou_intro}. We assume $\mathbf{F}_2$, $\mathbf{F}_1$, $\boldsymbol{\Theta}$, and $\boldsymbol{\Sigma}$ are $s$-sparse (i.e., have at most $s$ nonzero entries in each row and column). We assume there exists $\alpha > 0$ such that $\mathbf{F}_1 \preceq -\alpha I$. We parametrize the problem in terms of the following Lyapunov $R$-number satisfying
\begin{equation}\label{eq:R-P_intro}
R_P \;:=\; \frac{1}{-\mu_P(\mathbf{F}_1)}\left(\|\mathbf{F}_2\|_P\,\|\mathbf{x}(0)\|_P \;+\; \frac{\|\mathbf{F}_0\|_P}{\|\mathbf{x}(0)\|_P}\right) < 1,
\end{equation}
where the operator norm induced by $P$ is defined as
$
    \| \mathbf{F}_2\|_P = \left\| P^{1/2} \mathbf{F}_2\left( P^{-1/2} \otimes P^{-1/2}\right)\right\|
  $ and $\mu_P(\mathbf{F}_1)$ is the generalized logarithmic norm.
  For some given $T>0$, we assume the values $\|\mathbf{F}_2\|$, $\|\mathbf{F}_1\|$, $\|\boldsymbol{\Theta}\|$, and $\|\boldsymbol{\Sigma}\|$ are known. We are given oracles $O_{F_2}$, $O_{F_1}$, $O_{\Theta}$, and $O_{\Sigma}$ that provide the locations and values of the nonzero entries of $\mathbf{F}_2$, $\mathbf{F}_1$, $\boldsymbol{\Theta}$, and $\boldsymbol{\Sigma}$, respectively, for any desired row or column. We are also given the value $\|\mathbf{x}_{\mathrm{in}}\|$ and an oracle $O_x$ that maps $|00\ldots0\rangle \in\mathbb{C}^n$ to a quantum state proportional to $\mathbf{x}_{\mathrm{in}}$.
  Our goal is to produce a quantum state proportional to the solution $\mathbf{x}(T)$ (or a state encoding its statistical moments) for some given $T>0$ within some prescribed error tolerance $\eps>0$ with probability at least $1-\delta$.
\end{prob}

\subsection{Main results}
We present a quantum algorithm for Problem~\ref{prob:snode} via PCL of \cref{eq:main_system_intro} and SLCHS.
\begin{result}[Informal version of~\cref{thm:deltax-tail}]\label{result:stoch-carleman-convergence}
  Consider the $n$-dimensional quadratic system driven by an OU process as defined in \cref{eq:main_system_intro,eq:general_ou_intro}. Suppose the system is Lyapunov stable, meaning there exists a matrix $P \succ 0$ such that the logarithmic norm satisfies $\mu_P(\mathbf{F}_1) < 0$.
Let $\mathbf{x}(t)$ be the true solution and let $\hat{\mathbf{y}}^{(N)}(t)$ denote the $N$-th order Carleman approximation. 
Then, the linearization error $\boldsymbol{\eta}^{(N)}(t) = \mathbf{x}(t) - \hat{\mathbf{y}}^{(N)}_1(t)$ is bounded probabilistically. Specifically, there exist constants $c, \Delta_0 > 0$ (dependent on time $T$, the drift/diffusion of $\mathbf{F}_0$, and the system nonlinearity $\|\mathbf{F}_2\|$) such that for sufficiently large error thresholds $\Delta \ge \Delta_0$, the probability of the error exceeding $\Delta$ decays as a Weibull-like tail:
\begin{equation}
\mathbb{P}\left(\|\boldsymbol{\eta}^{(N)}(t)\|_{P_N}^2 \ge \Delta \right) \;\le\; \exp\!\left(- c\, \Delta^{\frac{1}{N+1}}\right).
\end{equation}
Consequently, for a fixed success probability $1-\delta$, the error satisfies $\|\boldsymbol{\eta}^{(N)}(t)\| \le O\big( (\ln(1/\delta))^{N+1} \big)$, implying convergence as $N \to \infty$ provided $R_P < 1$ as defined in \cref{eq:R-P_intro}.
\end{result}
The formal version, assumptions, and the proof of Result~\ref{result:stoch-carleman-convergence} are provided in~\cref{lem:lyap-tail} and~\cref{thm:deltax-tail}. We extend the dissipative condition of Carleman linearization~\citep{liu_efficient_2023, jennings2025quantum} to a stochastically driven case and provide probabilistic truncation error bounds. This is achieved by combining the Carleman linearization convergence around a steady state of the deterministic nonlinear dynamics and Gaussianity of the stochastic forcing.

\begin{result}[Informal version of~\cref{thm:main}]\label{result:NSPDE} 
 Consider an instance of the
 quantum quadratic ODE problem as defined in Problem~\ref{prob:snode}, and assuming that a block encoding is chosen such that for all instances of the noise, the block encoding constant $\alpha_A$ of the Carleman matrix $\mathbf{A}_N$ truncated at order $N$ and its normalization constant $C_\alpha$ satisfies  $\alpha_A=C_\alpha \sup_t\|\mathbf{A}_N\|$,
we have that there exists an algorithm that prepares an $\eps$-approximation of the normalized state $\ket{U(T)}$ with probability at least $1-\delta$ that uses
\begin{align*}
N_Q=\widetilde{\mathcal{O}}\!\left(
\frac{\|U_{\mathrm{in}}\|+\,\sqrt{\ \frac{T \|\boldsymbol{\Sigma}\|_{F}^{2}}{2\,\lambda_{\min}\delta}\,\left(\,T\ -\ \frac{1-e^{-2\lambda_{\min} T}}{2\,\lambda_{\min}}\,\right)}}{\|U(T)\|}
\;N C_{\alpha}\left(\|\mathbf{F}_1\|+\|\mathbf{F}_2\|+\sqrt{\frac{1-e^{-2\lambda_{\min} T}}{2\lambda_{\min}\delta}\,\|\boldsymbol{\Sigma}\|_{F}^{2}}\right)\;
T\;\left(\log\frac{1}{\eps}\right)^{1+1/\beta}
\right)
\end{align*}
queries to input models. 
Here $\lambda_{\min}$ is the smallest eigenvalue of the drift matrix $\boldsymbol{\Theta}$ of the OU process defined in \cref{eq:general_ou_intro} and $N_k \in O\left(\frac{\mathcal{K}^2(T{-}s)^{2k} }{ (k!)^2 }\ \frac{\sup_q\left(\mathbb{E}_\omega\left({\sup_{\tau\in[s,T]}\ \|H(\tau,\omega)\|^q/\delta}\right)\right)^{2k/q} }{ \delta\ \eps^2 }\right)$ is the Monte Carlo integration size for Dyson series truncation $\mathcal{K} \in \mathcal{O}(\log(1/\eps))$.

\end{result}
The formal version, assumptions, and the proof of Result~\ref{result:NSPDE} are provided in~\cref{thm:main}. Our quantum algorithm for NSDE retains the same near-optimal parameter scaling of simulating linear ODEs as \citet{an_quantum_2025} except for the quadratic scaling in time due to the growth of standard deviation of the stochastic process. The stochasticity in both the Carleman matrix and inhomogeneous term requires a probabilistic convergence of the inhomogeneous term of the LCHS and a Hamiltonian simulation subroutine for stochastic dynamics, which is achieved by using Monte Carlo integration for the time discretization.  

In summary, we provide a quantum algorithm for solving nonlinear stochastic differential equations driven by an OU process. A major challenge arises from the fact that such processes lead to forcing terms that are continuous but nowhere differentiable. These problems are compounded for nonlinear differential equations because the forcing terms become incorporated into the linearized generator of the dynamics used in quantum algorithms.

Our work addresses these issues by providing a new error analysis for Carleman linearization in the presence of stochastic forcing, as well as a new error analysis for truncated Dyson-series algorithms that accounts for a probabilistic distribution over simulation parameters. We find that such processes can be simulated using a number of operations that scales logarithmically with the error tolerance, but nearly quadratically with the simulation time. This deviates from the linear-time limit achievable for Hamiltonian simulation because the noise causes the standard deviation of the forcing term to grow with time, which in turn degrades the time scaling.

There are a number of open questions that remain. The first is whether these ideas can be generalized to address Wiener processes rather than the OU processes. This is challenging because the corresponding forcing function is unbounded. Understanding whether such processes can be simulated in polynomial time is an important step forward.

Further, understanding whether our algorithm is near-optimal is important in cases with non-trivial stochastic forcing. This in turn begs the question of whether examples of a true exponential advantage over classical methods can be identified. Identifying such instances would clarify regimes where the noise introduced by the stochastic differential equations is insufficient to cause the distribution to become essentially classical. In turn, understanding this point will, we believe, bring us closer to the ultimate question of whether concrete examples of exponential speed-up can be found for stochastic differential equations and whether killer applications remain to be found within this domain.

\subsection{Organization}

The rest of the paper is organized as follows:
In \cref{sec:carleman-review} we review Carleman linearization for quadratic nonlinear differential equations and summarize the deterministic truncation error bounds that underlie quantum nonlinear ODE solvers.
In \cref{sec:sde-review} we recall background on stochastic differential equations and the OU process used throughout the paper.
In \cref{sec:ou-norm-bounds} we derive high-probability bounds for the norm of a multivariate OU process that are later used to control random simulation parameters.
In \cref{sec:pcl} we develop probabilistic Carleman linearization for NSDEs, including the steady-state regime and probabilistic error bounds for the Carleman truncation.
In \cref{sec:slchs} we introduce the SLCHS framework and develop the associated probabilistic stability and time-discretization analyses needed to handle nowhere-differentiable forcing.
In \cref{sec:quantum-impl} we give the quantum implementation details and complexity analysis and assemble these ingredients to prove the main algorithmic guarantees.
In \cref{sec:conclusion} we summarize our results and discuss further improvements.

\section{Review of Carleman linearization for nonlinear differential equations}\label{sec:carleman-review}

Our aim in this paper is to provide a quantum algorithm that outputs an approximate solution to a system of nonlinear differential equations. This task is surprisingly challenging: quantum computers excel at enacting unitary transformations, which are linear by definition, and even though we have found ways of embedding general linear (i.e., non-unitary) transformations into unitary ones, it takes exponentially large resources to embed a nonlinear transformation into a linear one without approximations~\cite{jin2022quantum}. The main concept we use to tackle this challenge is Carleman linearization, which rewrites a nonlinear differential equation as an infinite set of coupled linear differential equations, after which one performs a moderate polynomial-size truncation.

Our aim in this and the next sections is to provide a technical review of methods used to linearize differential equations and to give a basic introduction to stochastic differential equations. This latter point is particularly significant, as the mathematics of stochastic differential equations is much more subtle than that of ordinary differential equations.

A nonlinear differential equation, at a high level, is a system of equations that relates the time derivative of a multivariate function to the values of the function and its derivatives. It is called nonlinear if the equation contains a power of the function, or its derivatives, that is higher than linear. Without loss of generality, we can always introduce new variables to eliminate second or higher derivatives from a differential equation. Similarly, we can introduce new variables to ensure that a differential equation has only quadratic nonlinearities. For example, if we had a differential equation of the form
\begin{equation}
    \frac{dx}{dt} = ax+bx^3,
\end{equation}
we could introduce the variable $w=x^2$ to give us the system 
\begin{align}
    \frac{dx}{dt} &= ax+bxw\\
    \frac{dw}{dt} &= 2x \frac{dx}{dt} = 2aw +2bw^2,
\end{align}
which has only quadratic nonlinearities. In this work we focus on real-valued differential equations rather than complex differential equations to connect better with standard results from the stochastic differential equations literature, but in general these results can also be used to describe the complex case by representing each complex number as a tuple of its real and imaginary components.
Given these simplifying assumptions, which hold without loss of generality, we can describe an arbitrary quadratic nonlinear differential equation below.
\begin{definition}\label[definition]{def:qdns}
We define a quadratic nonlinear first order system of differential equations for a function $x:\mathbb{R}\rightarrow \mathbb{R}^n$ as the solution vector to the quadratic nonlinear ODE system:
\begin{equation}
\label{eq:main_system}
\frac{d\mathbf{x}}{dt} = \mathbf{F}_2 \mathbf{x}^{\otimes 2} + \mathbf{F}_1 \mathbf{x} + \mathbf{F}_0, \qquad \mathbf{x}(0) = \mathbf{x}_{\mathrm{in}}
\end{equation}
where $\mathbf{x} = [x_1, \ldots, x_n]^T \in \mathbb{R}^n$ with each $x_j = x_j(t)$ being a function of $t$ on the interval $[0,T]$ for $j\in\range{n}\coloneqq\{1,\ldots,n\}$,
$\mathbf{x}^{\otimes 2} = [x_1^2, x_1x_2, \ldots, x_1x_n, x_2x_1, \ldots, x_nx_{n-1}, x_n^2]^T \in \mathbb{R}^{n^2}$, $\mathbf{F}_2 \in \mathbb{R}^{n \times n^2}$ is a bounded time-independent linear operator from $\mathbb{R}^{n^2}$ to $\mathbb{R}^n$, $\mathbf{F}_1 \in \mathbb{R}^{n \times n}$ is a time-independent bounded linear operator such that there exists $\alpha>0$ for
$\mathbf{F}_1 \preceq - \alpha I$  
and $\mathbf{F}_0(t) \in \mathbb{R}^n$ 
is either a $C^1$ continuous function of $t$ or a stochastic process (e.g., the OU process in \cref{eq:general_ou}). \end{definition}
\noindent
Carleman linearization then gives a method for approximating the equation of motion of the quadratic first-order system of differential equations as a much larger system of linear differential equations. We describe this below for reference. We use $\| \cdot \|$ for the Euclidean norm on vectors and the corresponding operator norm on matrices, unless otherwise stated.

\begin{definition}[Carleman linearization]
  \label[definition]{pre:cl}
Given a system of quadratic ODEs (\cref{eq:main_system}), the Carleman procedure can be applied to obtain the system of linear ODEs
\begin{equation}
  \frac{\d{\mathbf{y}}}{\d{t}} = \mathbf{A}\, \mathbf{y} + \mathbf{F_0}(t), \qquad
  \mathbf{y}(0) = \mathbf{y}_{\mathrm{in}}
\label{eq:LODE}
\end{equation}
where $\mathbf{A}$ has the tri-diagonal block structure
\begin{equation}
\frac{\d{}}{\d{t}}
  \begin{pmatrix}
    \mathbf{y}_1 \\
    \mathbf{y}_2 \\
    \mathbf{y}_3 \\
    \vdots \\
    \mathbf{y}_{N-1} \\
    \mathbf{y}_N \\
  \end{pmatrix}
=
  \begin{pmatrix}
     \mathbf{A}_1^1 & \mathbf{A}_2^1 &  &  &  &  \\
     \mathbf{A}_1^2 & \mathbf{A}_2^2 & \mathbf{A}_3^2 & &  &  \\
      & \mathbf{A}_2^3 & \mathbf{A}_3^3 & \mathbf{A}_4^3 &  &  \\
     &  & \ddots & \ddots & \ddots &  \\
      &  &  & \mathbf{A}_{N-2}^{N-1} & \mathbf{A}_{N-1}^{N-1} & \mathbf{A}_N^{N-1} \\
      &  &  &  & \mathbf{A}_{N-1}^N & \mathbf{A}_N^N \\
  \end{pmatrix}
  \begin{pmatrix}
    \mathbf{y}_1 \\
    \mathbf{y}_2 \\
    \mathbf{y}_3 \\
    \vdots \\
    \mathbf{y}_{N-1} \\
    \mathbf{y}_N \\
  \end{pmatrix}+
  \begin{pmatrix}
    \mathbf{F}_0(t) \\
    0 \\
    0 \\
    \vdots \\
    0 \\
    0 \\
  \end{pmatrix},
\label{eq:UODE}
\end{equation}
where $\mathbf{y}_j=\mathbf{x}^{\otimes j}\in\R^{n^j}$, $\mathbf{y}_{\mathrm{in}}=[\mathbf{x}_{\mathrm{in}}; \mathbf{x}_{\mathrm{in}}^{\otimes 2}; \ldots; \mathbf{x}_{\mathrm{in}}^{\otimes N}]$, and $\mathbf{A}_{j+1}^j \in \R^{n^j\times n^{j+1}}$, $\mathbf{A}_j^j \in \R^{n^j\times n^j}$, $\mathbf{A}_{j-1}^j \in \R^{n^j\times n^{j-1}}$ for $j\in\range{N}$ satisfying
\begin{align}
\mathbf{A}_{j+1}^j &= \mathbf{F}_2\otimes I^{\otimes j-1}+I\otimes \mathbf{F}_2\otimes I^{\otimes j-2}+\cdots+I^{\otimes j-1}\otimes \mathbf{F}_2, \label{eq:tensor2} \\
\mathbf{A}_j^j &= \mathbf{F}_1\otimes I^{\otimes j-1}+I\otimes \mathbf{F}_1\otimes I^{\otimes j-2}+\cdots+I^{\otimes j-1}\otimes \mathbf{F}_1, \label{eq:tensor1} \\
\mathbf{A}_{j-1}^j &= \mathbf{F}_0(t)\otimes I^{\otimes j-1}+I\otimes \mathbf{F}_0(t)\otimes I^{\otimes j-2}+\cdots+I^{\otimes j-1}\otimes \mathbf{F}_0(t). \label{eq:tensor0}
\end{align}
Note that $\mathbf{A}$ is a $(3Ns)$-sparse matrix. The dimension of \cref{eq:UODE} is
\begin{equation}
  d_N \coloneqq n+n^2+\cdots+n^N=\frac{n^{N+1}-n}{n-1}=\mathcal{O}(n^N).
\end{equation}
\end{definition}

By truncating at finite order $N$, due to the tri-diagonal structure, we effectively omit the matrix $\mathbf{A}_{N+1}^N$ from the final row. We define the residual term $\mathbf{R}_N(t) \in \mathbb{R}^{d_N}$ by placing $\mathbf{A}_{N+1}^N \mathbf{x}^{\otimes (N+1)}$ in the last Carleman block and zeros elsewhere, i.e.,
\begin{equation}
    \mathbf{R}_N(t) \;=\; \begin{bmatrix} 0 \\ \vdots \\ 0 \\ \mathbf{A}_{N+1}^N\, \mathbf{x}(t)^{\otimes (N+1)} \end{bmatrix},
\end{equation}
from which the following holds:
\[
  \|\mathbf{R}_N(t)\| = \|\mathbf{A}_{N+1}^N \mathbf{x}^{\otimes (N+1)}\|.
\]
We provide a bound in the truncated Carleman linearization as follows. This bound only depends on the norm of the solution to the original nonlinear ODE system \cref{eq:main_system} and the structure of the Carleman transform and is independent of the stability of the nonlinear system. It will be used in Carleman linearization error analysis in \cref{lemma:carleman-delta-f} and \cref{thm:error_bound}. 

\begin{lemma}[Bound on the Carleman Residual Term]\label[lemma]{lemma:RNt}
For the residual term $\mathbf{R}_N(t)$ in the truncated Carleman linearization, the following bound holds:
\begin{equation}\label{eq:RNt}
\|\mathbf{R}_N(t)\| = \|\mathbf{A}_{N+1}^N \mathbf{x}^{\otimes(N+1)}\| \leq N \|\mathbf{F}_2\| \|\mathbf{x}(t)\|^{N+1}
\end{equation}
\end{lemma}

\begin{proof}
  The $N$th-order truncated system \cref{eq:UODE} with explicit truncation error vector $\mathbf{R}_N(t) \in \mathbb{R}^{d_N}$ arising from neglecting tensor powers higher than $N$ can be written as:
\begin{equation}\label{eq:truncated_system}
\frac{d\mathbf{y}_N}{dt} = \mathbf{A}_N\mathbf{y}_N + \mathbf{b}_N + \mathbf{R}_N(t),
\end{equation}
where $\mathbf{A}_N \in \mathbb{R}^{d_N \times d_N}$ is the truncated Carleman coefficient matrix and $\mathbf{b}_N \in \mathbb{R}^{d_N}$ is the inhomogeneous term in the Carleman space defined in \cref{eq:UODE}. 

From the structure of the operator $\mathbf{A}_{N+1}^N$:
\begin{equation}
\mathbf{A}_{N+1}^N = \mathbf{F}_2\otimes I^{\otimes N-1}+I\otimes \mathbf{F}_2\otimes I^{\otimes N-2}+\cdots+I^{\otimes N-1}\otimes \mathbf{F}_2
\end{equation}
This is a sum of $N$ terms, each containing $\mathbf{F}_2$ in a different position in the tensor product. 
For any tensor product operator $A \otimes B$, the operator norm satisfies $\|A \otimes B\| = \|A\| \cdot \|B\|$. Since $\|I\| = 1$, each term in the sum has norm $\|\mathbf{F}_2\|$:
\begin{equation}
\|\mathbf{F}_2\otimes I^{\otimes N-1}\| = \|\mathbf{F}_2\| \cdot \|I\|^{N-1} = \|\mathbf{F}_2\|
\end{equation}
By the triangle inequality, with $N$ terms in the sum:
\begin{equation}
\|\mathbf{A}_{N+1}^N\| \leq \sum_{i=1}^N \|\mathbf{F}_2\otimes I^{\otimes i-1} \otimes I^{\otimes N-i}\| = N\|\mathbf{F}_2\|
\end{equation}
For the tensor power, we have $\|\mathbf{x}^{\otimes(N+1)}\| = \|\mathbf{x}\|^{N+1}$. Therefore:
\begin{equation}
\|\mathbf{R}_N(t)\| = \|\mathbf{A}_{N+1}^N \mathbf{x}^{\otimes(N+1)}\| \leq \|\mathbf{A}_{N+1}^N\| \cdot \|\mathbf{x}^{\otimes(N+1)}\| \leq N \|\mathbf{F}_2\| \cdot \|\mathbf{x}\|^{N+1}
\end{equation}
\end{proof}

This error bound is the instantaneous version of \citet[Lemma 1]{liu_efficient_2021}.
According to \cref{eq:LODE} and \cref{eq:truncated_system}, the dynamics of the cumulative Carleman linearization error vector $\boldsymbol{\eta}^{(N)}(t)$ at order $N$ satisfies:
\begin{equation}\label{eq:deta-dt}
\frac{d \boldsymbol{\eta}^{(N)}}{dt} = \mathbf{A}_N \, \boldsymbol{\eta}^{(N)} + \mathbf{R}_N(t), \qquad \boldsymbol{\eta}^{(N)}(0) = 0.
\end{equation}
Applying variation-of-constants formula to the \textit{deterministic} \cref{eq:deta-dt}, we have:  
\begin{equation}\label{eq:eta-integral}
\boldsymbol{\eta}^{(N)}(t) = \int_0^t \Phi(t, s) \mathbf{R}_N(s) ds,
\end{equation}
where $\Phi(t,s)$ is the fundamental solution matrix.

\section{Review of stochastic differential equations}\label{sec:sde-review}

In this section, we give a short overview of Stochastic Differential Equations (SDEs). For additional details, we refer the reader to the standard references in this area \cite{oksendal_stochastic_2003}, \cite{Karatzas1998} and \cite{kloeden1992numerical}. SDEs are a natural generalization of ordinary differential equations to the case where one or more coefficients or variables is a stochastic process. Most SDEs of note (those not featuring discontinuous jumps) typically include a \emph{Wiener process} stochastic term. The formal description of the Wiener process can be given as follows:

\begin{definition} \label[definition]{def:wiener_process}
    Let $(\Omega,\mathcal{F}, \mathbb{P})$ be a probability space and let $t>0$. For any particular realization of the stochastic process $\omega \in \Omega$ and partition of the interval $[0,t]$ such that $0=t_0<t_1\dots<t_m=t$, the function $W: \Omega\times \mathbb{R}^+\rightarrow\mathbb{R}$ is a Wiener process if it satisfies the following criteria:
    \begin{enumerate}
        \item $W(\omega,0)=0$ almost surely, without loss of generality.
        \item The non-overlapping increments $W(\omega, t_{i+1})-W(\omega, t_{i})$ and $W(\omega, t_{j+1})-W(\omega, t_{j})$ are independent of each other for $i \neq j$, and $\forall \: i, j < m$. 
        \item $W(\omega, t_{i+1})-W(\omega, t_{i})$ is normally distributed as $\mathcal{N} (0, t_{i+1} - t_{i}) $, that is $\mathbb{E}_\omega[W(\omega, t_{i+1})-W(\omega, t_{i})] = 0$ and $\text{var}_\omega[W(\omega, t_{i+1})-W(\omega, t_{i})]=t_{i+1}-t_i$ for all $0\leq i<m$.
        \item $W_{t}$ is almost surely continuous in $t$.
    \end{enumerate}
    for all $m>0$ and partition $\{t_i\}_{i=0}^m$. An m-dimensional Wiener process $\mathbf{W}(t)$ is defined as the vector of m independent Wiener processes $\mathbf{W}(t):=[W^{(1)}(t),\dots,W^{(m)}(t)]$.
\end{definition}
The $\omega$ label denotes a particular realization of the Wiener process over the interval of interest. Since, for a given closed interval $[0,t]$, each unique realization of $W(t)$ can be bijectively mapped to $\omega$, in notation we will usually suppress the $\omega$ argument and make the underlying probability space implicit. To facilitate our definition of filtrations over a probability space, we first define a Borel $\sigma$-algebra, an essential construction necessary for us to meaningfully assign a probability measure (in the Lebesgue sense) to probability spaces of interest.
\begin{definition} \label[definition]{def:BorelSigmaAlgebra}
  Denoting $\mathcal{P}(X)$ as the power set for some set $X$, the set $\Sigma \subset \mathcal{P}(X) $ is a Borel $\sigma$-algebra if the following holds:
  \begin{enumerate}
      \item $X \in \Sigma$.
      \item If a set $A$ is in $\Sigma$, then its complement $X \setminus A$ is also in $\Sigma$.
      \item $\Sigma$ is closed under countable unions, that is $\bigcup_{i=1}^{\infty} A_{i} \in \Sigma$.
  \end{enumerate}
\end{definition}
The Borel $\sigma$-algebra essentially defines the set of possible events that can occur in our probability space. 

\begin{definition}
  Consider a set $X$ and a $\sigma$-algebra $\mathcal{F}$ on $X$. The tuple $(X, \mathcal{F})$ is called a measurable space, and the elements of $\mathcal{F}$ are measurable sets within this measurable space.
\end{definition}

These notions extend directly to stochastic processes. The Wiener process also comes equipped with a \textit{filtration} $\mathcal{F}_{s}$, which contains all the information available to us at any time $s \in [0,t]$ regarding the set of possible events that can occur. Formally, we define this as follows:
\begin{definition}\label[definition]{def:filtration}
    A filtration for the Wiener process is a sequence of $\sigma$-algebras $\mathcal{F}(s)$, $0\leq s\leq t$ with the following properties:
    \begin{enumerate}
        \item $\mathcal{F}(s_1) \subseteq \mathcal{F}(s_2)$ if $s_1 \leq s_2$. 
        \item $W(s)$ is $\mathcal{F}(s)$ adapted.
        \item $W(s_2)-W(s_1)$ is independent of $\mathcal{F}(s_1)$. 
    \end{enumerate}
\end{definition}
Filtrations can be interpreted as finite-time-interval-limited subsets of the associated Borel $\sigma$-algebra. The purpose of this somewhat cryptic construction for a stochastic process is to express the amount of knowable information for our stochastic process at any particular point in time as it relates to the set(s) of permissible stochastic trajectories under a $\sigma$-algebra. When a stochastic process $X(t)$ is said to be \textit{adapted} to a filtration $\mathcal{F}(t)$, the possibility of each realization and in fact the particular realized trajectory at time $t$ can be determined from the information in $\mathcal{F}(t)$ alone. For example, the integral process $X(t):=\int_0^t W(s)ds$ is adapted to the filtration $\mathcal{F}(t)$ because its value has no randomness once the path $W(t)$ took until time $t$ is known, but the process $X(t):= W(t) + W(\tau)$ is not adapted to $\mathcal{F}(t)$ for some $\tau > t$ since it uses information from the future that is not known with certainty at time $t$.

\begin{definition} \label[definition]{def:MeasurableFunction}
     A measurable function is a function between the underlying sets of two measurable spaces that preserves the structure of the spaces: the preimage of any measurable set is measurable.
\end{definition}

The last ingredient necessary for our purposes is the It\'o integral, which defines integration with respect to a measure with respect to a Wiener process (and any other process that is adapted to the filtration of a Wiener process).

\begin{definition}\label[definition]{def:Wiener}
    For Wiener process $W(t)$ with filtration $\mathcal{F}(t)$, let $\Delta(t)$ be another bounded (i.e., $\EV{\int_0^t \Delta^2(s)ds}<\infty$) stochastic process that is adapted to the filtration. For a sequence of partitions $\{\pi_n\}$ of the interval $[0,t]$ such that $t_0:=0$, $t_n:=t$, and a maximum step $\max_i (t_{i+1} - t_i)$ that goes to zero as $n$ goes to infinity, the so-called It\'o integral is defined as
    \begin{equation}
        \int_0^t \Delta(s)dW_s := \lim_{n\rightarrow\infty} \sum_{i=0}^{n-1}\Delta(t_i)(W(t_{i+1})-W(t_{i})),
    \end{equation}
    and the associated differential equation should be understood as a notational shorthand for the corresponding It\'o integral equation:
    \begin{equation}
        dI_t = \Delta_t dW_t\qquad\qquad\iff\qquad\qquad I_t = \int_0^t \Delta(s)dW_s \nonumber.
    \end{equation}
    It should also be noted that $I_t = \int_0^t \Delta(s)dW_s$ itself is a stochastic process.
\end{definition}

The key point to note in the definition of the It\'o integral is that the integrand $\Delta(t_i)$ in its Riemann sum decomposition is evaluated at the left time endpoint $t_i$ for each term in the sum. This particular choice selects a specific interpretation of the dynamics of our SDE rooted in a non-anticipatory nature of the integrand or differential equation coefficients. While seemingly innocuous, this assumption has implications when coefficient terms include stochasticity. Armed with these definitions, we can then define the stochastic differential equations we will focus on.
\begin{definition}
\label[definition]{ass:time_dependent}
A linear vector stochastic differential equation with additive Wiener noise is a multivariate SDE of the form
\begin{equation}
d\mathbf{X}(t) = \mathbf{A}(t)\mathbf{X}(t)dt + \mathbf{B}(t)d\mathbf{W}(t), \quad t \in [t_0, T], \quad \mathbf{X}(t_0) = \mathbf{X}_0
\label{eq:time_dependent_sde}
\end{equation}
or, in the equivalent integral form as
\begin{equation}
\mathbf{X}(t)
= \mathbf{X}_0
+ \int_{t_0}^t \mathbf{A}(s)\,\mathbf{X}(s)\,ds
+ \int_{t_0}^t \mathbf{B}(s)\,d\mathbf{W}(s),
\label{eq:integral_sde}
\end{equation}
where:
\begin{enumerate}
\item $\mathbf{A}:[t_0,T]\to\mathbb{R}^{d\times d}$ is measurable and locally integrable, equivalently
\(
\int_{t_0}^T \|\mathbf{A}(s)\|\,ds<\infty,
\)
with $\|\cdot\|$ the operator norm on $\mathbb{R}^{d\times d}$.
\item $\mathbf{B}:\Omega\times[t_0,T]\to\mathbb{R}^{d\times m}$ is progressively measurable (with respect to $(\mathcal{F}_t)$) and square-integrable:
\(
\int_{t_0}^T \|\mathbf{B}(s)\|^2\,ds<\infty
\)
almost surely.
\item $\mathbf{X}_0\in L^2(\Omega,\mathcal{F}_{t_0};\mathbb{R}^d)$ is the initial condition satisfying
$\mathbb{E}[\|\mathbf{X}_0\|^2] < \infty$
\item $\mathbf{W}(t) \in \mathbb{R}^m$ is $m$-dimensional Brownian motion defined as in \cref{def:wiener_process}.
\end{enumerate}
\end{definition}

In this work, we are particularly interested in OU systems, which are a class of stochastic differential equations in which the forcing term is given by the integral of Brownian noise alongside a drift parameter $\Theta$, which, when suitably constrained, ensures the dynamics is mean-reverting. This mean-reversion property implies that in the long-time limit $t\to\infty$, the noise becomes concentrated about its mean value (which, without loss of generality, we take here to be zero). The formal definition of the process is given below. This mean-reverting property allows the OU process to have a steady-state distribution, in contrast to general Wiener diffusion processes with multiplicative noise, which do not reach a steady state in general.
\begin{definition}[Ornstein--Uhlenbeck Driven Systems]
\label[definition]{def:OU-LPDE}
Consider the $n$-dimensional linear stochastic system:
\begin{equation}
\frac{d\mathbf{x}}{dt} = \mathbf{F}_1 \mathbf{x} + \mathbf{F}_0(t)
\label{eq:general_system}
\end{equation}
where
$\mathbf{x}(t) \in \mathbb{R}^n$ is the state vector and
$\mathbf{F}_1 \in \mathbb{R}^{n \times n}$ is an arbitrary stable matrix with $\text{Re}(\lambda_i(\mathbf{F}_1)) < 0$ for all $i$.
The driving term $\mathbf{F}_0(t) \in \mathbb{R}^n$ follows the general $n$-dimensional OU process:
\begin{equation}
d\mathbf{F}_0(t) = -\boldsymbol{\Theta} \mathbf{F}_0(t) \, dt + \boldsymbol{\Sigma} \, d\mathbf{W}(t)
\label{eq:general_ou}
\end{equation}
where
$\boldsymbol{\Theta} \in \mathbb{R}^{n \times n}$ is the mean-reversion drift matrix with $\text{Re}(\lambda_i(\boldsymbol{\Theta})) > 0$,
$\boldsymbol{\Sigma} \in \mathbb{R}^{n \times n}$ is the diffusion matrix,
and $\mathbf{W}(t) \in \mathbb{R}^n$ is $n$-dimensional Brownian motion with independent components defined as in \cref{def:wiener_process}.
\end{definition}

\section{Probabilistic bounds for the norm of a multivariate Ornstein--Uhlenbeck process}\label{sec:ou-norm-bounds}

In this section, we derive probabilistic bounds for the norm of a multivariate OU process that will be used in \cref{eq:VE}. Such bounds are an important ingredient in controlling both the linearization error and the success probability of our algorithm. Heuristically, the stochastically driven multivariate OU process in \cref{eq:general_ou} incorporates Wiener increments that are, in principle, unbounded. By bounding the norm of this multivariate OU process probabilistically, we can derive bounds on both the linearization-error probability and the algorithm's success probability. Suppose $\mathbf{F}_{0}(t)\in \mathbb{R}^{n}$ follows a general $n$-dimensional OU process as defined in \cref{eq:general_ou}. Assuming a time-independent drift matrix $\boldsymbol{\Theta}$ and using the integrating factor $e^{\boldsymbol{\Theta}t}$, we can directly obtain an expression for the $n$-dimensional OU process

\begin{eqnarray*}
  e^{\boldsymbol{\Theta}t} \, d\mathbf{F}_{0}(t) &=& - \boldsymbol{\Theta} e^{\boldsymbol{\Theta}t} \mathbf{F}_{0}(t) \, dt + e^{\boldsymbol{\Theta}t} \boldsymbol{\Sigma} \, d\mathbf{W}(t) \\
  d\bigl(e^{\boldsymbol{\Theta}t} \mathbf{F}_{0}(t)\bigr)  &=& e^{\boldsymbol{\Theta}t} \boldsymbol{\Sigma} \, d\mathbf{W}(t)
\end{eqnarray*}
which readily integrates to
\begin{equation}
  \mathbf{F}_{0}(t) = e^{-\boldsymbol{\Theta}t} \mathbf{F}_{0}(0) + e^{-\boldsymbol{\Theta}t} \int_{0}^{t} e^{\boldsymbol{\Theta}s} \boldsymbol{\Sigma} \, d\mathbf{W}(s). \label{eq:MOUsemiclosed}
\end{equation}
From \cref{eq:MOUsemiclosed}, we can immediately conclude that the $n$-dimensional OU process $\mathbf{F}_{0}(t)$ follows an $n$-dimensional normal distribution with mean $\mathbf{M}(t) = e^{-\boldsymbol{\Theta} t} \mathbf{F}_{0}(0)$ and covariance matrix $\mathrm{Cov}(\mathbf{F}(t)) = \int_{0}^{t}  e^{\boldsymbol{\Theta}(s-t)} \boldsymbol{\Sigma} \boldsymbol{\Sigma}^{T} e^{\boldsymbol{\Theta}^{T}(s-t)} ds$. Setting $\mathbf{F}_{0}(0) = \mathbf{0}$ we see that the covariance matrix coincides with the variance matrix of $\mathbf{F}_{0}(t)$. From Chebyshev's inequality for multivariate vectors, we have that
\begin{equation}
  P(\| \mathbf{F}_{0}(t) \| \geq k\, \sqrt{\mathrm{Cov}(\mathrm{Tr}(\mathbf{F}(t)))} \,\:) \leq \frac{1}{k^2} \label{eq:Chebyshev1}
\end{equation}
which in turn implies that
\begin{equation}
  P(\| \mathbf{F}_{0}(t) \| < k \sqrt{\mathrm{Tr}(\mathrm{Cov}(\mathbf{F}(t)))} \:) \geq 1- \frac{1}{k^2}
\end{equation}
where $\| \cdot \|$ denotes the Euclidean norm.
To obtain our desired bound over the entire interval $t \in [0,T]$, we note that our initial Chebyshev inequality \cref{eq:Chebyshev1} gives the probability bound for the norm of the OU process exceeding the threshold $k \sqrt{\mathrm{Tr}(\mathrm{Cov}(\mathbf{F}(t)))}$ at some time $t$. To compute the relevant bound, we interpret $t$ as a random time sampled from a probability distribution over $[0,T]$ with uniform probability density $\rho(t) = \frac{1}{T}$. Therefore, using the law of total probability, we obtain that the probability of a large deviation $\| \mathbf{F}_{0}(t) \|$ over the interval $t \in [0,T]$ satisfies
\begin{align}
 P_{s \in [0,T], \Omega} \left( \| \mathbf{F}_{0}(s) \| \geq k \sqrt{\mathrm{Tr}(\mathrm{Cov}(\mathbf{F}(s)))} \right ) = & \int_0^T \mathbb{E}_{\Omega(s)} \mathbf{1_{\| \mathbf{F}_{0}(s) \| \geq k \sqrt{\mathrm{Tr}(\mathrm{Cov}(\mathbf{F}(s)))}}}\frac{1}{T} ds\\
 =&\int_{0}^{T} P_{\Omega(s)}(\| \mathbf{F}_{0}(s) \| \geq k \sqrt{\mathrm{Tr}(\mathrm{Cov}(\mathbf{F}(s)))} \,\:) \frac{1}{T} ds \leq \frac{1}{k^2}
\end{align}
where $\mathbf{1}$ is taken to be an indicator function and where we have used~\cref{eq:Chebyshev1} to bound the probability over $\Omega$. Rewriting this in terms of some constant threshold $x_{*} > 0 $, we have
\begin{equation}
 P_{s \in [0,T], \Omega}\left(  \| \mathbf{F}_{0}(s) \| \geq x_{*}\right)  \leq \int_{0}^{T} P(\| \mathbf{F}_{0}(s) \| \geq x_{*} \,\:) \frac{1}{T} ds \leq \frac{\int_{0}^{T} \mathrm{Tr}(\mathrm{Cov}(\mathbf{F}(s))) ds}{T x_{*}^2}
\end{equation}
and consequently,
\begin{equation} 
 P_{s \in [0,T]}\left( \| \mathbf{F}_{0}(s) \| < x_{*}\right) \geq 1-\frac{\int_{0}^{T} \mathrm{Tr}(\mathrm{Cov}(\mathbf{F}(s))) ds}{T x_{*}^2}. \label{eq:Prob-OUnormThreshold}
\end{equation}
We also note that $\mathrm{Tr}(\mathrm{Cov}(\mathbf{F}(t)))$ is the trace of the covariance matrix $\mathrm{Cov}(\mathbf{F}(t))$ and has the convenient interpretation as the sum of the variances of each entry of $\mathbf{F}_{0}(t)$ since the covariance and variance matrices are identical when we set $\mathbf{F}_{0}(0) = \mathbf{0}$.

Another lower bound used extensively in our probabilistic Carleman linearization and SLCHS derivations can be derived via application of the Ito isometry to the semi-closed form expression for the $n$-dimensional OU process \cref{eq:MOUsemiclosed}. As a necessary intermediate step, we first prove a multivariate version of the Ito-isometry that will be useful in our subsequent derivations.
\begin{lemma} [Multivariate Ito isometry]
    Let $\mathbf{M}_{t}: [0, T] \times \Omega \xrightarrow{} \mathbb{R}^{n \times n}$ be a matrix-valued stochastic process adapted to the natural filtration $\mathcal{F}_{t}$ of a multivariate Brownian motion. Then
    \begin{equation}
        \mathbb{E}[\| \int_{0}^{T} \mathbf{M}_{t} d \mathbf{W}_{t} \|^{2}] = \mathbb{E}[ \int_{0}^{T} \| \mathbf{M}_{t} \|_{F}^{2} dt ]
    \end{equation}
    where $\| \cdot \|$ and $\| \cdot \|_{F}$ denote the Euclidean and Frobenius norms respectively.
\end{lemma}
\begin{proof}
\begin{eqnarray*}
    \mathbb{E}[\| \int_{0}^{T} \mathbf{M}_{t} d \mathbf{W}_{t} \|^{2}] &=& \mathbb{E}[ \sum_{i=1}^{n} \int_{0}^{T}  \int_{0}^{T} \sum_{j=1}^n ( \mathbf{M}_{t,ij} \: d W_{t,j} )^{2}] \\
    &=& \sum_{i=1}^{n} \sum_{j=1}^{n} \mathbb{E}( \int_{0}^{T} \int_{0}^{T} ( \mathbf{M}_{t,ij} \: dW_{j} )^{2} ) \\
    &=& \sum_{i=1}^{n} \sum_{j=1}^{n} \mathbb{E} (\int_{0}^{T} \mathbf{M}_{t,ij}^{2} dt) \\
    &=& \mathbb{E} [\int_{0}^{T} \| \mathbf{M}_{t} \|_{F}^{2} dt] \\
  \end{eqnarray*}

where we have also used the 1-dimensional Ito isometry.
If $\mathbf{M}_{t}$ is deterministic, we can simplify this result further to
\begin{eqnarray*}
    \mathbb{E}[\| \int_{0}^{T} \mathbf{M}_{t} d \mathbf{W}_{t} \|^{2}] = \int_{0}^{T} \| \mathbf{M}_{t} \|_{F}^{2} dt.
\end{eqnarray*}
\end{proof}

In view of the multivariate Ito isometry, if we set $\mathbf{F}_{0}(0) = \mathbf{0}$, we immediately see that
\begin{equation}
    \mathbb{E}(\|\mathbf{F}_{0}(t)\|^2) = \int_{0}^{t} \|e^{\boldsymbol{\Theta}(s-t)} \boldsymbol{\Sigma}\|_{F}^{2} ds.
\end{equation}

And utilizing the Markov inequality, we obtain for a threshold $x_{*} > 0 $
\begin{equation}
    \mathbb{P}(\|\mathbf{F}_{0}(t)\| \geq x_{*}) \leq \frac{\mathbb{E}(\|\mathbf{F}_{0}(t)\|^2)}{x_{*}^{2}}
\end{equation}
which assuming a strictly positive drift matrix $\boldsymbol{\Theta}$ further simplifies to
\begin{equation}
    \mathbb{P}(\|\mathbf{F}_{0}(t)\| \geq x_{*}) \leq \frac{\int_{0}^{t} \| e^{\boldsymbol{\Theta}(s-t)}\|_{op}^{2} \|\boldsymbol{\Sigma} \|_{F}^{2}ds}{x_{*}^{2}} = \frac{[1-e^{-2 \lambda_{\min} t}]}{2 \lambda_{\min} x_{*}^{2}} \|\boldsymbol{\Sigma} \|_{F}^{2}. \label{eq:MarkovUnion1}
\end{equation}
where $\lambda_{\min}$ is the smallest eigenvalue of the drift matrix $\boldsymbol{\Theta}$, $\| \cdot \|_{\mathrm{op}}$ denotes the operator norm, and we have used the fact that $\| \mathbf{A B} \|_{F} \leq \|\mathbf{A} \|_{\mathrm{op}} \|\mathbf{B} \|_{F}$. Since a strictly positive definite drift matrix $\boldsymbol{\Theta}$ is a sufficient condition for mean reversion in the $n$-dimensional OU process, imposing this condition ensures that we always have positive eigenvalues. Thus, the probability that the norm of $\mathbf{F}_{0}(t)$ stays below a threshold $x_{*}$ at time $t$ is bounded from below by
\begin{equation}\label{eq:Prob-F0}
    \mathbb{P}(\|\mathbf{F}_{0}(t)\| < x_{*} ) \geq 1- \frac{[1-e^{-2 \lambda_{\min} t}]}{2 \lambda_{\min} x_{*}^{2}} \|\boldsymbol{\Sigma} \|_{F}^{2}.
\end{equation}
And again, using similiar arguments as before by sampling $t$ at random times uniformly distributed over the interval $[0,T]$, we can use conditional probabilities and the law of total probability
to obtain a lower bound for the probability that the norm stays below a threshold $x_{*}$ over a time interval $[0,T]$
\begin{equation}
 \mathbb{P}_{s \in [0,T]}( \| \mathbf{F}_{0}(s) \| < x_{*}) \geq 1-\frac{\int_{0}^{T}[1-e^{-2 \lambda_{\min} s}] \frac{1}{T} ds}{2 \lambda_{\min} x_{*}^{2}} \|\boldsymbol{\Sigma} \|_{F}^{2} = 1- \frac{[1-\frac{1}{2 T \lambda_{\min}}(1-e^{-2 \lambda_{\min} T})] }{2 \lambda_{\min} x_{*}^{2}} \|\boldsymbol{\Sigma} \|_{F}^{2}.
\end{equation}

\section{Probabilistic Carleman linearization for NSDEs}\label{sec:pcl}

We consider the same general quadratic ODE of \cref{eq:main_system}
but with the inhomogeneous term $\mathbf{F}_0(t) \in \mathbb{R}^n$ being an OU process defined by \cref{eq:general_ou}.
The Carleman linearized form is given in \cref{pre:cl}. Our aim here is to provide error bounds for the case of stochastic driving. This is challenging because worst-case error estimates from conventional methods fail, and the Carleman linearization process maps stochastic inhomogeneities into both the generator and inhomogeneous term of the linearized system.

A further issue that we address involves the fact that Carleman linearization truncation bounds are originally developed for dissipative systems \cite{liu_efficient_2021}. For the non-dissipative cases, in solutions to nonlinear Schr\"{o}dinger equations, continuity equations or nonlinear Fokker-Planck equations, the study of Carleman linearization bounds has attracted increasing attention \cite{wu2025quantum,jennings2025quantum}.  We provide a method below for partially addressing this issue by providing bounds for Carleman linearization for systems that have a known approximate steady state distribution.  This is typical in ergodic systems, but also arises in Schr\"{o}dinger dynamics when considering operators such as the time-averaged density matrix $\bar{\rho}(T) := T^{-1} \int_0^T e^{-iHt}\,\rho(0)\,e^{iHt}\,dt$.

\subsection{Carleman linearization for steady-state NSDEs}
\label{sec:steady-state}
Carleman linearization is the workhorse of quantum algorithms for solving nonlinear differential equations, but despite this the range of validity of such linearizations can be quite limited.  In particular, a major challenge that we face is that in order to guarantee rapid convergence we need to have that all of the eigenvalues of the differential operator have negative real parts of their eigenvalues.  This excludes a number of significant applications, such as certain nonlinear Schr\"{o}dinger equations or systems with mass or probability conservation.  We generalize these approaches here by discussing how to generalize these stability criteria to situations where the dynamics is not dissipative, but steady state distributions are reached.  This will be useful for many Brownian processes, wherein mass conservation may be satisfied but the forces acting on them are stochastic.

We further define the notion of a stationary state to be a solution for $\mathbf{x}$ wherein the time-derivative is zero.
\begin{definition}[Stationary State]
  A stationary state of \cref{eq:main_system}, $\xst \in \mathbb{R}^n$ satisfies:
\begin{equation}
\label{eq:stationary_condition}
\frac{\d{\xst}}{\d{t}} = \mathbf{F}_2 \xst^{\otimes 2} + \mathbf{F}_1 \xst + \mathbf{F}_0 = 0
\end{equation}
\end{definition}
We then anticipate that if multiple trajectories evolve to the same stationary distribution then we anticipate that for systems that are sufficiently close to the equilibrium distribution that the dynamics can lead to exponentially decaying differences from the stationary distribution.  These decaying differences will then allow us to apply Carleman linearization on the transformed state.  In order to understand the impact of subtracting off a stationary distribution, it is useful to introduce the following bilinear operator.

\begin{definition}[Bilinear Operator]
  \label[definition]{def:bo}
For vectors $a, b \in \mathbb{R}^n$, define the bilinear operator $B_1(a): \mathbb{R}^n \to \mathbb{R}^{n^2} \times \mathbb{R}^{n}$ such that:
\begin{equation}
(a + b)^{\otimes 2} = a^{\otimes 2} + B_1(a) \: b + b^{\otimes 2}
\end{equation}
where $B_1(a) = \frac{\partial (x^{\otimes 2})}{\partial x^{T}}\bigg|_{x=a}$ and $a^{\otimes 2} = [a_1^2, a_1 a_2, \ldots, a_1 a_n, a_2 a_1, \ldots, a_n a_{n-1}, a_n^2]^T \in \mathbb{R}^{n^2}$ is a full state Kronecker tensor product.
Specifically if we enumerate the positions in the vector by the pair $(i,j)$ then
\begin{enumerate}
    \item For diagonal terms ($i = j$): bilinear part is $[B(a)b]_{(i,i)}=2a_ib_i$
    \item For off-diagonal terms ($i \neq j$): bilinear part is $[B(a)b]_{(i,j)}=a_i b_j + a_j b_i$
\end{enumerate}
\end{definition}

A major challenge facing the application of quantum algorithms for differential equations stems from the fact that the dynamics needs to be dissipative.  This creates a problem for dynamical systems that are conservative, such as nonlinear Schr\"{o}dinger equations or lattice Boltzmann equations.  However, if the dynamics in question reaches a steady state distribution then we can re-arrange the dynamics in such a way as to alleviate this problem if the Jacobian is negative-definite at the steady state distribution.  To clarify, by steady state we mean that $\partial_t \xst(t) = 0$ for the distribution $\xst$.

\begin{lemma}
	\label[lemma]{lemma:perturbation_dynamics}
Let $\delta \mathbf{x} = \mathbf{x} - \xst$ be the perturbation of $\mathbf{x}$ from its stationary state $\xst$. Then the perturbation dynamics are governed by the homogeneous differential equation.
\begin{equation}
\label{eq:perturbation_dynamics}
\frac{d\delta \mathbf{x}}{dt} =  \mathbf{F}_2 (\delta \mathbf{x})^{\otimes 2} + J (\delta \mathbf{x})
\end{equation}
where $J = \mathbf{F}_1 + \mathbf{F}_2 B_1(\xst)$ is the Jacobian matrix evaluated at the stationary state.
\end{lemma}

\begin{proof}
    Since $\delta \mathbf{x} = \mathbf{x} - \xst$ and $\frac{d\xst}{dt} = 0$:
    $$\frac{d\delta \mathbf{x}}{dt} = \frac{d\mathbf{x}}{dt} - \frac{d\xst}{dt} = \frac{d\mathbf{x}}{dt}$$
Which implies from the differential equation~\eqref{eq:main_system} that
    $$\frac{d\delta \mathbf{x}}{dt} = \mathbf{F}_2 (\xst + \delta \mathbf{x})^{\otimes 2} + \mathbf{F}_1 (\xst + \delta \mathbf{x}) + \mathbf{F}_0$$

      Using the bilinear operator definition of \cref{def:bo}:

  \begin{align}
    (\delta \mathbf{x} + \xst)^{\otimes 2} & = (\delta \mathbf{x})^{\otimes 2} + (\delta \mathbf{x})\otimes \xst + \xst \otimes  (\delta \mathbf{x}) + \xst^{\otimes 2} \\
  				  & = \xst^{\otimes 2} + B_1(\xst) \delta \mathbf{x} + (\delta \mathbf{x})^{\otimes 2}
  \end{align}
which yields
    \begin{align}
    \frac{d\delta \mathbf{x}}{dt} &= \mathbf{F}_2 \xst^{\otimes 2} + \mathbf{F}_2 B_1(\xst) \delta \mathbf{x} + \mathbf{F}_2 (\delta \mathbf{x})^{\otimes 2} \nonumber\\
    &\quad + \mathbf{F}_1 \xst + \mathbf{F}_1 \delta \mathbf{x} + \mathbf{F}_0\label{eq:feqn}
    \end{align}

    From equation \eqref{eq:stationary_condition}, we have $\mathbf{F}_2 \xst^{\otimes 2} + \mathbf{F}_1 \xst + \mathbf{F}_0 = 0$, so:
    $$\frac{d\delta \mathbf{x}}{dt} = \mathbf{F}_2 B_1(\xst) \delta \mathbf{x} + \mathbf{F}_1 \delta \mathbf{x} + \mathbf{F}_2 (\delta \mathbf{x})^{\otimes 2}$$

    Setting $J = \mathbf{F}_1 + \mathbf{F}_2 B_1(\xst)$ yields the desired result:
    $$\frac{d\delta \mathbf{x}}{dt} = J \delta \mathbf{x} + \mathbf{F}_2 (\delta \mathbf{x})^{\otimes 2}$$
\end{proof}

There are a few features that emerge as a direct consequence of the previous lemma.  The first thing to note is that for bounded $\mathbf{F}_2$ and any $\eps>0$ there exists $\widetilde{\delta \mathbf{x}}$ such that for all $\|\delta \mathbf{x}\| \le \|\widetilde{\delta \mathbf{x}}\|$ we have that $\|d \delta \mathbf{x} /dt - J(\delta \mathbf{x})\|\le \eps$.  By which we mean that we can always find a region close enough to the stationary point that the nonlinearity becomes arbitrarily weak.  This means that by performing this translation of the solution vector we can make the differential equation arbitrarily close to linear if the initial configuration is sufficiently close to the stationary distribution.

This result shows that we can approximately linearize the differential equation by shifting the differential equation, but this does not necessarily mean that we will have stable dynamics within an $\eps$-ball of the stationary dynamics.  We can understand this through the application of perturbation theory, as given by the Gershgorin Circle Theorem (stated below).

\begin{theorem}[Gershgorin Circle Theorem]\label{thm:gersh}
Let $\mathbf{A} = [a_{ij}]$ be a diagonalizable $n \times n$ complex matrix. For $i = 1, 2, \ldots, n$, let 
$R_i = \sum_{j \neq i} |a_{ij}|$
be the sum of the absolute values of the non-diagonal entries in the $i$-th row. Let $D(a_{ii}, R_i)$ denote the closed disk centered at $a_{ii}$ with radius $R_i$. Then every eigenvalue of $\mathbf{A}$ lies within at least one of the Gershgorin disks $D(a_{ii}, R_i)$,
$$D_i = \{z \in \mathbb{C} : |z - a_{ii}| \leq \sum_{j \neq i} |a_{ij}|\}$$

Moreover, if a set of $k$ Gershgorin disks is isolated from the other $n-k$ disks (i.e., they form a connected region that doesn't intersect with the remaining disks), then exactly $k$ eigenvalues of $\mathbf{A}$ lie within the union of these $k$ disks.
\end{theorem}
The central idea now is to probe how a weak shift to the Jacobian matrix can shift the stability of the system.  Let us assume first that $J$ is a matrix that has eigenvalues that only have negative real parts for their eigenvalues.  Further, let us assume that any eigenvalue of $J$ satisfies ${\mathrm{Re}}(\lambda(J)) \le -\alpha$.  Thus any perturbation to the dynamics that shifts the eigenvalues by less than $\alpha$ cannot remove the fact that the dynamics will converge to the stationary distribution.  This further leads to an issue because the stationary distribution will often not be known precisely.  This prevents the approximate linearization scheme discussed above from being directly applied.  We discuss this issue below and show that sufficiently small errors in the stationary distribution will not qualitatively change the nature of the convergence.

\begin{lemma}
\label{thm:critical_perturbation}
Let $\xst\in \mathbb{R}^n$ be an approximation to the stationary distribution and let $\mathbf{x}_*$ be the true stationary distribution.  We then have that if the dynamics of $\delta\!\mathbf{x}$ can be expressed as
$$
\partial_t (\mathbf{x}-\mathbf{x}_*) = \mathbf{A} (\mathbf{x}-\mathbf{x}_*)+B(\mathbf{x}-\mathbf{x}_*)^{\otimes 2}
$$
for $\mathbf{A}\preceq -\alpha I$ for $\alpha>0$ then if $|\xst-\mathbf{x}_*|_1\le \alpha/2$ then 
there exist $\tilde{\mathbf{A}}\prec (-\alpha + 2|\xst-\mathbf{x}_*|_1)I $ and $G$ such that
the differential equation for the approximate evolution satisfies
$$
\partial_t \delta\!\mathbf{x} = \tilde{\mathbf{A}} \delta\!\mathbf{x} + \tilde{B}(\delta\!\mathbf{x})^{\otimes 2} + G
$$
for inhomogeneity $G$ then we have that $\tilde{\mathbf{A}}\prec (-\alpha + 2|\xst-\mathbf{x}_*|_1)I $.
\end{lemma}

\begin{proof}
If we assume that $\xst$ is only approximately stationary then we have that
\begin{equation}
    \frac{\partial \delta\!\mathbf{x}}{dt} = \mathbf{F}_2 B_1(\xst) \delta\!\mathbf{x} + \mathbf{F}_1 \delta\!\mathbf{x} +\mathbf{F}_2(\delta\!\mathbf{x})^{\otimes 2} + G_0
\end{equation}
for an inhomogeneity $G_0$ which arises from the assumption that $\xst$ is not stationary.  Now if we let $\mathbf{x}_*$ be the true stationary distribution then
\begin{equation}
  \frac{\partial \delta\!\mathbf{x}}{dt} = \mathbf{F}_2 B_1(\mathbf{x}_{*} + (\xst-\mathbf{x}_*)) \delta\!\mathbf{x} + \mathbf{F}_1 \delta\!\mathbf{x} +\mathbf{F}_2(\delta\!\mathbf{x})^{\otimes 2} + G_0(\xst)
\end{equation}
Now we have that from bilinearity
\begin{equation}
  [B_1(\mathbf{x}_*+ (\xst -\mathbf{x}_*))\delta\!\mathbf{x}]_{(i,j)} - [B_1(\mathbf{x}_*)\delta\!\mathbf{x}]_{(i,j)}=[B_1(\xst-\mathbf{x}_*)\delta\!\mathbf{x}]_{(i,j)} = [\xst -\mathbf{x}_*]_i \,\delta\!\mathbf{x}_j+[\xst -\mathbf{x}_*]_j \,\delta\!\mathbf{x}_i.
\end{equation}
We then see that the Gershgorin Radii of this operator satisfy
\begin{equation}
    \max_i R_i \le 2 \sum_j |\xst - \mathbf{x}_*|_{j} = 2 |\xst - \mathbf{x}_*|_1.
\end{equation}
  Under the assumption that the Jacobian evaluated at $\mathbf{x}_*$ is negative definite with largest eigenvalue $-\alpha$, we then have that the perturbation will only change this to negative semi-definite if the perturbation strength is less than $\alpha/2$ from \cref{thm:gersh}, we know that every eigenvalue of the perturbed matrix lies within a radius of at most $\alpha$ of the unperturbed eigenvalue.  Rephrasing this observation, we have that if we express the differential equation as 
$$
  \partial_t \delta\!\mathbf{x} = \tilde{\mathbf{A}} \delta\!\mathbf{x} + \tilde{B}(\delta\!\mathbf{x})^{\otimes 2} + G
$$
  then there are no eigenvalues of $\tilde{\mathbf{A}}$ that are greater than $-\alpha + 2|\xst - \mathbf{x}_*|_1$, which justifies the claim.

\end{proof}
This shows that if we provide an estimate that is sufficiently close to the true fixed point and if the Jacobian evaluated there is negative definite then small errors in the stationary distribution can be tolerated.  This means that we do not need to know the precise equilibrium distribution at the stationary point in order to guarantee that the dynamics will be locally dissipative.  This is important for error bounds on Carleman linearization, which is a procedure that replaces the nonlinear differential equation with a truncation of an infinite system of linear differential equations.  The error in this truncation can be estimated under a number of different assumptions, but the most general and widely used assumption is that the linear part of the differential equation is negative-definite.

\Cref{lemma:perturbation_dynamics} allows us to rewrite the forced nonlinear differential equations in terms of the perturbation from their steady state, which is governed by \cref{eq:perturbation_dynamics} with a stability guarantee given by \cref{thm:critical_perturbation}. We can now apply Carleman linearization to the dynamics of the perturbation $\delta \mathbf{x}$ (\cref{eq:perturbation_dynamics}), which is a homogenized system that dissipates by design. The resulting Carleman linearization error for the deterministic $\delta \mathbf{x}$ in \cref{lemma:carleman-delta-f} will then be used for the stochastic case in \cref{thm:deltax-tail}.   

\begin{lemma}[Carleman linearization error for $\delta \mathbf{x}$]\label[lemma]{lemma:carleman-delta-f}
Consider an instance of \cref{eq:perturbation_dynamics}, with its corresponding Carleman linearization as defined in \cref{pre:cl} and stability criterion defined in \cref{def:stable-ODE}.
Then for any $j \in [N]$, the error 
  \[
    \boldsymbol{\eta}_j(t) \coloneqq (\delta \mathbf{x})^{\otimes j}(t)-\mathbf{y}_j
  \]
  satisfies
\begin{equation}\label{eq:etaj-bound}
\|\boldsymbol{\eta}_j(t)\|_P \;\le\; \|\delta \mathbf{x}_{\mathrm{in}}\|_P^{\,N+1}\, \frac{\|\mathbf{F}_2\|_P^{\,N+1-j}}{|\mu_P(\mathbf{F}_1)|^{\,N+1-j}},
\end{equation}
where $\mu_P(\mathbf{F}_1)$ is the logarithmic norm of $\mathbf{F}_1$ with respect to the norm $\|\cdot\|_P$ defined in \cref{eq:muP}.
For $j=1$, we have the tighter bound
\begin{equation}\label{eq:eta1-exact}
\|\boldsymbol{\eta}_1(t)\|_P \;\le\; \|\delta \mathbf{x}_{\mathrm{in}}\|_P\, \left(\frac{\|\mathbf{F}_2\|_P}{|\mu_P(\mathbf{F}_1)|}\right)^N\, \big(1-e^{\mu_P(\mathbf{F}_1)\, t}\big)^{N}.
\end{equation}
\end{lemma}

The proof of \cref{lemma:carleman-delta-f} is provided in \cref{sec:proof-carleman-delta-f}.

Given a steady state $\xst$ of a forced nonlinear system \cref{eq:main_system}, we show in \cref{lemma:perturbation_dynamics} that it can be transformed to a dissipative homogeneous system, the Carleman linearization error of which can be bounded using the Carleman linearization for dissipative system (\cref{lemma:carleman-delta-f}). Assuming a steady state $\mathbf{y}_N^s \in \mathbb{R}^{d_N}$ of the Carleman truncated solution state $\mathbf{y}_N$ exists, we now show that $\mathbf{y}_N^s$ converges exponentially to $\mathbf{y}_N$ in \cref{thm:error_bound} providing that the forced nonlinear system is stable, i.e., \cref{eq:eAt-norm} is satisfied.

Now that we have shown that by subtracting off a stationary distribution we can make dynamics dissipative we will now assume that the dynamics is fundamentally dissipative.  We will now proceed under these assumptions to bound the error in Carleman linearization.  An elementary result, known as Gr\"{o}nwall's inequality, is useful for bounding the dynamics of the error.  At a high level, the theorem states that if the derivatives of a function are consistently greater than those of your function in a differential equation then the solution to the corresponding differential equation will upper bound the solution to your differential equation.  The result is straightforward to prove from the fundamental theorem of calculus, and we state this inequality below for reference.
\begin{theorem}[Grönwall's Inequality]\label{prel:gronwall}
Let $\alpha(t)$ and $\beta(t)$ be continuous non-negative functions defined on $[0,T]$, and let $u(t)$ be a continuous function satisfying
\begin{equation}
u(t) \leq \alpha(t) + \int_0^t \beta(s)u(s)ds \quad \text{for all } t \in [0,T]
\end{equation}
Then
\begin{equation}
u(t) \leq \alpha(t) + \int_0^t \alpha(s)\beta(s)\exp\left(\int_s^t \beta(r)dr\right)ds \quad \text{for all } t \in [0,T]
\end{equation}

In particular, if $\alpha(t) = \alpha$ is a constant (or non-decreasing ), then
\begin{equation}
u(t) \leq \alpha\exp\left(\int_0^t \beta(s)ds\right) \quad \text{for all } t \in [0,T]
\end{equation}
\end{theorem}

\begin{theorem}\label{thm:error_bound}
Assuming a steady state $\mathbf{y}_N^s \in \mathbb{R}^{d_N}$ of the Carleman truncated solution state $\mathbf{y}_N$ exists and $\mathbf{x}^{(N)}(t)$ and $\mathbf{x}_s^{(N)}$ both satisfy $\|\mathbf{x}^{(N)}(t)\|_{\mathbb{R}^n}, \|\mathbf{x}_s^{(N)}\|_{\mathbb{R}^n} \leq \hat{C}$,  under \cref{def:qdns}, the error between the time-dependent solution and the stationary state, $\mathbf{e}_N(t) = \mathbf{y}_N(t) - \mathbf{y}_N^s$, satisfies for $D_N := N(N+1)\|\mathbf{F}_2\|$:
  \begin{equation}\label{eq:eNt-bound}
\|\mathbf{e}_N(t)\| \leq \hat{C}e^{-(\alpha - D_N \hat{C}^{N+1})t}\|\mathbf{e}_N(0)\|
\end{equation}
provided that $\alpha - D_N \hat{C}^{N+1} > 0$, where $\alpha = -\lambda(\mathbf{A}_N)_{\max} > 0$.
\end{theorem}

\begin{proof}
According to \cref{eq:truncated_system}, a stationary state of the $N$th-order truncated Carleman system $\mathbf{y}_N^s \in \mathbb{R}^{d_N}$ satisfying:
\begin{equation}\label{eq:truncated_stationary}
\mathbf{0} = \mathbf{A}_N\mathbf{y}_N^s + \mathbf{b}_N + \mathbf{R}_N^s
\end{equation}
where $\mathbf{R}_N^s$ represents the truncation error at the stationary state satisfying:
\begin{equation}\label{eq:R-Ns}
\|\mathbf{R}_N^s\| = \|\mathbf{A}_{N+1}^N \mathbf{x}_s^{\otimes(N+1)}\| \leq N \|\mathbf{F}_2\| \|\mathbf{x}_s(t)\|^{N+1},
\end{equation}
which is a direct result of \cref{eq:RNt} at a steady state.

Subtracting \cref{eq:truncated_stationary} from \cref{eq:truncated_system},
the error vector $\mathbf{e}_N(t)$ evolves according to:
\begin{equation}\label{eq:eNt}
\frac{d\mathbf{e}_N}{dt} = \mathbf{A}_N\mathbf{e}_N + \boldsymbol{\delta}_N(t)
\end{equation}
where 
\begin{equation}\label{eq:delta-Nt}
\boldsymbol{\delta}_N(t) = \mathbf{R}_N(t) - \mathbf{R}_N^s
\end{equation}
is the difference between the time-dependent truncation error $\mathbf{R}_N(t)$ defined in \cref{lemma:RNt} and the stationary truncation error $\mathbf{R}_N^s$. The solution of \cref{eq:eNt} is given by
\begin{align}
\mathbf{e}_N(t) = e^{\mathbf{A}_N t} \mathbf{e}_N(0) + \int_0^t e^{\mathbf{A}_N(t-s)} \boldsymbol{\delta}_N(s) ds
\end{align}
Taking the norm of the error solution and using the property that $\mathbf{A}_N \prec -\alpha I$ we have that
  \begin{equation}
      \label{eq:eAt-norm}
\|e^{\mathbf{A}_N t}\| \leq \hat{C} e^{-\alpha t},
 \end{equation}
we obtain
\begin{align}
\label{eq:eNt-norm1}
\|\mathbf{e}_N(t)\| &\leq \|e^{\mathbf{A}_N t}\|\|\mathbf{e}_N(0)\| + \int_0^t \|e^{\mathbf{A}_N(t-s)}\| \|\boldsymbol{\delta}_N(s)\|ds \\
\label{eq:eNt-norm2}
&\leq \hat{C} e^{-\alpha t}\|\mathbf{e}_N(0)\| + \hat{C} \int_0^t e^{-\alpha(t-s)} \|\boldsymbol{\delta}_N(s)\| ds.
\end{align}
To bound \cref{eq:eNt-norm2}, we first bound $\|\boldsymbol{\delta}_N(s)\|$. 
Combining \cref{eq:delta-Nt}, \cref{eq:R-Ns}, and \cref{eq:RNt}:
\begin{align}
\|\boldsymbol{\delta}_N(t)\| &= \|\mathbf{R}_N(t) - \mathbf{R}_N^s\| \\
\label{eq:delta-Nt-dum}
&\leq N\|\mathbf{F}_2\|\|\mathbf{x}^{(N)}(t)^{N+1} - (\mathbf{x}_s^{(N)})^{N+1}\|
\end{align}
By the mean value theorem for the function $h(x) = x^{N+1}$, there exists $\xi$ between $\mathbf{x}^{(N)}(t)$ and $\mathbf{x}_s^{(N)}$ such that:
\begin{align}\label{eq:fNt-fsNt}
\|\mathbf{x}^{(N)}(t)^{N+1} - (\mathbf{x}_s^{(N)})^{N+1}\| &= \|(N+1)\xi^N(\mathbf{x}^{(N)}(t) - \mathbf{x}_s^{(N)})\| \\
&\leq (N+1)\hat{C}^N\|\mathbf{x}^{(N)}(t) - \mathbf{x}_s^{(N)}\| \\
&= (N+1)\hat{C}^N\|\mathbf{e}_{1,N}(t)\|
\end{align}
where we used the assumption that both $\|\mathbf{x}^{(N)}(t)\|_{L^\infty}, \|\mathbf{x}_s^{(N)}\|_{L^\infty} \leq \hat{C}$, which implies $\|\xi\|_{L^\infty} \leq \hat{C}$.
Plugging \cref{eq:fNt-fsNt} into \cref{eq:delta-Nt-dum} yields
\begin{align}
\|\boldsymbol{\delta}_N(t)\| &\leq N\|\mathbf{F}_2\|(N+1)\hat{C}^N\|\mathbf{e}_{1, N}(t)\| \\
			     &\leq N\|\mathbf{F}_2\|(N+1)\hat{C}^N\|\mathbf{e}_N(t)\| \\
			     \label{eq:deltaNt}
			     &= D_N \hat{C}^N\|\mathbf{e}_N(t)\|,
\end{align}
where $D_N = N(N+1)\|\mathbf{F}_2\|$.

Plugging the bound of $\|\boldsymbol{\delta}_N(s)\|$ in \cref{eq:deltaNt} into \cref{eq:eNt-norm2}, we obtain 
\begin{align}
\label{eq:eNt-norm3}
\|\mathbf{e}_N(t)\| \leq  \hat{C}e^{-\alpha t}\|\mathbf{e}_N(0)\| + \hat{C}\int_0^t e^{-\alpha(t-s)}D_N \hat{C}^N\|\mathbf{e}_N(s)\|ds
\end{align}

Let $g(t) = e^{\alpha t}\|\mathbf{e}_N(t)\|$ and apply it to \cref{eq:eNt-norm3}, we obtain:
\begin{align}\label{eq:gt-dum}
g(t) &\leq \hat{C}\|\mathbf{e}_N(0)\| + \hat{C}D_N \hat{C}^N\int_0^t g(s)ds.
\end{align}
Applying Grönwall's inequality presented in \cref{prel:gronwall} to \cref{eq:gt-dum}, we get
\begin{align}
g(t) &\leq \hat{C}\|\mathbf{e}_N(0)\|e^{\hat{C}D_N \hat{C}^N t}
\end{align}
Therefore:
\begin{align}
  \|\mathbf{e}_N(t)\| = e^{-\alpha t} g(t) &\leq \hat{C}e^{-\alpha t}\|\mathbf{e}_N(0)\|e^{\hat{C}D_N \hat{C}^N t} \\
&= \hat{C}e^{-(\alpha - D_N \hat{C}^{N+1})t}\|\mathbf{e}_N(0)\|
\end{align}
ensuring $\alpha - D_N \hat{C}^{N+1} > 0$ completes the proof.
\end{proof}
 
\subsection{Probabilistic Carleman linearization error bounds}

In this section, we provide a probabilistic Carleman linearization error bound for the solution of the quadratic ODE system driven by an OU process \cref{eq:quad-OU-ODE} that enters into both the Carleman matrix and the inhomogeneous term.
Since the Carleman linearization error depends on the solution norm of the \textit{stochastic} quadratic ODE system, we first derive a \textit{pathwise} bound for the solution norm of \cref{eq:quad-OU-ODE} in \cref{lem:lyap-path-solution-norm} below, which will then be used in \cref{lem:lyap-path} and \cref{lem:lyap-tail}.

\begin{lemma}\label[lemma]{lem:lyap-path-solution-norm}
Let $\mathbf{x}(t)\in\mathbb{C}^n$ be the solution of the OU-driven quadratic system
\begin{equation}\label{eq:quad-OU-ODE}
\dot{\mathbf{x}}(t) \;=\; \mathbf{F}_2\big(\mathbf{x}(t)^{\otimes 2}\big) \;+ \mathbf{F}_1\, \mathbf{x}(t) \;+\; \mathbf{F}_0(t),\qquad \mathbf{x}(0)=\mathbf{x}_{\mathrm{in}}\in\mathbb{C}^n,
\end{equation}
with $\mathbf{F}_1\in\mathbb{C}^{n\times n}$, $\mathbf{F}_2:\mathbb{C}^{n^2}\to \mathbb{C}^n$ linear, and $\mathbf{F}_0(t)$ an $n$-dimensional OU process as in \cref{eq:general_ou}.
Suppose \cref{eq:quad-OU-ODE} is Lyapunov stable (\cref{def:stable-ODE}), namely, there exists a Hermitian $P>0$ (Lyapunov matrix) such that the logarithmic norm
\begin{align}
\label{eq:lyapunov-ineq-F1}
\mu_P(\mathbf{F}_1)
&:= \lambda_{\max}\!\Big(
\tfrac{1}{2}\big(P^{1/2}\mathbf{F}_1P^{-1/2} + (P^{1/2}\mathbf{F}_1P^{-1/2})^\dagger\big)
\Big) \;<\; 0,  \\[4pt]
\label{eq:lyapunov-ineq}
\beta_P
&:= \sup_{x\neq 0}
\frac{\Re\langle \mathbf{x},\, \mathbf{F}_2(\mathbf{x}^{\otimes 2})\rangle_P}{\|\mathbf{x}\|_P^2}
\;<\;\infty.
\end{align}
Define, for a fixed $\gamma>0$, the Lyapunov growth rate 
\begin{equation}\label{eq:kappa_P}
\kappa_P \;:=\; 2\mu_P(\mathbf{F}_1) + 2\beta_P + \gamma.
\end{equation}
Then, for all $t\ge 0$,
\begin{equation}\label{eq:xP-bound}
\|\mathbf{x}(t)\|_P^2 \;\le\; e^{\kappa_P t}\,\|\mathbf{x}_{\mathrm{in}}\|_P^2 \;+\; \frac{e^{\kappa_P t}-1}{\gamma\,\kappa_P}\, \sup_{0\le s\le t} \|\mathbf{F}_0(s)\|_P^2.
\end{equation}
\end{lemma}

\begin{proof}
We first state the Cauchy-Schwarz inequality in the $P$-inner product:
\begin{equation}\label{eq:CS-standard}
|\langle u,v\rangle_P| \;\le\; \|u\|_P\, \|v\|_P
\qquad u,v\in\mathbb{C}^n.
\end{equation}
and an elementary algebra inequality
\begin{equation}\label{eq:young-standard}
(\sqrt{\gamma}\,a - b/\sqrt{\gamma})^2 \ge 0 \Rightarrow \gamma a^2 + \frac{1}{\gamma}b^2 \ge 2ab
\end{equation}
that can be also derived from 
Young’s inequality with $p=q=2$ and substituting $a\leftarrow\sqrt{\gamma}a\,\text{and} \,  b\leftarrow b/\sqrt{\gamma}$:
for any $a,b\in\mathbb{R}$ and any $\gamma>0$ \cite[Section B.2.c, Page 706]{evans2022partial},
$
  ab \le a^p / p + b^q / q.
$
Equivalently, in an inner-product space,
$
2\,|\Re\langle u,v\rangle| \;\le\; \gamma\, \|u\|^2 \;+\; \frac{1}{\gamma}\, \|v\|^2,
\forall\,\gamma>0.
$
Let $V(t):=\|\mathbf{x}(t)\|_P^2 = \mathbf{x}(t)^\dagger P \mathbf{x}(t)$.  By differentiating $V$ using  the product rule we arrive at
\begin{align}
\frac{d}{dt}\,V(t)
&= \frac{d}{dt}\big(\mathbf{x}^\dagger P \mathbf{x}\big)
= \dot{\mathbf{x}}^\dagger P \mathbf{x} \;+\; \mathbf{x}^\dagger P \dot{\mathbf{x}}
= 2\,\Re\big(\mathbf{x}^\dagger P \dot{\mathbf{x}}\big)
= 2\,\Re\,\langle \mathbf{x},\, \dot{\mathbf{x}}\rangle_P.
\label{eq:dV-dt}
\end{align}
Substituting \cref{eq:quad-OU-ODE} into \cref{eq:dV-dt} yields
\begin{align}
\frac{d}{dt}\,\|\mathbf{x}(t)\|_P^2
&= 2\,\Re\,\langle \mathbf{x},\, \mathbf{F}_1 \mathbf{x}\rangle_P \;+\; 2\,\Re\,\langle \mathbf{x},\, \mathbf{F}_2(\mathbf{x}^{\otimes 2})\rangle_P \;+\; 2\,\Re\,\langle \mathbf{x},\, \mathbf{F}_0(t)\rangle_P.
\label{eq:dV-expanded}
\end{align}
We bound each term on the right-hand side of \cref{eq:dV-expanded}.

For the $\mathbf{F}_1$-term, observe that
\begin{align}
2\,\Re\,\langle \mathbf{x},\, \mathbf{F}_1 \mathbf{x}\rangle_P
&= \mathbf{x}^\dagger(\mathbf{F}_1^\dagger P + P \mathbf{F}_1)\, \mathbf{x}
\;\le \; 2\mu_P(\mathbf{F}_1)\, \|\mathbf{x}\|_P^2,
\label{eq:F1-bound}
\end{align}
where we used the Lyapunov inequality \cref{eq:lyapunov-ineq-F1}.

The one-sided bound \cref{eq:lyapunov-ineq} can be used to bound the inner product of the quadratic nonlinear term
\begin{equation}\label{eq:F2-bound}
2\,\Re\,\langle \mathbf{x},\, \mathbf{F}_2(\mathbf{x}^{\otimes 2})\rangle_P \;\le\; 2\beta_P\, \|\mathbf{x}\|_P^2.
\end{equation}

For the forcing term, first apply Cauchy-Schwarz \cref{eq:CS-standard} in the $P$-inner product:
\begin{equation}\label{eq:CS-applied}
\big|\,\langle \mathbf{x},\,\mathbf{F}_0(t)\rangle_P\,\big|
\;\le\; \|\mathbf{x}\|_P\, \|\mathbf{F}_0(t)\|_P.
\end{equation}
Then by applying Young’s inequality \cref{eq:young-standard} with the substitution $a\gets \|\mathbf{x}\|_P$ and $b\gets \|\mathbf{F}_0(t)\|_P$, for an arbitrary fixed $\gamma>0$ yields:
\begin{equation}\label{eq:young-applied}
2\,\|\mathbf{x}\|_P\, \|\mathbf{F}_0(t)\|_P
\;\le\; \gamma\, \|\mathbf{x}\|_P^2 \;+\; \frac{1}{\gamma}\, \|\mathbf{F}_0(t)\|_P^2.
\end{equation}
Combining \cref{eq:CS-applied} and \cref{eq:young-applied} yields:
\begin{equation}\label{eq:young}
2\,\Re\,\langle \mathbf{x},\, \mathbf{F}_0(t)\rangle_P
\;\le\; \gamma\, \|\mathbf{x}\|_P^2 \;+\; \frac{1}{\gamma}\, \|\mathbf{F}_0(t)\|_P^2.
\end{equation}
Note: \cref{eq:young-applied} (hence \cref{eq:young}) holds for all $\gamma>0$. In addition, the ``free'' rescaling factor $\gamma$ as a tunable parameter in \cref{eq:kappa_P} is crucial for controlling the stability of $\mathbf{x}(t)$. 
Combining \cref{eq:dV-expanded} with \cref{eq:F1-bound} $-$ \cref{eq:young} yields
\begin{align}\label{eq:energy-P}
  \frac{d}{dt}\|\mathbf{x}(t)\|_P^2 & \;\le\;  \big(2\mu_P(\mathbf{F}_1) + 2\beta_P + \gamma\big)\, \|\mathbf{x}(t)\|_P^2 \;+\; \frac{1}{\gamma}\, \|\mathbf{F}_0(t)\|_P^2. \\
  \label{eq:enerby-kappa-P}
				   & = \kappa_P \|\mathbf{x}(t)\|_P^2 \;+\; \frac{1}{\gamma}\, \|\mathbf{F}_0(t)\|_P^2.
\end{align}
To bound $\|\mathbf{x}(t)\|_P^2$, we invoke the standard integral form of Grönwall-Bellman’s inequality: if $u:[0,t]\to\mathbb{R}$ satisfies
\[
  \dot{u}(s) \le \kappa\, u(s) + g(s)\quad\text{for almost all } s\in[0,t],
\]
with $\kappa\in\mathbb{R}$ constant and $g\in L^1([0,t])$, then
\begin{equation}\label{eq:gronwall}
u(t) \;\le\; e^{\kappa t}\, u(0) \;+\; \int_0^t e^{\kappa (t-s)}\, g(s)\, ds.
\end{equation}
Then by applying \cref{eq:gronwall} to \cref{eq:enerby-kappa-P} with
\[
u(s)=\|\mathbf{x}(s)\|_P^2,\qquad \kappa:=\kappa_P,\qquad g(s)=\|\mathbf{F}_0(s)\|^2,
\]
we obtain
\begin{equation}\label{eq:r-int}
\|\mathbf{x}(s)\|_P^2 \;\le\; e^{\kappa_P t}\, \|\mathbf{x}_{\mathrm{in}}\|_P^2 \;+\; \frac{1}{\gamma}\int_0^t e^{\kappa_P (t-s)}\, \|\mathbf{F}_0(s)\|^2\, ds.
\end{equation}
The integral in \cref{eq:r-int} can be then bounded by replacing the integrand by its supremum
\[
\int_0^t e^{\kappa_P (t-s)}\, \|\mathbf{F}_0(s)\|^2\, ds \;\le\; \sup_{0\le s\le t} \|\mathbf{F}_0(s)\|_P^2  \int_0^t e^{\kappa_P (t-s)}\, ds \;=\; \sup_{0\le s\le t} \|\mathbf{F}_0(s)\|_P^2  \,\frac{e^{\kappa_P t}-1}{\kappa_P},
\]
with the usual convention that $\frac{e^{\kappa_P t}-1}{\kappa_P}=t$ when $\kappa_P=0$. Combining this with \cref{eq:r-int} yields \cref{eq:xP-bound}.
\end{proof}

With the pathwise solution norm bound of \cref{lem:lyap-path-solution-norm}, we can now derive a pathwise Carleman linearization error bound for the OU-driven quadratic system \cref{eq:quad-OU-ODE}. This \textit{pathwise} Carleman linearization error bound (\cref{lem:lyap-path}) can then be used to derive a probabilistic Carleman linearization error bound in \cref{lem:lyap-tail}. We build our analysis upon the Carleman linearization method of \citet{liu_efficient_2021} for deterministic quadratic systems generalized into a Lyapunov stability framework \cite{jennings_quantum_2025}, extending it to the OU-driven quadratic system \cref{eq:quad-OU-ODE}.

\begin{lemma}[Pathwise Carleman linearization error bound for OU-driven quadratic system]\label[lemma]{lem:lyap-path}
Under the assumptions of and variables defined in \cref{lem:lyap-path-solution-norm},  
the Carleman linearization error vector $\boldsymbol{\eta}^{(N)}(t):=\mathbf{y}^{(N)}(t)-\hat{\mathbf{y}}^{(N)}(t)$ at order $N\ge 1$ of \cref{eq:quad-OU-ODE} satisfies
\begin{equation}\label{eq:lyap-path-main}
\|\boldsymbol{\eta}^{(N)}(t)\|_{P_N}^2 \;\le\; \Phi_t^2\, \|\mathbf{A}_{N+1}^N\|_P^2\, \big(a_t + b_t\, S_{t,P}^2\big)^{N+1},
\end{equation}
where $P_N:=\bigoplus_{j=1}^N P^{\otimes j}$ is the block-diagonal Lyapunov matrix induced by $P$ on the Carleman space, $\mathbf{A}_{N+1}^N$ is the super-diagonal Carleman block induced by $\mathbf{F}_2$, and
\begin{equation}\label{eq:params}
\Phi_t:=\frac{e^{\chi_P t}-1}{\chi_P},\qquad a_t:=e^{\kappa_P t}\,\|\mathbf{x}_{\mathrm{in}}\|_P^2,\qquad
b_t:=\begin{cases} \dfrac{e^{\kappa_P t}-1}{\gamma\,\kappa_P}, & \kappa_P\neq 0,\\[6pt] \dfrac{t}{\gamma}, & \kappa_P=0, \end{cases}
\end{equation}
with $\chi_P > 0$.
\end{lemma}

\begin{proof}
The Carleman residual vector $\mathbf{R}_N(t)$ from \cref{lemma:RNt} satisfies
\begin{equation}\label{eq:xi-bound2}
\|\mathbf{R}_N(s)\|_{P_N} \;\le\; \|\mathbf{A}_{N+1}^N\|_P\, \|\mathbf{x}(s)\|_P^{\,N+1}
\end{equation}
by applying the triangle inequality to \cref{eq:RNt}.
Then the Carleman linearization error vector $\boldsymbol{\eta}^{(N)}(t)$ at order $N$ evolves according to
\begin{equation}\label{eq:eta-evolution}
\dot{\boldsymbol{\eta}}^{(N)}(t) \;=\; \mathbf{A}_N(t)\, \boldsymbol{\eta}^{(N)}(t) \;+\; \mathbf{R}_N(t),
\end{equation}
Then the squared $P_N$-norm of the error satisfies
\begin{align}
\frac{d}{dt}\,\|\boldsymbol{\eta}^{(N)}(t)\|_{P_N}^2
&= \frac{d}{dt}\big(\boldsymbol{\eta}^{(N)}(t)^\dagger P_N \boldsymbol{\eta}^{(N)}(t)\big) \nonumber\\
&= \dot{\boldsymbol{\eta}}^{(N)}(t)^\dagger P_N \boldsymbol{\eta}^{(N)}(t) \;+\; \boldsymbol{\eta}^{(N)}(t)^\dagger P_N \dot{\boldsymbol{\eta}}^{(N)}(t) \nonumber\\
&= \boldsymbol{\eta}^{(N)}(t)^\dagger\big(\mathbf{A}_N(t)^\dagger P_N + P_N \mathbf{A}_N(t)\big)\boldsymbol{\eta}^{(N)}(t)
\;+\; 2\,\Re\,\big(\boldsymbol{\eta}^{(N)}(t)^\dagger P_N \mathbf{R}_N(t)\big),
\label{eq:d-norm-squared}
\end{align}
by plugging \cref{eq:eta-evolution} into the Leibniz rule.

Next we use the definition of the logarithmic norm given by \cref{eq:muP-eig} in the $P_N$-metric for a generic matrix $\mathcal{A}$ \citep[Section~5.6]{horn_topics_1991}:
\begin{equation}\label{eq:lognorm-Pk-def}
\mu_{P_N}(\mathcal{A}) \;:=\; \lambda_{\max}\!\left(\tfrac{1}{2}\,\Big(P_N^{1/2} \,\mathcal{A}\, P_N^{-1/2} + (P_N^{1/2} \,\mathcal{A}\, P_N^{-1/2})^\dagger\Big)\right),
\end{equation}
and set $B:=P_N^{1/2} \,\mathcal{A}\, P_N^{-1/2}$ and $H:=\tfrac12(B+B^\dagger)$, the Hermitian part of $B$.
Let $y:=P_N^{1/2}z$, so $z=P_N^{-1/2}y$ for an arbitrary vector $z$. Then
\begin{align}
z^\dagger\big(\mathcal{A}^\dagger P_N + P_N \,\mathcal{A}\big)z
&= (P_N^{-1/2}y)^\dagger \,\mathcal{A}^\dagger P_N (P_N^{-1/2}y) \;+\; (P_N^{-1/2}y)^\dagger P_N \,\mathcal{A}\, (P_N^{-1/2}y) \nonumber\\
&= y^\dagger P_N^{-1/2} \,\mathcal{A}^\dagger P_N P_N^{-1/2} y \;+\; y^\dagger P_N^{-1/2} P_N \,\mathcal{A}\, P_N^{-1/2} y \nonumber\\
&= y^\dagger \big( (P_N^{1/2} \,\mathcal{A}\, P_N^{-1/2})^\dagger + P_N^{1/2} \,\mathcal{A}\, P_N^{-1/2}\big) y \nonumber\\
&= 2\, y^\dagger H\, y.
\label{eq:Hermitian-part}
\end{align}
Since $H$ is Hermitian, the Rayleigh quotient bound yields $y^\dagger H y \le \lambda_{\max}(H)\, y^\dagger y$. Therefore, using $y^\dagger y = z^\dagger P_N z = \|z\|_{P_N}^2$ and $\lambda_{\max}(H)=\mu_{P_N}(\mathcal{A})$ by \cref{eq:lognorm-Pk-def}, \cref{eq:Hermitian-part} becomes
\begin{equation}\label{eq:mm-ineq}
z^\dagger\big(\mathcal{A}^\dagger P_N + P_N \,\mathcal{A}\big)z \;\le\; 2\, \lambda_{\max}(H)\, \|z\|_{P_N}^2 \;=\; 2\, \mu_{P_N}(\mathcal{A})\, \|z\|_{P_N}^2.
\end{equation}
Applying \cref{eq:mm-ineq} with $\mathcal{A} = \mathbf{A}_N(t)$ and $z = \boldsymbol{\eta}^{(N)}(t)$ in \cref{eq:d-norm-squared} yields
\[
\frac{d}{dt}\,\|\boldsymbol{\eta}^{(N)}(t)\|_{P_N}^2
\;\le\; 2\, \mu_{P_N}\big(\mathbf{A}_N(t)\big)\, \|\boldsymbol{\eta}^{(N)}(t)\|_{P_N}^2 \;+\; 2\, \|\boldsymbol{\eta}^{(N)}(t)\|_{P_N}\, \|\mathbf{R}_N(t)\|_{P_N}.
\]
When $\|\boldsymbol{\eta}^{(N)}(t)\|_{P_N}>0$, we divide both sides by $2\,\|\boldsymbol{\eta}^{(N)}(t)\|_{P_N}$ to get
the following differential inequality for the $P_N$-norm:
\begin{equation}\label{eq:d-norm}
\frac{d}{dt}\,\|\boldsymbol{\eta}^{(N)}(t)\|_{P_N}
\;\le\;
\mu_{P_N}\big(\mathbf{A}_N(t)\big)\,\|\boldsymbol{\eta}^{(N)}(t)\|_{P_N}
\;+\;
\|\mathbf{R}_N(t)\|_{P_N};
\end{equation}
in the case $\|\boldsymbol{\eta}^{(N)}(t)\|_{P_N}=0$, the inequality is trivially satisfied by continuity.  Thus the derivative obeys the upper bound for all $\boldsymbol{\eta}^{(N)}$.
Integrating \cref{eq:d-norm} with $\boldsymbol{\eta}^{(N)}(0)=0$ yields the variation-of-constants norm bound
\begin{equation}\label{eq:VoC-bound}
\|\boldsymbol{\eta}^{(N)}(t)\|_{P_N} \;\le\; \int_0^t \exp\!\Big(\int_s^t \mu_{P_N}\big(\mathbf{A}_N(\tau)\big)\, d\tau\Big)\, \|\mathbf{R}_N(s)\|_{P_N}\, ds,
\end{equation}
and, letting $\chi_P\ge 0$ be a constant such that $\mu_{P_N}(\mathbf{A}_N(\tau))\le \chi_P$ on $[0,t]$,
\begin{equation}\label{eq:VoC-bound-chi}
\|\boldsymbol{\eta}^{(N)}(t)\|_{P_N} \;\le\; \int_0^t e^{\chi_P (t-s)}\, \|\mathbf{R}_N(s)\|_{P_N}\, ds.
\end{equation}
Plugging \cref{eq:xi-bound2} into \cref{eq:VoC-bound-chi} yields
\[
\|\boldsymbol{\eta}^{(N)}(t)\|_{P_N}\;\le\; \int_0^t e^{\chi_P (t-s)}\, \|\mathbf{A}_{N+1}^N\|_P\, \|\mathbf{x}(s)\|_P^{\,N+1}\, ds.
\]
Bounding $\|\mathbf{x}(s)\|_P^{N+1}$ by $\sup_{0\le u\le t}\|\mathbf{x}(u)\|_P^{N+1}$ and integrating $e^{\chi_P (t-s)}$ yields
\[
\|\boldsymbol{\eta}^{(N)}(t)\|_{P_N} \;\le\; \frac{e^{\chi_P t}-1}{\chi_P}\, \|\mathbf{A}_{N+1}^N\|_P\, \sup_{0\le s\le t}\|\mathbf{x}(s)\|_P^{\,N+1}.
\]
Note that this term is positive because $\chi_P\ge 0$ by assumption.
Squaring and using \cref{eq:xP-bound} yields \cref{eq:lyap-path-main} and completes the proof.
\end{proof}

In the Lyapunov-based pathwise Carleman linearization truncation for finite-dimensional OU-driven quadratic systems, we assume a \emph{strict} Lyapunov inequality, \cref{eq:lyapunov-ineq},      
rather than the weaker nonexpansive condition $P \mathbf{F}_1 + \mathbf{F}_1^\dag P \le 0$ in \cref{lem:Carleman}. This choice encodes a quantitative dissipation condition $\mu_P(\mathbf{F}_1)<0$ for the linear part and yields sharper, uniform-in-time bounds on the Carleman residual under a stochastic forcing. Under \cref{eq:lyapunov-ineq-F1} and \eqref{eq:lyapunov-ineq}, one can take $|\mu_P(\mathbf{F}_1)|>\beta_P$ (recall that $\mu_P(\mathbf{F}_1) < 0$ and $\beta_P \ge 0$) and small $\gamma$, so that 
$\kappa_P < 0$ (\cref{eq:kappa_P}),
which gives exponential \emph{decay} of the homogeneous energy and the weighted energy estimate \cref{eq:r-int}
with $e^{\kappa (t-s)}$ suppressing contributions from the past. Plugging this into the Carleman residual $\|\boldsymbol{\eta}^{(N)}(t)\|^2 \lesssim \|\mathbf{A}_{N+1}^N\|^2 \|\mathbf{x}(t)\|^{2(N+1)}$ produces tighter constants in expectation and tail bounds (e.g., via Borell-TIS).
By contrast, if one only assumes $\mu_P(\mathbf{F}_1) \;\le\; 0$, then the Lyapunov growth rate $\kappa_P=2\beta_P+\gamma\ge \gamma>0$ and the homogeneous energy does not decay; the same derivation still holds, but the bound \cref{eq:r-int} inflates with $t$,
leading to weaker pathwise and probabilistic Carleman linearization truncation bounds. In particular, monotonicity $\|\mathbf{x}(t)\|_P\le \|\mathbf{x}(0)\|_P$ generally fails unless the quadratic drift is energy-neutral ($\beta_P=0$). Thus, the strict Lyapunov gap \cref{eq:lyapunov-ineq} is the preferred quantitative hypothesis for uniform-in-time control and sharper stochastic truncation bounds; the nonexpansive case remains valid but with degraded constants and no decay.

The matrix measure constant $\chi_P$ captures possible growth of the truncated Carleman propagator; in the contractive case $\chi_P=0$ the factor $\Phi_t$ reduces to $t$.
For quadratic $\mathbf{F}_2$ one may bound $\|\mathbf{A}_{N+1}^N\|_P\le N\,\|\mathbf{F}_2\|_P$. These lead to the bound,
\begin{equation}\label{eq:lyap-path-main-F2}
\|\boldsymbol{\eta}^{(N)}(t)\|_{P_N}^2 \;\le\; t^2\, N^2\, \|\mathbf{F}_2\|_P^2\, \big(a_t + b_t\, S_{t,P}^2\big)^{N+1},
\end{equation}
Let $G_t := a_t + b_t S_t^2$ be the energy gain factor in the Carleman linearization error bound (\cref{eq:lyap-path-main}), which is clearly from two sources of contributions:  
\begin{enumerate*}[label=(\roman*)]
\item $a_t$ is the deterministic ``homogeneous'' energy amplification from the initial state, governed by net dissipation ($\lambda_{\max}(S)$), nonlinearity ($\beta$), and horizon $t$;
\item $b_t S_t^2$ is the ``forced'' energy contribution driven by peak excursions of the OU load over $[0,t]$; stronger diffusion $\sigma$, weaker mean reversion $\theta$, and longer horizons yield larger peaks (larger $S_t$), hence larger $G_t$;
\item small $a_t$ and small $b_t S_t^2$ (i.e., $G_t<1$) correspond to strong damping, weak nonlinearity, short horizons, and moderate forcing. In that regime, the Carleman residual decays exponentially in $N$.
\end{enumerate*}

We now provide a probabilistic Carleman linearization error bound for the OU-driven quadratic system \cref{eq:quad-OU-ODE} by combining the pathwise Carleman linearization error bound of \cref{lem:lyap-path} with the Gaussian concentration of the OU supremum $S_{t,P}$ in \cref{lem:BTIS-OU}.

\begin{lemma}[Tail bound via Lyapunov inequality and Gaussian concentration]\label[lemma]{lem:lyap-tail}
Under the assumptions of Lemma~\ref{lem:lyap-path}, fix $t\in[0,T]$, $N\ge 1$, 
the Carleman linearization error satisfies
\begin{equation}\label{eq:lyap-tail-main}
\mathbb{P}\big(\|\boldsymbol{\eta}^{(N)}(t)\|_{P_N}^2 \ge \Delta_\eta \big)
\;\le\;
\exp\!\left(
-\frac{1}{2Q_{*,P}}
\left(
\sqrt{\frac{1}{b_t}}
\left[
\left(\frac{\Delta_\eta}{\Phi_t^2\, \|\mathbf{A}_{N+1}^N\|_P^2}\right)^{\frac{1}{N+1}}
- a_t
\right]^{1/2}
- \mathbb{E}[S_{t,P}]
\right)^2
\right),
\end{equation}
with $a_t,b_t,\Phi_t$ as in \cref{eq:params} and $Q_{*,P}$ is the proxy variance defined in \cref{eq:QstarP-def}. In particular, for sufficiently large $\Delta_\eta$,
\begin{equation}\label{eq:lyap-tail-weibull}
\mathbb{P}\big(\|\boldsymbol{\eta}^{(N)}(t)\|_{P_N}^2 \ge \Delta_\eta\big) \;\le\; \begin{cases}
  \exp\!\big(- c\, \Delta_\eta^{1/(N+1)}\big) &, \Delta_\eta\ge \Delta_{\eta,0}\\
    1&,\text{otherwise}
\end{cases}
\end{equation}
where $c=\frac{1}{16\, Q_{*,P}\, b_t}\,\Big(\Phi_t^2\, \|\mathbf{A}_{N+1}^N\|_P^2\Big)^{-\frac{1}{N+1}}$ and $\Delta_{\eta,0}= \Phi_t^2\, \|\mathbf{A}_{N+1}^N\|_P^2(\max\big\{\, 2a_t,\; a_t + 4 b_t m^2\,\big\})^{N+1}$. \end{lemma}

\begin{proof}
From \cref{eq:lyap-path-main}, the event $\{\|\boldsymbol{\eta}^{(N)}(t)\|_{P_N}^2 \ge \Delta_\eta\}$ implies
\begin{equation}
  \big(a_t + b_t S_{t,P}^2\big)^{N+1} \ge \frac{\Delta_\eta}{\Phi_t^2\, \|\mathbf{A}_{N+1}^N\|_P^2},
\end{equation}
hence
\begin{equation}\label{eq:StP-threshold}
S_{t,P} \;\ge\; \sqrt{\frac{1}{b_t}}\, \left[ \left(\frac{\Delta_\eta}{\Phi_t^2\, \|\mathbf{A}_{N+1}^N\|_P^2}\right)^{\frac{1}{N+1}} - a_t \right]^{1/2}.
\end{equation}
Since \cref{lem:BTIS-OU} guarantees the sub-Gaussianity of $S_{t,P}$, we apply the sub-Gaussian tail bound (Borell--TIS stated for completeness in \cref{prelim:BorellTIS}) for $S_{t,P}$ with variance proxy $Q_{*,P}$, and find that for any $r\ge 0$
\[
  \mathbb{P}\big(S_{t,P} \ge \mathbb{E}[S_{t,P}] + r\big) \le \exp\!\Big(-\frac{r^2}{2Q_{*,P}}\Big),
\]
and letting $r$ be the right-hand side of \cref{eq:StP-threshold} minus $\mathbb{E}[S_{t,P}]$ yields \cref{eq:lyap-tail-main}. 

To prove \cref{eq:lyap-tail-weibull}, we set
\[
K \;:=\; \Phi_t^2\, \|\mathbf{A}_{N+1}^N\|_P^2,\qquad m \;:=\; \mathbb{E}S_{t,P},
\]
and define the threshold function
\begin{equation}\label{eq:u-def}
u(\Delta_\eta) \;:=\; \sqrt{\frac{1}{b_t}}\;\Big( \big(\tfrac{\Delta_\eta}{K}\big)^{\frac{1}{N+1}} - a_t \Big)^{\frac12}.
\end{equation}
Then \cref{eq:lyap-tail-main} can be rewritten as
\begin{equation}\label{eq:tail-rewritten}
\mathbb{P}\!\big(\|\boldsymbol{\eta}^{(N)}(t)\|_{P_N}^2 \ge \Delta_\eta\big)
\;\le\;
\exp\!\left(
-\frac{1}{2Q_{*,P}}\, \big(u(\Delta_\eta) - m\big)^2
\right),
\qquad \Delta_\eta > K\, a_t^{N+1}.
\end{equation}
To pass from \cref{eq:tail-rewritten} to \cref{eq:lyap-tail-weibull}, fix 
\begin{equation}\label{eq:D-def}
D \;:=\; \max\big\{\, 2a_t,\; a_t + 4 b_t m^2\,\big\}
\end{equation}
and $\Delta_{\eta,0}:= K\, D^{N+1}$. For any $\Delta_\eta\ge \Delta_{\eta,0}$ one has
\begin{equation}
  (\Delta_\eta/K)^{1/(N+1)} \ge 2 a_t,
\end{equation} hence 
\begin{equation}\label{eq:Klb}
  (\Delta_\eta/K)^{1/(N+1)} - a_t \ge \tfrac12 (\Delta_\eta/K)^{1/(N+1)},
\end{equation}
and in turn 
\begin{equation}
  (\Delta_\eta/K)^{1/(N+1)} \ge a_t + 4 b_t m^2.
\end{equation}
This allows us to then say that
\begin{equation}\label{eq:uxi}
u(\Delta_\eta) \;=\; \sqrt{\tfrac{1}{b_t}}\,\Big( \big(\tfrac{\Delta_\eta}{K}\big)^{\frac{1}{N+1}} - a_t \Big)^{\frac12} \;\ge\; \sqrt{\tfrac{1}{b_t}}\,\Big(4 b_t m^2 \Big)^{\frac12} \;=\; 2m,
\end{equation}
and consequently 
\begin{equation}\label{eq:half-u}
u(\Delta_\eta) - m \;\ge\; \tfrac12\, u(\Delta_\eta).
\end{equation}
Combining \cref{eq:half-u} with~\cref{eq:uxi} and then finally using~\cref{eq:Klb} yields
\begin{equation}\label{eq:square-lower}
\big(u(\Delta_\eta) - m\big)^2 \;\ge\; \tfrac14\, u(\Delta_\eta)^2 
\;=\; \frac{1}{4 b_t}
\Big( \big(\tfrac{\Delta_\eta}{K}\big)^{\frac{1}{N+1}} - a_t \Big)
\;\ge\; \frac{1}{8 b_t}\, K^{-\frac{1}{N+1}}\, \Delta_\eta^{\frac{1}{N+1}}.
\end{equation}
Plugging \cref{eq:square-lower} into \cref{eq:tail-rewritten} gives, for all $\Delta_\eta\ge \Delta_{\eta,0}$,
\begin{equation}\label{eq:weibull-large-x}
\mathbb{P}\!\big(\|\boldsymbol{\eta}^{(N)}(t)\|_{P_N}^2 \ge \Delta_\eta\big)
\;\le\;
\exp\!\left(
-\frac{1}{2Q_{*,P}}\, \frac{1}{8 b_t}\, K^{-\frac{1}{N+1}}\, \Delta_\eta^{\frac{1}{N+1}}
\right)
\;=\;
\exp\!\big(- c\, \Delta_\eta^{1/(N+1)}\big),
\end{equation}
with
\begin{equation}\label{eq:c-def}
c \;:=\; \frac{1}{16\, Q_{*,P}\, b_t}\, K^{-\frac{1}{N+1}}
\;=\; \frac{1}{16\, Q_{*,P}\, b_t}\,\Big(\Phi_t^2\, \|\mathbf{A}_{N+1}^N\|_P^2\Big)^{-\frac{1}{N+1}}.
\end{equation}
Then for $\Delta_\eta\in(0,\Delta_{\eta,0})$ one has the trivial bound $\mathbb{P}(\cdot)\le 1$.
The emergence of the $\Delta_\eta^{1/(N+1)}$ exponent is intrinsic to the polynomial map $u\mapsto (L_t+u)^{N+1}$ linking $\|\boldsymbol{\eta}^{(N)}(t)\|_{P_N}^2$ to the OU-driven energy; consequently the tail is sub-Weibull of order $1/(N+1)$, and becomes heavier as $N$ increases.
\end{proof}

The intuition that a larger truncation order $N$ improves the probability of achieving a given error is only valid \emph{under a small-gain regime}, but it does not hold uniformly without additional assumptions. Let 
\begin{equation}\label{eq:Gt-def}
G_t := a_t + b_t S_t^2
\end{equation}
be the energy gain factor in the Carleman truncation error bound
$
\|\boldsymbol{\eta}^{(N)}(t)\|_{P_N}^2 \;\le\; \|\mathbf{A}_{N+1}^N\|_P^2 \big(G_t\big)^{N+1}.
$
If $G_t \le 1$ almost surely (e.g., due to strong dissipation, short time horizon, and bounded forcing), then
\[
\mathbb{E}\big[\|\boldsymbol{\eta}^{(N)}(t)\|_{P_N}^2\big] \;\le\; \|\mathbf{A}_{N+1}^N\|_P^2\, G_t^{N+1}
\; \text{and} \;
\mathbb{P}\big(\|\boldsymbol{\eta}^{(N)}(t)\|_{P_N}^2 \ge \Delta_\eta\big) \;\le\; \frac{\|\mathbf{A}_{N+1}^N\|_P^2\, G_t^{N+1}}{\Delta_\eta}
\]
by Markov's inequality. Both bounds decay exponentially in $N$, so $\mathbb{P}(\|\boldsymbol{\eta}^{(N)}(t)\|_{P_N}^2 \le \Delta_\eta)$ increases with $N$.
Mathematically, $G_t<1$ is driven by the following factors:
\begin{enumerate*}[label=(\roman*)]
\item Linear dissipation ($\lambda_{\max}(S)$):
  More negative $\lambda_{\max}(S)$ decreases $\kappa$, thus reduces both $a_t=e^{\kappa t}\|f(0)\|^2$ and $b_t=(e^{\kappa t}-1)/(\gamma\kappa)$ (in the dissipative regime). This directly shrinks $G_t$.
\item Nonlinearity strength ($\beta$):
  Larger $\beta$ increases $\kappa$, inflating $a_t$ and $b_t$, and hence $G_t$. Weaker quadratic drift (smaller $\beta$) favors $G_t<1$.
\item Time horizon ($t$):
  Longer $t$ increases $e^{\kappa t}$ and $(e^{\kappa t}-1)$, amplifying both $a_t$ and $b_t$. Short horizons help enforce $G_t<1$.
\item Young parameter ($\gamma$):
  The choice of $\gamma$ balances the trade-off between the linear drift term and the forcing. In the borderline case $\kappa\to 0$ (if we choose $\gamma=-(2\beta+2\lambda_{\max}(S))$ in the dissipative regime), $b_t\to t/\gamma$, so larger $|\gamma|$ reduces $b_t$. Away from $\kappa=0$, one may tune $\gamma$ to minimize $b_t$ for given $t$ and $\beta+\lambda_{\max}(S)$.
\item Forcing strength (OU covariance and mean reversion):
  The size of $S_t$ is controlled by the OU covariance radius $Q_*:=\sup_{s\le t}\|\mathrm{Cov}(\mathbf{F}_0(s))\|$. In the isotropic stationary case, $Q_*=\sigma^2/(2\theta)$, so larger diffusion $\sigma$ and smaller mean-reversion $\theta$ increase $S_t$ (and thus $G_t$). Borell-TIS gives
 $\mathbb{P}\big(S_t \le \mathbb{E}[S_t] + r\big) \;\ge\; 1 - \exp\!\Big(-\frac{r^2}{2Q_*}\Big)$,
  yielding a high-probability criterion: for any $\delta\in(0,1)$, with probability at least $1-\delta$,
  $G_t \;\le\; a_t + b_t \big(\mathbb{E}[S_t] + \sqrt{2Q_* \log(1/\delta)}\big)^2$.
  Thus $G_t<1$ with probability $\ge 1-\delta$ whenever
  $a_t + b_t \big(\mathbb{E}[S_t] + \sqrt{2Q_* \log(1/\delta)}\big)^2 \;<\; 1$.
\end{enumerate*}

Without $G_t<1$, two competing effects appear:
\begin{enumerate*}[label=(\roman*)]
    \item the residual scales like a degree-$(N+1)$ polynomial of $G_t$, so if $G_t>1$ with non-trivial probability, increasing $N$ can \emph{inflate} $\|\boldsymbol{\eta}^{(N)}(t)\|_{P_N}^2$;
    \item the tail bound (\cref{eq:lyap-tail-weibull}) is sub-Weibull of order $1/(N+1)$, reflecting polynomial dependence on the supremum $S_t$. For fixed $x$, the exponent $x^{1/(N+1)}$ \emph{decreases} with $N$, yielding a weaker (larger) upper bound as $N$ grows.
\end{enumerate*}
This reflects a general phenomenon for polynomials of sub-Gaussian variables: increasing the polynomial degree makes the guaranteed tail class heavier unless the argument is uniformly small.

\Cref{lem:lyap-tail} provides a tail bound for the Carleman linearization error in terms of the supremum $S_{t,P}$ of the $P$-norm of an arbitrary driving OU process $\mathbf{F}_0(t)$. We now particularize this result to the case $P=I$ and use explicit bound on the supremum of the Euclidean norm of $\mathbf{F}_0(t)$ (\cref{lemma:sup-OU}) to obtain explicit probabilistic bounds on the Carleman linearization error.

\begin{corollary}\label[corollary]{coro:P-eta}
Under the assumptions and definitions of \cref{lem:lyap-path} and \cref{lem:lyap-tail} and setting $P = I$, \cref{eq:lyap-tail-main} can be reduced to
\begin{equation}\label{eq:lyap-tail-main-explicit}
\mathbb{P}\big(\|\boldsymbol{\eta}^{(N)}(t)\|_{P_N}^2 \ge \Delta_\eta \big)
\;\le\;
\exp\!\left(
-\frac{1}{2 \sigma_*^2}
\left(
\sqrt{\frac{1}{b_t}}
\left[
\left(\frac{\Delta_\eta}{(t N \|\mathbf{F}_2\|)^2}\right)^{\frac{1}{N+1}}
- a_t
\right]^{1/2}
- \mathbb{E}[S_t]
\right)^2
\right),
\end{equation}
where  $\sigma_*^2$ and  $\mathbb{E}[S_t]$ are the covariance and expectation of the supremum of the OU process.\end{corollary}
\begin{proof}
  The proof follows directly from \cref{lem:lyap-tail} by setting $P=I$.
  In the contractive case with $P=I$, the pathwise Carleman linearization error in \cref{eq:lyap-path-main-F2} can be further simplified to
  \[
    \|\boldsymbol{\eta}^{(N)}(t)\|^2 \;\le\; t^2\, N^2\, \|\mathbf{F}_2\|^2\, \big(a_t + b_t\, S_{t}^2\big)^{N+1}, 
    \]
  from which the event $\{\|\boldsymbol{\eta}^{(N)}(t)\|^2 \ge \Delta_\eta\}$ implies
  \begin{equation}\label{eq:St-threshold}
  S_{t} \;\ge\; \sqrt{\frac{1}{b_t}}\, \left[ \left(\frac{\Delta_\eta}{t^2\, N^2\, \|\mathbf{F}_2\|^2}\right)^{\frac{1}{N+1}} - a_t \right]^{1/2}.
  \end{equation}
  This simplifies the variance proxy to $Q_{*,P=I} = Q_*=\sigma_*^2$ (\cref{eq:sigma-star-bound}) and the expectation $\mathbb{E}[S_{t,P=I}] = \mathbb{E}[S_t]$. Plugging \cref{eq:St-threshold} and \cref{eq:m-compact} into \cref{lemma:sup-OU} yields \cref{eq:lyap-tail-main-explicit} and completes the proof.
\end{proof}

The OU-driven quadratic system \cref{eq:quad-OU-ODE} often leads to a stationary state distribution. As we show in \cref{sec:steady-state}, when the initial state $\mathbf{x}(0)$ is drawn from the stationary distribution, the state process $\mathbf{x}(t)$ is stationary for all $t\ge 0$. In that case, we can further simplify the Carleman linearization error bound by removing the time dependence in the energy amplification factor in \cref{eq:Gt-def}. This leads to a sharper tail bound for the Carleman linearization error as follows.

\begin{theorem}\label{thm:deltax-tail}
Let $\boldsymbol{\eta}^{(N)}(t)=\mathbf{y}^{(N)}(t)-\hat{\mathbf{y}}^{(N)}(t)$ be the $N$-th order Carleman linearization error for the perturbation dynamics \cref{eq:perturbation_dynamics} with the block-lift of \cref{pre:cl}.  
Under the assumptions of Lyapunov stability of $\delta \mathbf{x}$ in \cref{lem:deltax-path} and \cref{lem:deltax-carleman-path} and the OU concentration for $S_{t,P}$ (Borell--TIS; cf. \cref{lem:lyap-tail}), for any $\Delta_\eta>0$,
\begin{equation}\label{eq:deltax-tail-main}
\mathbb{P}\big(\|\boldsymbol{\eta}^{(N)}(t)\|_{P_N}^2 \ge \Delta_\eta\big)
\;\le\;
\exp\!\left(
-\frac{1}{2Q_{*,P}}\,
\left(
\frac{\mu_P(\mathbf{F}_1)+\beta_P}{2 C_{P,B}\, t (N+1)}\,
\log\!\left(\frac{\Delta_\eta}{\Phi_t^2\, \|\mathbf{A}_{N+1}^N\|_P^2\, e^{(N+1)\kappa_0 t}\, \|\delta \mathbf{x}(0)\|_P^{2(N+1)}}\right) + \mathbb{E}[S_{t,P}]
\right)^{\!2}
\right),
\end{equation}
where 
$Q_{*,P}$ is the variance proxy of $S_{t,P}$ in \cref{lem:lyap-tail}, the Jacobian growth constant $C_{P, B}$ is defined in \cref{eq:CPB-def}, and the Lyapunov growth rate $\kappa_0$ is defined in \cref{lem:deltax-path}.
In particular, there exist constants $c_t,C_t>0$ given in \cref{eq:c_t-C-t} such that, for all sufficiently large $\Delta_\eta$,
\begin{equation}\label{eq:deltax-tail-logweibull}
\mathbb{P}\big(\|\boldsymbol{\eta}^{(N)}(t)\|_{P_N}^2 \ge \Delta_\eta\big)
\;\le\;
C_t\, \exp\!\Big(\,- c_t\, [\log \Delta_\eta]^2 \Big).
\end{equation}
\end{theorem}

\begin{proof}
We first combine the pathwise $\|\delta \mathbf{x}\|_P^2$ bound of \cref{eq:deltax-path-bound} in \cref{lem:deltax-path} with \cref{eq:xstar-pointwise} in \cref{lem:xstar-bound} to get an explicit bound for $\|\delta \mathbf{x}\|_P^2$.
From \cref{eq:xstar-pointwise},
\[
(-\mu_P(\mathbf{F}_1)-\beta_P)\,\|\xst(s)\|_P \;\le\; \|\mathbf{F}_0(s)\|_P,
\]
so, assuming $\mu_P(\mathbf{F}_1)+\beta_P<0$,
\[
\|\xst(s)\|_P \;\le\; \frac{1}{-\mu_P(\mathbf{F}_1)-\beta_P}\,\|\mathbf{F}_0(s)\|_P
\quad\Rightarrow\quad
\sup_{0\le s\le t}\|\xst(s)\|_P \;\le\; \frac{S_{t,P}}{-\mu_P(\mathbf{F}_1)-\beta_P}.
\]
Substituting this into \cref{eq:deltax-path-bound} yields
\[
\sup_{0\le s\le t}\|\delta \mathbf{x}(s)\|_P^2
\;\le\;
\exp\!\Big(\big[\kappa_0+\tfrac{2 C_{P,B}}{-\mu_P(\mathbf{F}_1)-\beta_P}\, S_{t,P}\big]\, t\Big)\, \|\delta \mathbf{x}(0)\|_P^2.
\]
Plugging this into pathwise Carleman linearization error for $\delta \mathbf{x}$, i.e.,  \cref{eq:deltax-eta-path} in \cref{lem:deltax-carleman-path} yields
\[
\|\boldsymbol{\eta}^{(N)}(t)\|_{P_N}^2 \;\le\; K_t\, \exp\!\big((N+1)\frac{2 C_{P,B}\, t}{-\mu_P(\mathbf{F}_1)-\beta_P}\, S_{t,P}\big),
\]
where
\begin{equation}
\label{eq:K_t}
K_t \;:=\; \Phi_t^2\, \|\mathbf{A}_{N+1}^N\|_P^2\, e^{(N+1)\kappa_0 t}\, \|\delta \mathbf{x}(0)\|_P^{2(N+1)},
\end{equation}
Hence, the event $\{\|\boldsymbol{\eta}^{(N)}(t)\|_{P_N}^2 \ge \Delta_\eta\}$ implies
\[
S_{t,P} \;\ge\; \frac{-\mu_P(\mathbf{F}_1)-\beta_P}{2 C_{P,B}\, t}\, \log\!\left(\frac{\Delta_\eta}{K_t}\right)^{\!\frac{1}{N+1}} \;=\; s_{\Delta_\eta}(t).
\]
Apply the sub-Gaussian tail of $S_{t,P}$ with variance proxy $Q_{*,P}$ (cf. \cref{lem:lyap-tail}) to obtain \cref{eq:deltax-tail-main}. For \cref{eq:deltax-tail-logweibull}, note that for $\Delta_\eta\ge \Delta_{\eta,0}:=K_t e^{(N+1)\frac{2 C_{P,B}\, t}{-\mu_P(\mathbf{F}_1)-\beta_P}(2\mathbb{E}S_{t,P})}$, one has $s_{\Delta_\eta}(t)-\mathbb{E}S_{t,P}\ge \frac{-\mu_P(\mathbf{F}_1)-\beta_P}{4 C_{P,B}\, t}\log(\Delta_\eta/K_t)^{1/(N+1)}$. Thus
\[
\mathbb{P}\big(\|\boldsymbol{\eta}^{(N)}(t)\|_{P_N}^2 \ge \Delta_\eta\big)
\le
\exp\!\left(
-\frac{1}{2Q_{*,P}}\, \frac{(-\mu_P(\mathbf{F}_1)-\beta_P)^2}{16C_{P,B}^2}\, [\log(\Delta_\eta/K_t)]^2
\right)
\le
\exp\!\big(-c_t\, [\log \Delta_\eta]^2\big),
\]
with
\begin{equation}\label{eq:c_t-C-t}
c_t \;:=\; \frac{(-\mu_P(\mathbf{F}_1)-\beta_P)^2}{32 Q_{*,P}\,C_{P,B}^2},\qquad
C_t \;:=\; \max\big\{\,1,\; \exp\big(c_t [\log \Delta_{\eta,0}]^2\big)\big\}.
\end{equation}
\Cref{eq:K_t} can be further determined by using \cref{lemma:RNt}, $\|\mathbf{A}_{N+1}^N\|_P \le N\, \|\mathbf{F}_2\|_P$:
\[
K_t \;\le\; \Phi_t^2\, N^2\, \|\mathbf{F}_2\|_P^2\, e^{(N+1)\kappa_0 t}\, \|\delta \mathbf{x}(0)\|_P^{2(N+1)}.
\]
\end{proof}

\Cref{thm:deltax-tail} incorporates pathwise growth through $S_{t,P}$ over $[0,t]$ and captures cumulative OU fluctuations over time. It is preferable when the randomness accumulated along the trajectory dominates the initial uncertainty. When $\|\delta \mathbf{x}(0)\|_P$ is small (e.g., near a stationary state), we provide \cref{cor:eta-initial-tail}, which yields a tighter bound since it avoids the potentially large fluctuations of $S_{t,P}$.

\begin{corollary}\label[corollary]{cor:eta-initial-tail}
Consider the Carleman linearization error $\boldsymbol{\eta}$ for the perturbation dynamics $\delta \mathbf{x}$ in \cref{lemma:carleman-delta-f}. 
Then for any $j\in[N]$ and any $\Delta_\eta>0$,
\begin{equation}\label{eq:eta-j-tail}
\mathbb{P}\!\big(\|\boldsymbol{\eta}_j(t)\|_P^2 \ge \Delta_\eta\big)
\;\le\;
\exp\!\left(
-\frac{1}{2Q_{0,P}}\,
\Bigg(
\Delta_\eta^{\frac{1}{2(N+1)}}\left(\frac{|\mu_P(\mathbf{F}_1)|}{\|\mathbf{F}_2\|_P}\right)^{\!\frac{N+1-j}{N+1}} - \mathbb{E}\|\delta \mathbf{x}(0)\|_P
\Bigg)^{\!2}
\right)
\end{equation}
assuming the initial perturbation norm is sub-Gaussian:
\begin{equation}\label{eq:init-subG}
\mathbb{P}\!\left(\|\delta \mathbf{x}(0)\|_P \ge \mathbb{E}\|\delta \mathbf{x}(0)\|_P + s\right) \le \exp\!\left(-\frac{s^2}{2Q_{0,P}}\right)\quad\forall s\ge 0.
\end{equation}
For $j=1$, using \cref{eq:eta1-exact},
\begin{equation}\label{eq:eta-1-tail}
\mathbb{P}\!\big(\|\boldsymbol{\eta}_1(t)\|_P^2 \ge \Delta_\eta\big)
\;\le\;
\exp\!\left(
-\frac{1}{2Q_{0,P}}\,
\Bigg(
\sqrt{\Delta_\eta}\,
\left(\frac{|\mu_P(\mathbf{F}_1)|}{\|\mathbf{F}_2\|_P}\right)^{\!N}
\frac{1}{\bigl|1-e^{\mu_P(\mathbf{F}_1)\, t}\bigr|^{N}}
- \mathbb{E}\|\delta \mathbf{x}(0)\|_P
\Bigg)^{\!2}
\right).
\end{equation}
\end{corollary}
\begin{proof}
From \cref{eq:etaj-bound}, squaring gives
\[
\|\boldsymbol{\eta}_j(t)\|_P^2 \le
\left(\frac{\|\mathbf{F}_2\|_P}{|\mu_P(\mathbf{F}_1)|}\right)^{\!2(N+1-j)}
\|\delta \mathbf{x}(0)\|_P^{\,2(N+1)}.
\]
Thus the event $\{\|\boldsymbol{\eta}_j(t)\|_P^2 \ge \Delta_\eta\}$ implies
\[
\|\delta \mathbf{x}(0)\|_P \;\ge\;
\Delta_\eta^{\frac{1}{2(N+1)}}\left(\frac{|\mu_P(\mathbf{F}_1)|}{\|\mathbf{F}_2\|_P}\right)^{\!\frac{N+1-j}{N+1}}.
\]
Apply the sub-Gaussian tail \cref{eq:init-subG} with
\[
s \;=\; \Delta_\eta^{\frac{1}{2(N+1)}}\left(\frac{|\mu_P(\mathbf{F}_1)|}{\|\mathbf{F}_2\|_P}\right)^{\!\frac{N+1-j}{N+1}} - \mathbb{E}\|\delta \mathbf{x}(0)\|_P \;\ge 0
\]
(for sufficiently large $\Delta_\eta$) to obtain \cref{eq:eta-j-tail}. For $j=1$, from \cref{eq:eta1-exact},
\[
\|\boldsymbol{\eta}_1(t)\|_P^2 \le
\left(\frac{\|\mathbf{F}_2\|_P}{|\mu_P(\mathbf{F}_1)|}\right)^{\!2N}
\bigl|1-e^{\mu_P(\mathbf{F}_1)\, t}\bigr|^{2N}
\|\delta \mathbf{x}(0)\|_P^{2},
\]
so the event $\{\|\boldsymbol{\eta}_1(t)\|_P^2 \ge \Delta_\eta\}$ implies
\[
\|\delta \mathbf{x}(0)\|_P \;\ge\;
\sqrt{\Delta_\eta}\,
\left(\frac{|\mu_P(\mathbf{F}_1)|}{\|\mathbf{F}_2\|_P}\right)^{\!N}
\frac{1}{\bigl|1-e^{\mu_P(\mathbf{F}_1)\, t}\bigr|^{N}},
\]
yielding \cref{eq:eta-1-tail} after applying \cref{eq:init-subG} with this choice of $s$. 
\end{proof}

\section{Stochastic linear combination of Hamiltonian simulations}\label{sec:slchs}

We are now ready to solve \cref{eq:LODE} with the multivariate OU process $\mathbf{F}_0 \in \R^n$ (\cref{eq:general_ou}) as the inhomogeneous term, which yields a time-dependent inhomogeneous stochastic equation. The stochasticity requires a quantum algorithm that can handle stochastic dynamics. The goal of this section is thus to generalize the recently developed Linear Combination of Hamiltonian Simulation (LCHS) algorithm~\citep{an_quantum_2025} to tackle the stochasticity. We choose LCHS for our approach given that it achieves near-optimal parameter-wise dependence on all parameters for quantum simulations of linear non-unitary dynamics \citep{an_quantum_2025, low_optimal_2025}, and has favorable constant-factor scaling over competing approaches \cite{pocrnic2025constant}. The LCHS algorithm solves general time-dependent inhomogeneous deterministic ODEs by mapping ODEs to a linear combination of Schr\"odinger equations. The Hamiltonian simulation problem, i.e.,
simulating the Schr\"odinger equation on quantum computers, may be the first
and most natural application to achieve quantum advantage. The Schr\"odinger
equation is a special case of ODE with an anti-Hermitian coefficient matrix $i\mathbf{A}(t)$ and zero inhomogeneous term $\mathbf{b}(t) \equiv 0$. It is
natural to map ODEs to Hamiltonian simulations, and doing so allows one to leverage highly optimized existing algorithms. In this section, we generalize
the LCHS to SDEs to solve \cref{eq:LODE}. This requires a generalized stability analysis of the Carleman-linearized SDE for LCHS and time integration for stochastic systems, since the OU process is continuous everywhere but \textit{nowhere} differentiable. Specifically, \cref{sec:SLCHS} and \cref{sec:inhomo_LCHS} provide probabilistic bounds for the homogeneous term of LCHS due to the stochastic Carleman matrix and for the inhomogeneous term due to both the stochastic Carleman matrix and forcing term, respectively. Finally, we develop a Monte Carlo (MC) truncated Dyson series for each Hamiltonian simulation of LCHS in \cref{sec:MC-TDS} to tackle the stochasticity.

In order for LCHS to be applicable, we need to show that the Carleman matrix $\mathbf{A}$ is stable. Since $\mathbf{A}$ is stochastic, the condition will only hold with some probability. In order to show this we will make use of the block diagonal Gershgorin theorem from \cite{vandersluis1979gershgorin}.

\begin{theorem}[\cite{vandersluis1979gershgorin}, Thm.~1.1]\label{thm:block_gershgorin}
  For a partitioned matrix $\mathbf{A}$ with partitions $\mathbf{A}_{i,j}$, $i,j\in[1,N]$ and induced norm $\norm{\cdot}$, all eigenvalues of $\mathbf{A}$ are contained in the set $\bigcup_i G_i$ where $G_i$ is the set of all $\lambda \in \mathbb{C}$ satisfying 
    \begin{align}
    \norm{(\mathbf{A}_{i,i}-\lambda \mathbb{I})^{-1}}^{-1} \leq \sum_{j\neq i}^N \norm{\mathbf{A}_{i,j}},
    \end{align}
where $\norm{(\mathbf{A}_{i,i}-\lambda \mathbb{I})^{-1}}^{-1}$ is defined to be 0 if $\lambda$ is an eigenvalue of $\mathbf{A}_{i,i}$.
\end{theorem}

Since $\mathbf{A}$ is block tri-diagonal, the inequality defining the available eigenvalues for the $j$th row of $\mathbf{A}$ for $j \neq 1,N$ can be written as
\begin{align}
  \norm{(\mathbf{A}_j^j-\lambda \mathbb{I})^{-1}}^{-1} \leq \norm{\mathbf{A}_{j+1}^{j}} + \norm{\mathbf{A}_{j-1}^{j}}.
\end{align}
Using the triangle inequality we can upper bound $\norm{\mathbf{A}_{j+1}^{j}}$ and $\norm{\mathbf{A}_{j-1}^{j}}$:

\begin{align}
  \norm{\mathbf{A}_{j+1}^{j}} &\leq j \norm{\mathbf{F}_2}\\
  \norm{\mathbf{A}_{j-1}^{j}} &\leq j \norm{\mathbf{F}_0(t)},
\end{align}
where $\norm{\mathbf{F}_0(t)}$ is the vector norm of $\mathbf{F}_0(t)$ that induced the operator norm defined in the theorem. Since we assume $\mathbf{F}_1$ to be negative semi-definite, $\mathbf{A}_j^j$ is also negative semi-definite and therefore normal. For spectral decomposition $\mathbf{A}_j^j=\sum_k \lambda_k \textbf{e}_k\textbf{e}_k^*$, and assuming that $\lambda$ is not an eigenvalue of $\mathbf{A}_j^j$,
\begin{align}
    (\mathbf{A}_j^j-\lambda \mathbb{I})^{-1} &= \sum_k \frac{1}{\lambda_k-\lambda}\textbf{e}_k\textbf{e}_k^*.
\end{align}
So that
\begin{align}
    \norm{(\mathbf{A}_j^j-\lambda \mathbb{I})^{-1}}^{-1}&= \frac{1}{\max_k\{\frac{1}{|\lambda_k-\lambda|}\}} \\
    &= \min_k\{|\lambda_k-\lambda|\}.
\end{align}
Then the inequality given by \cref{thm:block_gershgorin} becomes
\begin{align}
    \min_k\{|\lambda_k-\lambda|\} \leq j(\norm{\mathbf{F}_2}+\norm{\mathbf{F}_0(t)}).
\end{align}
Since the eigenvalues of $\mathbf{A}$ are guaranteed to be in the intervals given by $\lambda$ values satisfying this inequality for all $j$, and $\lambda_k$ are negative, if it is satisfied that the Gershgorin radius (or our upper bound on it) is less than the distance between the largest eigenvalue $\lambda_N$ and $0$, then it is also satisfied that $\lambda < 0$. This can be expressed succinctly as
\begin{equation}\label{eq:stabilitycondition1}
  |\lambda_N| > j (\norm{\mathbf{F}_2}+\norm{\mathbf{F}_0(t)})\quad \forall j\quad \Rightarrow \quad\mathbf{A}\prec 0.
\end{equation}
Since $\mathbf{A}_j^j$ is constructed out of $\mathbf{F}_1$, this condition can be trivially extended to a condition over the eigenvalues of $\mathbf{F}_1$:

\begin{align}
    |\lambda_N| &= \min_j |\lambda_j|\\
    &= j \min |\lambda(\mathbf{F}_1)|.
\end{align}
Thus the condition in~\cref{eq:stabilitycondition1} becomes
\begin{align}
    \min |\lambda(\mathbf{F}_1)| > \norm{\mathbf{F}_2}+\norm{\mathbf{F}_0(t)}
\end{align}
We now use the bounds on $\norm{\mathbf{F}_0(t)}$ derived in \cref{eq:Prob-F0} to provide expectation and high-probability bounds for the stability condition in \cref{thm:block_gershgorin}.

\begin{corollary}\label[corollary]{cor:margin_F1_F2_F0}
Let
$
\Delta\ :=\ \min_{i}\big|\lambda_i(\mathbf{F}_1)\big|\ -\ \|\mathbf{F}_2\|
$
and assume $\Delta>0$. Consider $\mathbf{F}_0(t)$ with mean-reversion rate $\lambda_{\min}>0$ (the minimal eigenvalue of the symmetric part of $\boldsymbol{\Theta}$) and diffusion matrix $\boldsymbol{\Sigma}$. Then, for any fixed $t\in[0,T]$, 
\begin{equation}\label{eq:prob_margin_bound}
\mathbb{P}\!\left(\ \Delta> \|\mathbf{F}_0(t)\|\ \right)
\ \ge\ 1\ -\ \frac{1-e^{-2\lambda_{\min} t}}{2\,\lambda_{\min}\,\Delta^{2}}\ \|\boldsymbol{\Sigma}\|_{F}^{2}.
\end{equation}
Further we have that $\mathbb{E}(\Delta -\|\mathbf{F}_0(t)\|)\ge 0$ if
\begin{equation}\label{eq:delta_requirement}
\Delta\ \ge\ \|\boldsymbol{\Sigma}\|_{F}\ \sqrt{\frac{1-e^{-2\lambda_{\min} t}}{2\,\lambda_{\min}\,\delta}}\;.
\end{equation}

\end{corollary}

\begin{proof}
Taking $x^* = \Delta$ in \cref{eq:Prob-F0}, yields 
\begin{equation}
\mathbb{P}\!\left(\  \|\mathbf{F}_0(t)\| <\Delta\ \right)
\ \ge\ 1\ -\ \frac{1-e^{-2\lambda_{\min} t}}{2\,\lambda_{\min}\,\Delta^{2}}\ \|\boldsymbol{\Sigma}\|_{F}^{2}.
\end{equation}
In turn we have that 
\begin{equation}
    \mathbb{E}(\|\mathbf{F}_0(t)\|) \ge  \Delta {\mathbb{P}(\Delta\ge \|\mathbf{F}_0(t)\|)}\ge {\Delta-\|\boldsymbol{\Sigma}\|_F^2 \sqrt{\frac{1-e^{-2\lambda_{\min}t}}{2\lambda_{\min}}}}
\end{equation}
Consequently, the expectation value of $\Delta - \|\mathbf{F}_0(t)\|$ obeys from the reverse triangle inequality
\begin{equation}\label{eq:expected_positive_margin}
\mathbb{E}\!\Big[\ \min_{i}|\lambda_i(\mathbf{F}_1)|\ -\ \|\mathbf{F}_2\|\ -\ \|\mathbf{F}_0(t)\|\ \Big]
\ \ge \Delta - \|\mathbb{E}(\mathbf{F}_0(t))\| \ge  \|\boldsymbol{\Sigma}\|_{F}\,\sqrt{\frac{1-e^{-2\lambda_{\min} t}}{2\,\lambda_{\min}}}
\end{equation}
and is non-negative if
\begin{equation}\label{eq:expected_condition}
\Delta\ >\ \|\boldsymbol{\Sigma}\|_{F}\,\sqrt{\frac{1-e^{-2\lambda_{\min} t}}{2\,\lambda_{\min}}}.
\end{equation}
\end{proof}

\subsection{LCHS for linear SDEs}
\label{sec:SLCHS}

With the stability guarantee of the Carleman linearized SDE in hand, we are now ready to solve \cref{eq:LODE} with the multivariate-OU process $\mathbf{F}_0 \in \R^n$ (\cref{eq:general_ou}) as the inhomogeneous term, which is a time-dependent inhomogeneous stochastic equation. We first generalize LCHS for ODEs developed by \citet{an_quantum_2025} to SDEs, where both the coefficient matrix $\mathbf{A}$ and inhomogeneous term $\mathbf{b}$ in \cref{eq:LODE} are stochastic. 

A solution of the truncated Carleman-linearized system, \cref{eq:LODE}, can be represented as  
\begin{equation}\label{eqn:ODE_solu}
    \mathbf{y}(t) = \mathcal{T}e^{-\int_0^t \mathbf{A}(s) \ud s} \mathbf{y}_0 + \int_0^t \mathcal{T}e^{-\int_s^t \mathbf{A}(s') \ud s'} \mathbf{b}(s) \ud s,
\end{equation}
where $\mathcal{T}$ denotes the time-ordering operator. 
Quantum algorithms for~\cref{eq:LODE} aim at preparing an $\eps$-approximation of the quantum state $\ket{\mathbf{y}(T)} = \mathbf{y}(T)/\|\mathbf{y}(T)\|$, encoding the normalized solution in its amplitudes. 
One of the most natural ways is to map \cref{eq:LODE} to 
the Schr\"odinger equation (aka the Hamiltonian simulation problem) as it can readily be efficiently solved on quantum computers. The LCHS starts with Cartesian decomposition \citep{an_quantum_2025} of the stochastic Carleman coefficient matrix $A$:
\begin{align}
  \mathbf{A}(t)=L(t)+iH(t),
\end{align} where the real and imaginary parts of the Hermitian matrices are
\begin{align}\label{eq:LM}
  L(t)=\frac{\mathbf{A}(t)+\mathbf{A}(t)^\dagger}{2}, H(t)=\frac{\mathbf{A}(t)-\mathbf{A}(t)^\dagger}{2i},
\end{align}
assuming that $L(t)$ is positive semi-definite that guarantees the asymptotic stability of the dynamics because $\|\mathcal{T}e^{-\int \mathbf{A}(s) ds}\|\in \mathcal{O}(1)$. Under this stability assumption, the non-unitary evolution operator of \cref{eqn:ODE_solu} can be expressed as a linear combination of Hamiltonian simulations:      
\begin{equation}\label{eqn:LCHS_improved}
  \mathcal{T} e^{-\int_0^t \mathbf{A}(s) \ud s} = \int_{\mathbb{R}} \frac{f(k)}{1-ik} \mathcal{T} e^{-i \int_0^t (kL(s)+H(s)) \ud s} \ud k,
\end{equation}
    where $f(k)$ is a kernel function with the form:
\begin{align}\label{eqn:kernel_intro}
  f(z) = & \frac{1}{2\pi e^{-2^\beta} e^{(1+iz)^{\beta}} }, \quad \beta \in (0,1) \\
  = & \frac{1}{C_\beta e^{(1+iz)^\beta}}
\end{align}
that is near-optimal for LCHS\@.
We discretize~\cref{eqn:LCHS_improved} into a discrete sum of stochastic unitaries with a specific choice of kernel in~\cref{eqn:kernel_intro} by truncating the interval over the entire real line into a finite interval $[-K,K]$ and then use the composite Gaussian quadrature to discretize the interval. 
Specifically, let 
\begin{equation}\label{eqn:quadrature}
  g(k) = \frac{1}{C_{\beta} (1-ik) e^{(1+ik)^{\beta}} }, \quad U(T,k) = \mathcal{T} e^{-i \int_0^T (kL(s)+H(s)) \ud s}, 
\end{equation}
then \cref{eqn:LCHS_improved} can be discretized to 
\begin{equation}\label{eqn:LCHS_LCU_composite}
\begin{split}
  \mathcal{T} e^{-\int_0^T \mathbf{A}(s) \ud s} = \int_{\mathbb{R}} g(k) U(T,k) \ud k &\approx \int_{-K}^{K} g(k) U(T,k) \ud k, \\
\end{split}
\end{equation}
where $h_1$ is the step size used in the composite quadrature rule. To bound the truncation error in $k$-space and the quadrature error for the Riemann summation, we observe the fact that even though the unitary $U(T, k)$ is stochastic, it is evaluated at the final evolution time as a single realization or a given sampling path $\omega$. We state this below.

\begin{lemma}\label{lemma:quadrature} Let $\mathbf{A}(t) =L(t) + iH(t)$ be an integrable time-dependent operator, for Hermitian $L$ and $H$, acting on a finite-dimensional Hilbert space and let $\beta\in (0,1)$ be a constant that defines the kernel function $g(k)$ given in~\cref{eqn:quadrature}.  We then have 
\begin{enumerate}
  \item For the homogeneous LCHS truncation error:
    \begin{equation}\label{eq:Prob_trunc}
\Prob\left( \left\| \int_{\mathbb{R}} g(k) U(T,k,\omega) \ud k - \int_{-K}^{K} g(k) U(T,k,\omega) \ud k \right\| \leq \frac{2^{B+1} B!}{C_{\beta} (\cos(\beta\pi/2))^{B}} \frac{e^{-\frac{1}{2}K^{\beta} \cos(\beta\pi/2)}}{K} \right) = 1
\end{equation}
where $B = \lceil 1/\beta \rceil$.
\item For $\eps > 0$, in order to bound the error by $\eps$, it suffices to choose the truncation parameter of kernel space 
  \begin{equation}\label{eq:K-log-epsilon}
        K = \mathcal{O}\left( \left(\log\left(\frac{1}{\eps}\right)\right)^{1/\beta} \right). 
    \end{equation}
\end{enumerate}
\end{lemma}

\begin{proof}

      Since $U(T,k,\omega)$ is unitary for each $\omega$, we have $\|U(T,k,\omega)\| = 1$ deterministically. Therefore:
\begin{align}
\left\| \int_{|k|>K} g(k) U(T,k,\omega) \ud k \right\| &\leq \int_{|k|>K} |g(k)| \|U(T,k,\omega)\| \ud k \\
						       &= \int_{|k|>K} |g(k)| \ud k \label{eq:gk} \\
						       &\leq \frac{2^{B+1} B!}{C_{\beta} (\cos(\beta\pi/2))^{B}} \frac{e^{-\frac{1}{2}K^{\beta} \cos(\beta\pi/2)}}{K} \label{eq:gKK}
\end{align}
We apply \citet[Lemma 9]{an_quantum_2025} directly from \cref{eq:gk} to \cref{eq:gKK}.
This bound holds deterministically for every realization $\omega$, so the probability is 1.
\end{proof}

We note that it is possible to tighten the above bound on the truncation parameter $K$ using a Lambert-W function, as is done in Lemma 3 of \citet{pocrnic2025constant}; however, the asymptotics remain the same, which suffices for our purposes. The quadrature error follows the deterministic case other than the bound for $\max_t \|L(t)\|$, which requires a stochastic argument. Specifically, we provide probabilistic tail bounds for the term $\max_t \|L(t)\|$ in \citet[Equations 191--204]{an_quantum_2025}.

\begin{lemma}
For any $\delta_{\mathrm{quad}} \in (0,1)$, the following probabilistic bounds for the quadrature error holds:
\begin{equation}\label{eq:Prob-quadrature}
\Prob\left(\left\| \int_{-K}^{K} g(k) U(T,k, \omega) \ud k - \sum_{m = -K/h_1}^{K/h_1-1} \sum_{q=0}^{Q-1} c_{q,m} U(T,k_{q,m}, \omega) \right\| \leq \frac{8}{3C_{\beta}} K h_1^{2Q}  (eT\Lambda_{\mathrm{threshold}}/2)^{2Q}
  \right)  \geq 1 - \delta_{\mathrm{quad}},
\end{equation}
where $\Lambda_{\mathrm{threshold}}$ is defined as
$$
\Lambda_{\mathrm{threshold}}\ :=\ N\Big(\|\mathbf{F}_1\| + \|\mathbf{F}_2\| + \mathbb{E}\Big[\sup_{t\in[0,T]}\|\mathbf{F}_0(t)\|\Big] + \sqrt{2\,\sigma_*^2\,\log(1/\delta_{\mathrm{quad}})}\Big)
$$

To bound the error by $\eps>0$ with probability at least $1-\delta_{\mathrm{quad}}$, it suffices to choose
\begin{align}\label{eq:homo-h1}
h_1 = \frac{1}{eT\Lambda_{\mathrm{threshold}}},
\end{align}
\begin{align}\label{eq:homo-Q}
        Q = \left\lceil \frac{1}{\log 4} \log\left( \frac{8}{3C_{\beta}} \frac{K}{\eps} \right) \right\rceil = \mathcal{O}\left( \log\left( \frac{K}{\eps} \right) \right) = \mathcal{O}\left( \log\left( \frac{1}{\eps} \right) \right), 
\end{align}
and the overall number of unitaries in the summation formula for discretizing the homogeneous term in \cref{eqn:LCHS_LCU_composite} is
\begin{equation}\label{eq:homo-M}
    N_U = \frac{2KQ}{h_1} = \mathcal{O}\left( T N \left(\|\mathbf{F}_1\| + \|\mathbf{F}_2\| + \mathbb{E}\Big[\sup_{t\in[0,T]}\|\mathbf{F}_0(t)\|\Big] + \sqrt{2\,\sigma_*^2\,\log(1/\delta_{\mathrm{quad}})}\right) \left(\log\left(\frac{1}{\eps}\right)\right)^{1+1/\beta} \right). 
\end{equation}
\end{lemma}

\begin{proof}
  For each time interval $[mh_1,(m+1)h_1]$ for fixed $\omega$, we can bound the quadrature error by directly use the deterministic analysis of \citet[Equation 177]{an_quantum_2025}:
\begin{equation}\label{eqn:quadrature_error_proof_eq1}
\left\| \int_{mh_1}^{(m+1)h_1} g(k) U(T,k,\omega) \ud k - \sum_{q=0}^{Q-1} c_{q,m} U(T,k_{q,m},\omega) \right\| \leq \frac{(Q!)^4 h_1^{2Q+1}}{(2Q+1)((2Q)!)^3} \|(gU)^{(2Q)}(\omega)\|
\end{equation}
Since the $p$-th order derivative of the function $g$ and $U$ are respect to the kernel space $k$, the deterministic results of $\d U^{(p)}/\d t$ and $\d g^{(p)}/\d t$ are applicable to the stochastic case.
Bounding \cref{eqn:quadrature_error_proof_eq1} becomes a problem of bounding $\|(gU)^{(2Q)}(\omega)\|$, which is given by \citet[Equation 196]{an_quantum_2025}
\begin{equation}\label{eq:gU2Q}
\|(gU)^{(2Q)}(\omega)\| \leq \frac{4}{3C_{\beta}} (2Q)!(2Q+1) (eT)^{2Q} \Lambda(\omega)^{2Q}
\end{equation}
We now use sub-Gaussian tail bounds of $\Lambda(\omega)$ to bound \cref{eq:gU2Q}. By \cref{lemma:Lambda-subG}, for any $\delta_{\mathrm{quad}} \in (0,1)$, choosing
\begin{equation}\label{eq:lambda-threshold}
\Lambda_{\mathrm{threshold}}\ :=\ N\Big(\|\mathbf{F}_1\| + \|\mathbf{F}_2\| + \mathbb{E}\Big[\sup_{t\in[0,T]}\|\mathbf{F}_0(t)\|\Big] + \sqrt{2\,\sigma_*^2\,\log(1/\delta_{\mathrm{quad}})}\Big)
\end{equation}
yields
\begin{equation}\label{eq:PgU}
\Prob\left(\|(gU)^{(2Q)}(\omega)\| \leq \frac{4}{3C_{\beta}} (2Q)!(2Q+1) (eT)^{2Q} \Lambda_{\mathrm{threshold}}^{2Q} \right) \geq 1 - \delta_{\mathrm{quad}}
\end{equation}

Plugging \cref{eq:PgU} into \cref{eqn:quadrature_error_proof_eq1} and further simplify using \citet[Equations 199--200]{an_quantum_2025}, we have
\begin{equation}\label{eq:Prob-gU}
\Prob\left(\left\| \int_{mh_1}^{(m+1)h_1} g(k) U(T,k,\omega) \ud k - \sum_{q=0}^{Q-1} c_{q,m} U(T,k_{q,m},\omega) \right\| \leq \frac{4}{3C_{\beta}} h_1^{2Q+1}  (eT\Lambda_{\mathrm{threshold}}/2)^{2Q}
  \right)  \geq 1 - \delta_{\mathrm{quad}}
\end{equation}

By summing over all the short intervals following \citet[Equation 201]{an_quantum_2025}, we have 
\begin{align}\label{eq:m_K}
        & \left\| \int_{-K}^{K} g(k) U(T,k, \omega) \ud k - \sum_{m = -K/h_1}^{K/h_1-1} \sum_{q=0}^{Q-1} c_{q,m} U(T,k_{q,m}, \omega) \right\| \nonumber\\
        &\quad \leq \sum_{m = -K/h_1}^{K/h_1-1} \left\| \int_{mh_1}^{(m+1)h_1} g(k) U(T,k, \omega) \ud k -  \sum_{q=0}^{Q-1} c_{q,m} U(T,k_{q,m}, \omega) \right\|    
    \end{align}
Combining \cref{eq:m_K} and \cref{eq:Prob-gU} leads to \cref{eq:Prob-quadrature} and completes the proof with the stated probability bound.
We can further bound \cref{eq:Prob-quadrature} as:
    \begin{equation}\label{eq:Prob-quadrature-final}
      \Prob\left(\left\| \int_{-K}^{K} g(k) U(T,k, \omega) \ud k - \sum_{m = -K/h_1}^{K/h_1-1} \sum_{q=0}^{Q-1} c_{q,m} U(T,k_{q,m}, \omega) \right\| \leq \frac{8}{3C_{\beta}} K \frac{1}{2^{2Q}}
  \right)  \geq 1 - \delta_{\mathrm{quad}}
\end{equation}
by choosing \cref{eq:homo-h1}.
It suffices to choose \cref{eq:homo-Q}
to bound \cref{eq:Prob-quadrature-final} by $\eps$ due to \cref{eq:K-log-epsilon}.
Finally, \cref{eq:homo-M} is reached by combining \cref{eq:K-log-epsilon}, \cref{eq:lambda-threshold}, \cref{eq:homo-h1}, and \cref{eq:homo-Q}.
\end{proof}

\begin{corollary}[Complete probabilistic error bound - homogeneous]
For any $\eps_{\mathrm{trunc}} = \eps_{\mathrm{quad}} = \eps/2 > 0$,
\begin{equation}\label{eq:Prob_complete_homo}
\Prob\left( \left\| \T e^{\int_0^T A(s,\omega) ds} - \sum_{m = -K/h_1}^{K/h_1-1} \sum_{q=0}^{Q-1} c_{q,m} U(T,k_{q,m},\omega) \right\| \leq \eps \right) \geq 1 - \delta
\end{equation}
where the deterministic $\eps_{\mathrm{trunc}}$ and probabilistic $\eps_{\mathrm{quad}}$ are given by \cref{eq:Prob_trunc} and \cref{eq:Prob-quadrature}, respectively.
\end{corollary}
\begin{proof}
  \Cref{eq:Prob_complete_homo} follows from the triangle inequality and the union bound, given that the bound on $\eps_{\mathrm{trunc}}$ is deterministic.
\end{proof}

\subsection{Stochastic inhomogeneous term}
\label{sec:inhomo_LCHS}
Our aim now is to consider the effect of having a stochastic inhomogeneous term in our differential equation. We address this by using the resolvent expansion proposed by
the LCHS formalism on the inhomogeneous term in \cref{eqn:ODE_solu}.  This yields
\begin{equation}\label{eq:b-LCHS}
    \int_0^T \mathcal{T}e^{-\int_s^T A(s')\ud s'} b(s) \ud s = \int_0^T \int_{\mathbb{R}} \frac{f(k)}{1-ik} U(T,s,k)  b(s) \ud k \ud s, 
\end{equation}
where 
\begin{equation}
    U(T,s,k) = \mathcal{T} e^{-i \int_s^T (kL(s')+H(s')) \ud s'}. 
\end{equation}
By discretizing the variable $k$, we obtain 
\begin{align}
  \int_0^T \mathcal{T}e^{-\int_s^T A(s')\ud s'} b(s) \ud s & \approx \int_0^T \sum_{m_1 = -K/h_1}^{K/h_1-1} \sum_{q_1=0}^{Q_1-1} c_{q_1,m_1} U(T,s,k_{q_1,m_1})  b(s)  \ud s.
  \label{eq:lchs-b}
\end{align}
Since both $A(s^\prime)$ and $b(s)$ in \cref{eq:b-LCHS} are stochastic in time, i.e., continuous everywhere but nowhere differentiable, we use MC to discretize the time integral instead of Gaussian quadrature, which requires sufficient smoothness of $b(s)$. In contrast, MC integration does not require differentiability or smoothness. Specifically, we replace the derivative bounds of \citet[Lemma 12]{an_quantum_2025} with probabilistic regularity estimates.

To use MC integration for time discretization (as opposed to the Gaussian quadrature discretization of \citet[Equation~212]{an_quantum_2025}) for the stochastic inhomogeneous term of LCHS, we first define the MC estimator
\begin{equation}\label{eq:MC_estimator_lemma}
\widehat{\mathcal{I}}_M
:=
\frac{T}{M}\sum_{j=1}^M
\left(\sum_{m_1,q_1} c_{q_1,m_1}\,U(T,S_j,k_{q_1,m_1})\,b(S_j)\right).
\end{equation}
sampling i.i.d.\ times $S_1,\dots,S_M \sim \mathrm{Unif}[0,T]$, independent of the OU process and hence of $(A,L,H,b)$, where $M$ is the MC sampling size.

\begin{lemma}[Monte Carlo integration for time discretization]\label[lemma]{lem:MC_error_inhomo}
    Consider the discretization in~\cref{eq:lchs-b} and assume 
$
\sup_{t\in[0,T]}\,\mathbb{E}\big[\|L(t,\omega)\|\big] < \infty, \, 
\sup_{t\in[0,T]}\,\mathbb{E}\big[\|H(t,\omega)\|\big] < \infty.
$ 
For any probability $\delta > 0$, in order to bound the approximation error of~\cref{eq:lchs-b} by $\eps$:
\begin{equation}\label{eq:P_Sw_MC}
\mathbb{P}_{S,\omega}\!\Bigg(
\left\|
\int_0^T \mathcal{T}e^{-\int_s^T A(s')\,ds'}\,b(s)\,ds - \widehat{\mathcal{I}}_M
\right\|
\le \eps
\Bigg)
\;\ge\; 1-\delta,
\end{equation}
it suffices to choose
\begin{align}
K &= \mathcal{O}\!\left( \left[\log\!\left(1+\frac{T\,\sqrt{V}}{\eps\,\sqrt{\delta}}\right)\right]^{1/\beta} \right), \\
Q_1 &= \mathcal{O}\!\left( \log\!\left(1+\frac{T\,\sqrt{V}}{\eps\,\sqrt{\delta}}\right) \right), \\
h_1 &= \frac{1}{e\,T\,\sup_{t\in[0,T]}\mathbb{E}\!\big[\|L(t,\omega)\|\big]} = \frac{1}{e\,T\,N \Big(\|\mathbf{F}_1\| + \|\mathbf{F}_2\| + \sqrt{\frac{1-e^{-2\lambda_{\min} T}}{2\,\lambda_{\min}}}\Big)}, \label{eq:h1_choice_lemma}\\
\end{align}
together with the MC sample size
\begin{equation}\label{eq:M_choice_lemma}
M
=
\mathcal{O}\!
\left(
\frac{T\,(\sum_{m_1,q_1} |c_{q_1,m_1}|)^2}{\delta\,\eps^2}\,V
\right),
\end{equation}
where
\begin{equation}\label{eq:VE}
V := \frac{1}{T}\int_0^T \mathbb{E}_\omega\big[\|b(s)\|^2\big]\,ds \le \frac{\|\boldsymbol{\Sigma}\|_{F}^{2}}{2\,\lambda_{\min}\,T}\,\left(\,T\ -\ \frac{1-e^{-2\lambda_{\min} T}}{2\,\lambda_{\min}}\,\right)
\end{equation}
\end{lemma}

\begin{proof}  
We decompose the total approximation error into $k$-discretized integral and MC-discretized integral:
\begin{align}
\left\|
\int_0^T \mathcal{T}e^{-\int_s^T A(s')\,ds'}\,b(s)\,ds - \widehat{\mathcal{I}}_M
\right\|
&\le
\int_0^T \Big\| \int_{\mathbb{R}} \frac{f(k)}{1-ik}\,U(T,s,k)\,dk - \sum_{m_1,q_1} c_{q_1,m_1}\,U(T,s,k_{q_1,m_1})\Big\|\,\|b(s)\|\,ds \\
&\quad+\;
\left\|\int_0^T \sum_{m_1,q_1} c_{q_1,m_1}\,U(T,s,k_{q_1,m_1})\,b(s)\,ds - \widehat{\mathcal{I}}_M\right\| \\
\label{eq:total-error-MC-k}
&\le
\underbrace{\int_0^T \delta_k\,\|b(s)\|\,ds}_{\text{$k$-discretization error}}
+
\underbrace{\left\|
\int_0^T \sum_{m_1,q_1} c_{q_1,m_1}\,U(T,s,k_{q_1,m_1})\,b(s)\,ds - \widehat{\mathcal{I}}_M
\right\|}_{\text{Monte Carlo time error}}.
\end{align}

First, for brevity define 
\begin{equation}\label{eq:Fdisc_def}
F_{\mathrm{disc}}(s,\omega) := \sum_{m_1,q_1} c_{q_1,m_1}\,U(T,s,k_{q_1,m_1})\,b(s,\omega),
\end{equation}
so that the MC estimator in \cref{eq:MC_estimator_lemma} reads
\[
\widehat{\mathcal{I}}_M = \frac{T}{M}\sum_{j=1}^M F_{\mathrm{disc}}(S_j,\omega).
\]
By unitarity of $U$ and the triangle inequality,
\begin{equation}\label{eq:Fdisc_norm_bound}
\|F_{\mathrm{disc}}(s,\omega)\| \le \sum_{m_1,q_1} |c_{q_1,m_1}| \,\|b(s,\omega)\| = C_f^{\mathrm{disc}}\,\|b(s,\omega)\|.
\end{equation}
We now take expectations of \cref{eq:MC_estimator_lemma} over both $S$ and $\omega$.
For the $k$-discretization term:
\begin{equation}\label{eq:E_Sw_k}
\mathbb{E}_{S,\omega}\!\left[\int_0^T \delta_k\,\|b(s,\omega)\|\,ds\right]
=
\mathbb{E}_{\omega}\!\left[\int_0^T \delta_k\,\|b(s,\omega)\|\,ds\right]
=
\delta_k\,\mathbb{E}_{\omega}\!\big[\|b\|_{L^1(0,T)}\big],
\end{equation}
since $\delta_k$ is deterministic (uniform in $s,\omega$) and $\mathbb E_S$ acts only on $S$.
For the MC time discretization error term: conditional second moment over $S$ (for fixed $\omega$). Using unbiasedness and independence,
\begin{align*}
\mathbb{E}_{S}\!\left[\left\|\int_0^T F_{\mathrm{disc}}(s,\omega)\,ds - \frac{T}{M}\sum_{j=1}^M F_{\mathrm{disc}}(S_j,\omega)\right\|^2\right]
&=
\frac{T^2}{M}\,\mathbb{E}_{S}\!\left[\big\|F_{\mathrm{disc}}(S,\omega) - \mathbb{E}_{S}[F_{\mathrm{disc}}(S,\omega)]\big\|^2\right] \\
&\le \frac{T^2}{M}\,\mathbb{E}_{S}\!\left[\|F_{\mathrm{disc}}(S,\omega)\|^2\right] \\
&= \frac{T^2}{M}\cdot \frac{1}{T}\int_0^T \|F_{\mathrm{disc}}(s,\omega)\|^2\,ds \\
&\le \frac{T^2}{M}\cdot \frac{1}{T}\int_0^T (C_f^{\mathrm{disc}})^2\,\|b(s,\omega)\|^2\,ds,
\end{align*}
where we used \cref{eq:Fdisc_norm_bound}.
Taking $\mathbb E_\omega$ of this conditional bound and applying Jensen’s inequality to pass from second to first moment yields
\begin{align}
\mathbb{E}_{S,\omega}\!\left[\left\|
\int_0^T F_{\mathrm{disc}}(s,\omega)\,ds - \widehat{\mathcal{I}}_M
\right\|\right]
&\le
\sqrt{\mathbb{E}_{S,\omega}\!\left[\left\|
\int_0^T F_{\mathrm{disc}}(s,\omega)\,ds - \widehat{\mathcal{I}}_M
\right\|^2\right]} \\
\label{eq:E_Sw_F}
&\le \frac{T\,C_f^{\mathrm{disc}}}{\sqrt{M}}
\left(\frac{1}{T}\int_0^T \mathbb{E}_{\omega}\!\left[\|b(s,\omega)\|^2\right]\,ds\right)^{1/2}.
\end{align}
Combining \cref{eq:E_Sw_k} and \cref{eq:E_Sw_F} and using linearity of expectation, we obtain the joint expectation bound
\begin{equation}\label{eq:ESw_total_bound}
\mathbb{E}_{S,\omega}\!\left[
\left\|
\int_0^T \mathcal{T}e^{-\int_s^T A(s')\,ds'}\,b(s)\,ds - \widehat{\mathcal{I}}_M
\right\|
\right]
\le
\delta_k\,\mathbb{E}_{\omega}\!\big[\|b\|_{L^1(0,T)}\big]
+
\frac{T\,C_f^{\mathrm{disc}}}{\sqrt{M}}
\left(\frac{1}{T}\int_0^T \mathbb{E}_{\omega}\!\left[\|b(s)\|^2\right]\,ds\right)^{\!1/2}.
\end{equation}
For any $\delta\in(0,1)$, by Chebyshev’s inequality,
\begin{equation}\label{eq:hp_joint_chebyshev}
\mathbb{P}_{S,\omega}\!\left(
\mathcal{E}_{\mathrm{MC}}(S,\omega)
\le
\frac{T\,C_f^{\mathrm{disc}}}{\sqrt{\delta\,M}}\,
\left(\frac{1}{T}\int_0^T \mathbb{E}_{\omega}\!\left[\|b(s)\|^2\right]\,ds\right)^{\!1/2}
\right)
\;\ge\; 1-\delta.
\end{equation}
Combining \cref{eq:hp_joint_chebyshev} with \cref{eq:total-error-MC-k} observing the $k$-discretization term is deterministic given $\omega$ yields the desired joint high-probability bound on the total error:
\begin{equation}\label{eq:total_hp_bound}
\mathbb{P}_{S,\omega}\!\left(
\left\|
\int_0^T \mathcal{T}e^{-\int_s^T A(s')\,ds'}\,b(s)\,ds - \widehat{\mathcal{I}}_M
\right\|
\le
\underbrace{\int_0^T \delta_k\,\|b(s,\omega)\|\,ds}_{\text{$k$-discretization term}}
+
\underbrace{\frac{T\,C_f^{\mathrm{disc}}}{\sqrt{\eps\,M}}\,
\left(\frac{1}{T}\int_0^T \mathbb{E}_{\omega}\!\left[\|b(s)\|^2\right]\,ds\right)^{\!1/2}}_{\text{MC time term}}
\right)\;\ge\; 1-\delta.
\end{equation}
By Cauchy-Schwarz inequality, for any \(\omega\),
\begin{equation*}
\int_0^T \|b(s,\omega)\|\,ds \;\le\; \sqrt{T}\,\Big(\int_0^T \|b(s,\omega)\|^2\,ds\Big)^{1/2}.
\end{equation*}
By Markov's inequality, for any \(\delta_2\in(0,1)\),
\begin{equation*}
\mathbb{P}_\omega\!\left(
\int_0^T \|b(s,\omega)\|\,ds
\;\le\;
\sqrt{\tfrac{T}{\delta_2}}\,
\Big(\int_0^T \mathbb{E}_\omega[\|b(s)\|^2]\,ds\Big)^{1/2}
\right)\;\ge\; 1-\delta_2.
\end{equation*}
Hence, with probability at least \(1-\delta_2\) over \(\omega\),
\begin{equation}\label{eq:k_term_detRHS_delta}
\int_0^T \delta_k\,\|b(s,\omega)\|\,ds
\;\le\;
\delta_k\,\sqrt{\tfrac{T}{\delta_2}}\,
\Big(\int_0^T \mathbb{E}_\omega[\|b(s)\|^2]\,ds\Big)^{1/2}.
\end{equation}
Applying the union bound to \cref{eq:total_hp_bound} together with \cref{eq:k_term_detRHS_delta} and \cref{eq:hp_joint_chebyshev}, we obtain, for any \(\delta_1=\delta_2=\delta/2\),
\begin{align}
\mathbb{P}_{S,\omega}\!\Bigg(
\left\|
\int_0^T \mathcal{T}e^{-\int_s^T A(s')\,ds'}\,b(s)\,ds - \widehat{\mathcal{I}}_M
\right\|
&\le
\delta_k\,\sqrt{\tfrac{2T}{\delta}}\,
\Big(\int_0^T \mathbb{E}_\omega[\|b(s)\|^2]\,ds\Big)^{1/2}
\label{eq:hp_total_detRHS_delta}\\
&\quad+
\frac{T\,C_f^{\mathrm{disc}}}{\sqrt{(\delta/2)\,M}}\,
\Big(\frac{1}{T}\int_0^T \mathbb{E}_\omega[\|b(s)\|^2]\,ds\Big)^{1/2}
\Bigg)
\;\ge\; 1-\delta, \nonumber
\end{align}
which leads to the claimed \cref{eq:P_Sw_MC} and offers a deterministic bound compared to \cref{eq:total_hp_bound}.

To ensure the total error \(\le \eps\) with probability at least \(1-\delta\), it suffices to make each term in \cref{eq:hp_total_detRHS_delta} \(\le \eps/2\). For the $k$-discretization choices, together with \cref{eq:VE},
we then have
\[
\delta_k \,\sqrt{\tfrac{T}{\delta_2}}\,\sqrt{T\,V} \;\le\; \frac{\eps}{2}
\quad\Longrightarrow\quad
\delta_k \;\le\; \frac{\eps\,\sqrt{\delta_2}}{2\,T\,\sqrt{V}}.
\]
Under the LCHS filter construction (tail parameter \(\beta\in (0,1)\)), the result of~\cite{an_linear_2023} $k$-discretization error obeys (up to absolute constants)
\[
\delta_k \le \exp\!\big(-(K h_1)^\beta\big) + \exp\!\big(-Q_1\big) +o(1).
\]
where $h_1$ is the timestep size defined in~\eqref{eqn:LCHS_LCU_composite}.
We then bound each term to be upper bounded by \(\delta_k^\star/2\), where
\[
\delta_k^\star := \frac{\eps\,\sqrt{\delta_2}}{2\,T\,\sqrt{V}}.
\]
Then sufficient choices are
\begin{align}
K &\ge \frac{1}{h_1}\Big[\log\!\Big(\frac{2}{\delta_k^\star}\Big)\Big]^{1/\beta}
\;=\;
\frac{1}{h_1}\Big[\log\!\Big( 1 + \frac{4\,T\,\sqrt{V}}{\eps\,\sqrt{\delta_2}} \Big)\Big]^{1/\beta}, \label{eq:K_tail_delta}\\
Q_1 &\ge \log\!\Big(\frac{2}{\delta_k^\star}\Big)
\;=\;
\log\!\Big( 1 + \frac{4\,T\,\sqrt{V}}{\eps\,\sqrt{\delta_2}} \Big). \label{eq:Q1_tail_delta}
\end{align}
with the stepsize chosen to be
\begin{equation}\label{eq:h1_stability}
h_1 = \frac{1}{e\,T\,\sup_{t\in[0,T]}\mathbb{E}_\omega\!\big[\|L(t,\omega)\|\big]},
\end{equation}
which yields the stated choice of \cref{eq:h1_choice_lemma} by setting $\delta_1 = \delta/2$. The explicit bound for $h_1$ in terms of OU parameters is given in \cref{eq:h1-explicit} below.

For the MC sample size,
\[
\frac{T\,C_f^{\mathrm{disc}}}{\sqrt{\delta_1\,M}}\,\sqrt{V} \;\le\; \frac{\eps}{2}
\quad\Longleftrightarrow\quad
M \;\ge\; \frac{4\,T\,(C_f^{\mathrm{disc}})^2}{\delta_1\,\eps^2}\,V,
\]
which yields the stated choice of \cref{eq:M_choice_lemma} by setting $\delta_1 = \delta/2$.

It remains to express the mean energy input $V$ of \cref{eq:VE} in terms of the OU process parameters. Since $\|b(s)\|^2 = \|\mathbf{F}_0(s)\|^2$, combining \cref{eq:MarkovUnion1}, we have
\begin{align}
V\ & :=\ \frac{1}{T}\int_0^T \mathbb{E}_{\omega}\big[\|b(s)\|^2\big] \,ds
\ =\ \frac{1}{T}\int_0^T \mathbb{E}_{\omega}\big[\|\mathbf{F}_0(s)\|^2\big] \,ds \\
\label{eq:V_OU_bound}
& \le \ \frac{1}{T}\int_0^T \frac{1-e^{-2\lambda_{\min} s}}{2\,\lambda_{\min}}\ \|\boldsymbol{\Sigma}\|_{F}^{2}\,ds
\ =\ \frac{\|\boldsymbol{\Sigma}\|_{F}^{2}}{2\,\lambda_{\min}\,T}\,\left(\,T\ -\ \frac{1-e^{-2\lambda_{\min} T}}{2\,\lambda_{\min}}\,\right).
\end{align}
In particular, under the scalar mean-reversion case \(\boldsymbol{\Theta}=\lambda_{\min} I\) with \(\mathbf{F}_0(0)=0\), the above expression holds with equality; for general \(\boldsymbol{\Theta}\) with \(\lambda_{\min}=\lambda_{\min}((\boldsymbol{\Theta}+\boldsymbol{\Theta}^\top)/2)>0\), the same formula provides an explicit upper bound on \(V\). This further offers an explicit  bound for the stability choice \cref{eq:h1_stability} in terms of OU parameters.
To bound $\mathbb{E}_\omega\big[\|L_N(t,\omega)\|\big]$ in \cref{eq:h1_stability},
taking expectations for \cref{eq:L-simplified} and using Jensen's inequality,
\[
\mathbb{E}_\omega\big[\|L_N(t,\omega)\|\big]\ \le\ N \Big(\|\mathbf{F}_1\| + \|\mathbf{F}_2\|\ +\ \mathbb{E}_\omega\big[\|\mathbf{F}_0(t,\omega)\|\big] \Big)
\ \le\ N \Big(\|\mathbf{F}_1\| + \|\mathbf{F}_2\|+ \sqrt{\mathbb{E}_\omega\big[\|\mathbf{F}_0(t,\omega)\|^2\big]} \Big).
\]
Plugging the OU second-moment bound \cref{eq:MarkovUnion1} into the above inequality yields
\begin{equation}\label{eq:E-L-explicit}
\mathbb{E}_\omega\big[\|L_N(t,\omega)\|\big]\ \le\ N \Big(\|\mathbf{F}_1\| + \|\mathbf{F}_2\|+ \|\boldsymbol{\Sigma}\|_{F}\,\sqrt{\frac{1-e^{-2\lambda_{\min} t}}{2\,\lambda_{\min}}}\Big) \,.
\end{equation}
Since the rhs is monotone increasing in \(t\), taking the supremum of \cref{eq:E-L-explicit} over \(t\in[0,T]\) yields
\begin{align}\label{eq:sup-E-L-explicit}
\sup_{t\in[0,T]}\ \mathbb{E}_\omega\big[\|L_N(t,\omega)\|\big]\ & \le\ N \Big(\|\mathbf{F}_1\| + \|\mathbf{F}_2\| + \|\boldsymbol{\Sigma}\|_{F}\,\sqrt{\frac{1-e^{-2\lambda_{\min} T}}{2\,\lambda_{\min}}} \Big) \,.
\end{align}
Plugging \cref{eq:sup-E-L-explicit} into \cref{eq:h1_stability} yields the explicit lower bound,
\begin{equation}\label{eq:h1-explicit}
h_1\ \ge\ \frac{1}{\displaystyle e\,T\,N \left(\|\mathbf{F}_1\| + \|\mathbf{F}_2\|\ +\ \|\boldsymbol{\Sigma}\|_{F}\,\sqrt{\frac{1-e^{-2\lambda_{\min} T}}{2\,\lambda_{\min}}}\,\right)} \,.
\end{equation}
When \(\boldsymbol{\Theta}=\lambda_{\min} I\) and \(\mathbf{F}_0(0)=0\), equality holds in \cref{eq:E-L-explicit}–\cref{eq:sup-E-L-explicit}.
\end{proof}

Combining \cref{eqn:LCHS_LCU_composite} and \cref{eq:lchs-b}, we obtain the discretized approximation of \cref{eqn:ODE_solu} at the final evolution time $T$ is,
\begin{equation}\label{eqn:CRD_solu_discrete}
  y(T) \approx \sum_{j=0}^{N_U-1} c_j U(T, 0, k_j) \ket{y(0)} +  \sum_{j'=0}^{M'-1} \sum_{j=0}^{N_U-1} c'_{j} c_{j} U(T,s_{j'},k_{j})  \ket{b(s_{j'})}, 
\end{equation}
where 
\begin{equation}\label{eq:UTSK}
U(T,s,k) = \mathcal{T} e^{-i \int_s^T (kL(s')+H(s')) \ud s'} 
\end{equation}
are unitaries and $N_U$ is the total number of the nodes given by \cref{eq:homo-M}.
Given the circuit for the unitaries $U(T,s,k)$, each summation can be
implemented by the linear combination of unitaries (LCU)
technique~\cite{childs_hamiltonian_nodate} for both terms in
\cref{eqn:CRD_solu_discrete} and the sum of the two terms can be computed by
another LCU\@. For each time-dependent Hamiltonian simulation (\cref{eq:UTSK}),
we use the truncated Dyson series to achieve the best asymptotic scaling
\citep{low_hamiltonian_2019}. Since $L(t)$ and $H(t)$ in \cref{eq:UTSK} and $b(t)$ in \cref{eqn:CRD_solu_discrete} are OU-driven stochastic, a Lipschitz continuity guarantee of them are provided to apply Dyson series at short time segments.      

\subsection{Truncated Dyson series  simulation}
\label{sec:MC-TDS}

We now implement each $U(T, s, k)$ in \cref{eqn:CRD_solu_discrete} using truncated Dyson series. The truncation order is guaranteed \citep[Lemma 5]{low_hamiltonian_2019} but the time discretization error is not, due to the stochasticity of $U(T, s, k)$, i.e., $U(T, s, k)$ is everywhere continuous but \textit{nowhere differentiable}.  

Let $H:[0,T]\times\Omega\to\mathbb{C}^{n\times n}$ be a bounded Hermitian time-dependent Hamiltonian. For each $\omega$, the (backward) propagator $U_H(T,s,\omega)$ is unitary and satisfies the evolution equation \cite[Chapter~5]{pazy_semigroups_1983}
\begin{equation}\label{eq:evol}
\partial_s U_H(T,s,\omega) \;=\; -\, i\, U_H(T,s,\omega)\, H(s,\omega),
\qquad U_H(T,T,\omega)=I.
\end{equation}

The time-ordered exponential admits the Dyson expansion, which can be derived by using the fundamental theorem of calculus recursively:
\begin{equation}\label{eq:Dyson}
U_H(T,0,\omega) \;=\; \sum_{k=0}^{\infty} (-i)^k\, \mathcal{I}_k[H](\omega),
\end{equation}
where we set $\mathcal{I}_0[H]=I$ and, for $k\ge 1$,
\begin{equation}\label{eq:Ik-def}
\mathcal{I}_k[H](\omega) := \int_{0\le t_1\le\cdots\le t_k\le T}
H(t_k,\omega)\cdots H(t_1,\omega)\, dt_1\cdots dt_k.
\end{equation}
We fix a uniform temporal grid with mesh $h:=T/ \mathcal M$ and nodes $t_m:=m h$ (assume $\mathcal M\in\mathbb{N}$ for simplicity). The truncated Dyson series (TDS) of order $\mathcal K$ with left-endpoint Riemann discretization of the $k$-fold integrals reads
\begin{equation}\label{eq:TDS-KM}
\mathrm{TDS}_{\mathcal K,\mathcal M}[H](\omega) \;:=\; \sum_{k=0}^{\mathcal K} (-i)^k\, \mathcal{Q}_{k,\mathcal M}[H](\omega),
\end{equation}
where $\mathcal{Q}_{0,\mathcal M}[H]=I$, and for $k\ge 1$,
\begin{equation}\label{eq:QkM-def}
\mathcal{Q}_{k,\mathcal M}[H](\omega) := \sum_{0\le m_1\le \cdots\le m_k\le \mathcal M-1}
h^k\, H(t_{m_k},\omega)\cdots H(t_{m_1},\omega).
\end{equation}
We provide a pathwise error bound for
\begin{equation}\label{eq:total-error-def}
\mathcal{E}_{\mathcal K,\mathcal M}(\omega) := \left\|U_H(T,0,\omega)-\mathrm{TDS}_{\mathcal K,\mathcal M}[H](\omega)\right\|.
\end{equation}

\begin{proposition}[Telescoping product bound]\label{prop:telescoping}
Let $A_j,B_j$ be matrices with $\norm{A_j},\norm{B_j}\le \Lambda$ for $1\le j\le k$. Then
\begin{equation}\label{eq:telescoping}
\norm{A_k\cdots A_1 - B_k\cdots B_1} \;\le\; \Lambda^{k-1}\, \sum_{j=1}^{k} \norm{A_j - B_j}.
\end{equation}
\end{proposition}

\begin{proof}
The proof follows inductively.  We take as our base case $k=1$, which follows directly from the triangle inequality.  Now let us assume that the induction hypothesis holds for $k'\ge 1$.  Then
\begin{align}
    \|A_{k'+1} \prod_{p=1}^{k'} A_k - B_{k'+1}\prod_{p=1}^{k'}B_{p}\| &\le \|A_{k'+1}\left(\prod_{p=1}^{k'} A_p -\prod_{p=1}^{k'} B_p \right) \| +\|(B_{k'+1}-A_{k'+1})\prod_{p=1}^{k'}B_{p}\|\nonumber\\
    &\le \Lambda^{k'+1}\sum_{j=1}^{k'}\|A_j-B_j\| + \Lambda^{k'+1}\|A_{k'+1}-B_{k'+1}\|\nonumber\\
    &\le \Lambda^{k'+1}\sum_{j=1}^{k'+1}\|A_j-B_j\|,
\end{align}
which demonstrates the claim inductively.
\end{proof}
This result can be used to combine the error of multiple segments of a longer quantum algorithm.  Below, we investigate the error in a short time propagator of the truncated Dyson series neglecting, for the moment, the error in approximate integration.  This discussion is nearly identical to the case studied in~\cite{low_hamiltonian_2019} but here we include the explicit dependence on the noise resolution $\omega$.

We provide a method for discussing implementing a truncated Dyson series below for cases, such as the OU process, where the Hamiltonian is not differentiable.  Following the same reasoning as that used in the truncated Dyson series approach, we break the evolution into a large number of short timesteps and then will ultimately use fixed point amplitude amplification to glue these short evolutions together to form a single long evolution.

\begin{lemma}[Short time truncation error of TDS]\label[lemma]{thm:TDS-path-K}
Let $H:[0,\tau]\times\Omega\to\mathbb{C}^{n\times n}$ be a bounded Hermitian time-dependent Hamiltonian and let 
$
\Lambda(\omega):=\sup_{t\in[0,T]} \norm{H(t,\omega)}<\infty
$
and $\mathcal{K} \in \mathbb{Z}_+$ be the Dyson series truncation order.
A time dependent simulation of \cref{eq:evol} can be implemented within error 
\begin{align}
\mathcal E_{\mathcal K}(\omega)&=\left\|\mathcal{\tau}\left[e^{-i\int_0^\tau H(s,\omega) \mathrm{d}s}\right]- \sum^{\mathcal{K}}_{k=0} (-i)^k \mathcal{I}_k\right\| \le \eps_1
\label{eq:TDS-path-K}
\end{align}
if we take
\begin{enumerate}
  \item $\mathcal{K} = \left\lceil-1 + \frac{2\ln(1/\eps_1)}{\ln\ln(1/\eps_1) +1}\right\rceil$ 

\item $\Lambda(\omega) \tau \le \ln 2$ with $\tau$ being the duration of the time segment 
\end{enumerate}
\end{lemma}

\begin{proof}
For the Dyson truncation error $\mathcal E_{\mathcal K}(\omega)$, we first bound $\norm{\mathcal{I}_k[H]}$ by using $\sup_t\norm{H(\omega, t)}\le \Lambda(\omega)$ and submultiplicativity,
\begin{equation}\label{eq:tail-bound}
\norm{\mathcal{I}_k[H]} \le \int_{0\le t_1\le\cdots\le t_k\le T} \Lambda^k\, dt
= \Lambda^k\, \frac{\tau^k}{k!}.
\end{equation}
Then, following \citet[Lemma 11]{low_hamiltonian_2019}, we obtain
\begin{align}
  \label{eq:e_K-Ik}
\mathcal E_{\mathcal K}(\omega)&=\left\|\mathcal{T}\left[e^{-i\int_0^\tau H(s) \mathrm{d}s}\right]- \sum^{\mathcal K}_{k=0} (-i)^k \mathcal{I}_k\right\|\le \sum^\infty_{k=\mathcal{K}+1} \|\mathcal{I}_k\|
\le \sum^\infty_{k=\mathcal{K}+1} \Lambda^k\, \frac{\tau^k}{k!}
\nonumber\\
&
\le \frac{(\tau \Lambda)^{\mathcal{K}+1}}{(\mathcal{K}+1)!}\underbrace{\sum^{\infty}_{k=\mathcal{K}+2}\left(1/2\right)^{k-\mathcal{K}-1}}_{=1}
=\frac{(\tau \Lambda)^{\mathcal{K}+1}}{(\mathcal{K}+1)!}
\\
\label{eq:errorK-TDS}
&\le\left(\frac{\tau \Lambda(\omega) e}{\mathcal{K}+1}\right)^{\mathcal{K}+1},
\end{align}
where the factorial term $(\mathcal{K}+1)!$ is approximated by Stirling's approximation and $\mathcal{K}\ge 2 \Lambda \tau$ is assumed. \Cref{eq:errorK-TDS} now allows us to bound the Dyson truncation order $\mathcal{K}$ for a given error $\mathcal E(\omega)\le \eps_1 \in [0, 2^{-e}]$ \citep[Lemma 11]{low_hamiltonian_2019}:
\begin{align}
\mathcal{K} \ge \max\left\{-1+\frac{\ln(1/\eps_1)}{W\left(\frac{\ln(1/\eps_1)}{T \Lambda e} \right)},2\tau \Lambda\right\},
\end{align}
where $W$ is the Lambert-W function.  Using the fact that for $x\ge 1$, $W(x) \ge (\ln(x)+1)/2$ and $\ln(e\ln 2) <1$ we obtain a simpler bound
\begin{equation}
\mathcal{K} = \left\lceil-1 + \frac{2\ln(1/\eps_1)}{\ln\ln(1/\eps_1) +1}\right\rceil\in{\mathcal{O}}\left(\frac{\ln(1/\eps_1)}{\ln\ln(1/\eps_1) }\right).\label{eq:qprimebd}
\end{equation}
\end{proof}

The next piece represents a crucial difference between prior analyses and ours. The time-dependent Hamiltonian simulation methods of~\cite{low_hamiltonian_2019,kieferova2019simulating} assume a piecewise differentiable Hamiltonian, whereas here the Hamiltonian is everywhere continuous but nowhere differentiable. Below we provide an argument showing that we can still (for fixed noise resolution $\omega$) address this issue. The first pathway relies on Riemann integral discretization combined with Hölder continuity, which leads to implicit bounds for $\mathcal K$ and $h$ in \cref{lem:expectation} of \cref{sec:riemann-integral}. The second pathway adopts MC integration of the truncated Dyson series, which offers explicit bounds and avoids the Hölder-continuity condition in Riemann integration (\cref{sec:riemann-integral}), since MC integration does not involve discretization. We start by defining the MC estimator for Dyson integrals.

\begin{definition}[Monte Carlo estimator for Dyson integrals]
\label[definition]{def:uniform-simplex-MC}
Let $k\in\mathbb{N}$ and $T>0$. Define the ordered simplex
\[
\Delta_k\ :=\ \{(t_1,\dots,t_k)\in[0,T]^k:\ 0\le t_1\le \cdots \le t_k \le T\}.
\]
The uniform density on $\Delta_k$ is
\begin{equation}
\label{eq:rho-uniform}
\rho_k(t)\ :=\ \frac{k!}{T^k},\qquad t=(t_1,\dots,t_k)\in \Delta_k,
\end{equation}
so that 
\begin{equation}
\int_{\Delta_k}\rho_k(t)\,dt=1.
\end{equation} 
Equivalently, $\rho_k$ is the joint density of the order statistics of $k$ i.i.d.\ $\mathrm{Uniform}(0,T)$ random variables \cite[Theorem~2.1.1]{h_a_david_h_n_nagaraja_basic_2003}.

Given a bounded measurable matrix-valued integrand $Z_k:\Delta_k\to\mathbb{C}^{n\times n}$, the (ordered) Dyson $k$-fold integral over $[0,T]$ is
\begin{equation}
\label{eq:Ik-Dyson-def}
\mathcal{I}_k[Z]\ :=\ \int_{\Delta_k} Z_k(t)\, dt
\ =\ \frac{T^k}{k!}\,\mathbb{E}_{t\sim \rho_k}\big[\,Z_k(t)\,\big],
\end{equation}
and the standard MC estimator with $N_k$ i.i.d.\ samples $t^{(1)},\dots,t^{(N_k)}\sim \rho_k$ is
\begin{equation}
\label{eq:Ik-MC-def}
\widehat{\mathcal{I}}_k[Z]\ :=\ \frac{T^k}{k!}\cdot \frac{1}{N_k}\sum_{r=1}^{N_k} Z_k\big(t^{(r)}\big).
\end{equation}
\end{definition}

For Dyson integrals with 
\begin{equation}\label{eq:Z_k}
Z_k(t,\omega):=H(t_k,\omega)\cdots H(t_1,\omega),
\end{equation}
where $H:[0,T]\times\Omega\to\mathbb{C}^{n\times n}$ is a (pathwise) bounded Hermitian Hamiltonian, we have
\begin{align}\label{eq:Ik-Dyson}
&\mathcal{I}_k[H](\omega)\ := \ \int_{\Delta_k} H(t_k,\omega)\cdots H(t_1,\omega)\, dt= \int_{\Delta_k}Z_k(t,\omega)\, dt  \\
\label{eq:MC-dyson}
& \widehat{\mathcal{I}}_k[H](\omega)\ :=\ \frac{T^k}{k!}\cdot \frac{1}{N_k}\sum_{r=1}^{N_k} H(t_k^{(r)},\omega)\cdots H(t_1^{(r)},\omega).
\end{align}

\begin{lemma}[Monte Carlo estimator for ordered Dyson integrals: unbiasedness and moment bounds]
\label[lemma]{lem:MC-simplex}
Under the definitions and assumptions in \cref{def:uniform-simplex-MC},
the MC estimator for ordered Dyson integrals (\cref{eq:MC-dyson}) is unbiased:
\begin{align}
\label{eq:Ik-MC-properties}
& \mathbb{E}\big[\widehat{\mathcal{I}}_k[Z]\big]\ =\ \mathcal{I}_k[Z], 
\end{align}
 and its mean square error is bounded by
\begin{align}
\label{eq:Ik-MC-moments}
& \mathbb{E}\Big[\big\|\widehat{\mathcal{I}}_k[Z]-\mathcal{I}_k[Z]\big\|^2\Big]\ \le\ \frac{1}{N_k}\,\left(\frac{T^k}{k!}\right)^2\,\mathbb{E}_{t\sim\rho_k}\|Z_k(t)\|^2,
\end{align}
\end{lemma}

\begin{proof}
The sample mean is the optimal unbiased estimator of the population mean.
This estimator can be constructed by sampling $t^{(r)}$ by drawing $k$ i.i.d.\ $U(0,T)$ variables and sorting them; the joint density of the order statistics is exactly $\rho_k$ \cite[Theorem~2.1.1]{h_a_david_h_n_nagaraja_basic_2003}. \cref{eq:Ik-MC-properties} is a direct result of the fact that
sampling time-ordered interaction times uniformly and averaging observables reproduces the exact Dyson integral on average, because the time-ordered exponential is linear in the integral over those times.
Taking the expectation of \cref{eq:MC-dyson}:
\begin{equation}\label{eq:E-I}
\mathbb{E}\big[\widehat{\mathcal{I}}_k[H](\omega)\big] = \frac{T^k}{k!}\cdot \frac{1}{N_k}\sum_{r=1}^{N_k} \mathbb{E}[Z_k(t^{(r)},\omega)],
\end{equation}
where $\mathbb{E}[Z_k(t^{(r)},\omega)]$ can be obtained by
definition of the Dyson integral and Fubini’s theorem:
\begin{equation}\label{eq:E-Zk}
\mathbb{E}[Z_k(t,\omega)] \;=\; \int_{\Delta_k} Z_k(t,\omega)\,\rho_k(t)\,dt \;=\; \frac{k!}{T^k}\,\mathcal{I}_k[H](\omega)
\end{equation}
since $\rho_k$ is uniform on $\Delta_k$ ($\mathbb{E}[Z_k(t,\omega)]$ is the first moment of $Z_k$ with its density function $\rho_k$).
Plugging \cref{eq:E-Zk} into \cref{eq:E-I} and observing the definition in \cref{eq:Ik-Dyson} yields \cref{eq:Ik-MC-properties}.

For each sample define the random matrix $Y_r := \frac{T^k}{k!} Z_k(t^{(r)},\omega)$. Then the random error is
\[
\widehat{\mathcal{I}}_k[H](\omega) - \mathcal{I}_k[H](\omega) \;=\; \frac{1}{N_k}\sum_{r=1}^{N_k}\big(Y_r-\mathbb{E}[Y_r]\big).
\]
By Jensen inequality and independence of random matrix $Y_r$,
\[
\mathbb{E}\Big\| \frac{1}{N_k}\sum_{r=1}^{N_k}\big(Y_r-\mathbb{E}[Y_r]\big)\Big\|
\;\le\; \sqrt{\mathbb{E}\Big\| \frac{1}{N_k}\sum_{r=1}^{N_k}\big(Y_r-\mathbb{E}[Y_r]\big)\Big\|^2}
\;=\; \frac{1}{\sqrt{N_k}}\,\sqrt{\mathbb{E}\|Y_1-\mathbb{E}[Y_1]\|^2}
\;\le\; \frac{1}{\sqrt{N_k}}\,\sqrt{\mathbb{E}\|Y_1\|^2},
\]
which yields \cref{eq:Ik-MC-moments}.
\end{proof}

We now provide the MC integration error bound for TDS:
\begin{equation}
\mathcal{E}_{\mathcal K,\mathbf{N}}(\omega) := \| U_H(T,0,\omega) - \mathrm{TDS}^{\mathrm{MC}}_{\mathcal K,\mathbf{N}}[H](\omega) \|,
\end{equation}
where the MC TDS approximation of order $\mathcal K$ is defined as
\begin{equation}\label{eq:tdsdef}
\mathrm{TDS}^{\mathrm{MC}}_{\mathcal K,\mathbf{N}}[H](\omega)
\;:=\; \sum_{k=0}^{\mathcal K} (-i)^k\, \widehat{\mathcal{I}}_k[H](\omega),
\end{equation}
with $\widehat{\mathcal{I}}_0:=I$ and $\widehat{\mathcal{I}}_k$ as in \cref{eq:MC-dyson} for $k\ge1$.

\begin{lemma}\label[lemma]{lemma:PAC-MC}
Under the assumptions of \cref{def:uniform-simplex-MC},
let $[s,T] \subseteq [-\ln(2)/\Lambda(\omega),\ln(2)/\Lambda(\omega)]$. 
Let
$\Lambda(\omega;s,T)\ :=\ \sup_{\tau\in[s,T]}\ \|H(\tau,\omega)\|,$
and assume that $\Lambda(\cdot;s,T)$ has $\sup_q \mathbb{E}_{\omega}\big[\Lambda(\omega;s,T)^q\big]<\infty$. 
Let $\eps,\delta\in(0,1)$ be target accuracy and failure probability. Pick any split $\delta=\delta_{\mathcal L}+\delta_{\mathrm{MC}}$ with $\delta_{\mathcal L},\delta_{\mathrm{MC}}\in(0,1)$. 
Then the probability over $\omega$ of a deviation of at most $\eps$ in the MC approximation to the ordered operator exponential of \cref{{lem:MC-simplex}} is bounded by
\begin{equation}
\label{eq:PAC-final}
\mathbb{P}\!\left(\ \big\| \mathcal{T}e^{-i\int_s^T H(t',\omega) dt'}- \mathrm{TDS}^{\mathrm{MC}}_{\mathcal K,\mathbf{N}}[H](\omega;s,T)\big\| \ \le\ \eps\ \right)\ \ge\ 1-\delta.
\end{equation}
 by choosing the Dyson series truncation order $\mathcal{K} \in \mathcal{O}(\log(1/\eps))$ and
\begin{equation}
\label{eq:Nk-Cheb}
N_k\ \ge\ \frac{2\mathcal{K}^2(T{-}s)^{2k}}{(k!)^2}\ \frac{\mathcal L^{2k}}{\delta_{\mathrm{MC}}\ \eps^2}
\qquad (k=1,\dots,\mathcal{K}),
\end{equation}
where \begin{equation}
\label{eq:L-Markov}
\mathcal L\ :=\sup_q \left(\frac{\mathbb{E}_{\omega}\!\left[\Lambda(\omega;s,T)^q\right]}{\delta_{\mathcal L}}\right)^{\!1/q},
\end{equation}
\end{lemma}
\begin{proof}
By Markov's inequality,
we have that for any non-negative integer $q$
\[
\mathbb{P}\big(\Lambda^q > \mathcal L^q\big)\ \le\ \sup_q \frac{\mathbb{E}_{\omega}[\Lambda^q]}{\mathcal L^q}\ =\ \delta_{\mathcal L},
\]
which yields
\begin{equation}
\label{eq:event-L}
\mathbb{P}\!\left(\ \Lambda(\omega;s,T)\ \le\ \mathcal L\ \right)\ \ge\ 1-\delta_{\mathcal L}.
\end{equation}
On the event $E_{\mathcal L}\ :=\ \{\Lambda(\omega;s,T)\le \mathcal L\}$, $\|Z_k(t,\omega)\|\le \Lambda^k\le \mathcal L^k$ for all $t\in\Delta_k(s,T)$ by submultiplicativity  we have
\begin{equation}
\mathbb{P}\!\left(\|Z_k(t,\omega)\|\le \mathcal L^k\ \right)\ \ge\ 1-\delta_{\mathcal L}.
\end{equation}
By construction in \cref{def:uniform-simplex-MC}, the MC estimator is unbiased (\cref{eq:Ik-MC-properties} in \cref{lem:MC-simplex}).  This means that the MC estimate of the $k^{\mathrm{th}}$ order term in the Dyson series expansion $\widehat{\mathcal{I}}[H](\omega;s,T)$, obeys
\begin{equation}
\label{eq:unbiased}
\mathbb{E}\big[\widehat{\mathcal{I}}_k[H](\omega;s,T)\ \big|\ E_{\mathcal L}\big]
\;=\;
\mathcal{I}_k[H](\omega;s,T).
\end{equation}
Define the centered random matrix
\[
X_k\ :=\ \widehat{\mathcal{I}}_k[H](\omega;s,T)\ -\ \mathcal{I}_k[H](\omega;s,T),
\]
so $\mathbb{E}[X_k\,|\,E_{\mathcal L}]=0$ by \cref{eq:unbiased}. We now bound the tail of the \emph{scalar} nonnegative random variable $\|X_k\|$ via Chebyshev inequality:
\begin{equation}
\label{eq:cheb-1}
\mathbb{P}\Big(\ \|X_k\| \ge \eps_k\ \Big|\ E_{\mathcal L}\Big)\ \le\ \frac{\mathbb{E}\big[\|X_k\|^2\,\big|\,E_{\mathcal L}\big]}{\eps_k^2}.
\end{equation}
The conditional second moment can be bounded by writing the estimator as an average of i.i.d.\ centered terms. Set for any $r$ and sequences of times $t^{(k,r)}$
\[
Y_r\ :=\ \frac{(T{-}s)^k}{k!}\,Z_k\!\big(t^{(k,r)},\omega\big),\qquad
\overline{Y}\ :=\ \frac{1}{N_k}\sum_{r=1}^{N_k} Y_r,\qquad
\mu\ :=\ \mathbb{E}[Y_1\,|\,E_{\mathcal L}]=\mathcal{I}_k[H](\omega;s,T),
\]
so $X_k=\overline{Y}-\mu$ and $\mathbb{E}[Y_r\,|\,E_{\mathcal L}]=\mu$ for each $r$. Then we have from the independence of each $Y_r$
\begin{equation}
\label{eq:HS-variance}
\mathbb{E}\big[\|X_k\|^2\,\big|\,E_{\mathcal L}\big]\;=\;\mathbb{E}\Big\|\frac{1}{N_k}\sum_{r=1}^{N_k} (Y_r-\mu)\Big\|^2
\;=\;\frac{1}{N_k}\,\mathbb{E}\big[\|Y_1-\mu\|^2\,\big|\,E_{\mathcal L}\big]
\;\le\;\frac{1}{N_k}\,\mathbb{E}\big[\|Y_1\|^2\,\big|\,E_{\mathcal L}\big],
\end{equation}
where the last inequality uses $\mathbb{E}\|Z-\mathbb{E}[Z]\|^2\le \mathbb{E}[\|Z\|^2]$. Now on the event $E_{\mathcal L}$,
\[
\|Y_1\| \;=\; \frac{(T{-}s)^k}{k!}\,\|Z_k(t^{(k,1)},\omega)\|
\;\le\; \frac{(T{-}s)^k}{k!}\,\|Z_k(t^{(k,1)},\omega)\|
\;\le\; \frac{(T{-}s)^k}{k!}\,\mathcal L^{k},
\]
by submultiplicativity $ \|Z_k\|\le \prod_{j=1}^k\|H(t_j)\|\le \mathcal L^k$. Therefore
\begin{equation}
\label{eq:Y2-bound}
\mathbb{E}\big[\|Y_1\|^2\,\big|\,E_{\mathcal L}\big]\ \le\ \left(\frac{(T{-}s)^k}{k!}\right)^{\!2}\,\mathcal L^{2k}.
\end{equation}
Combining \cref{eq:cheb-1}, \cref{eq:HS-variance}, and \cref{eq:Y2-bound} yields
\[
\mathbb{P}\!\left(\Big\|\widehat{\mathcal{I}}_k-\mathcal{I}_k\Big\| \ge \eps_k\ \Big|\ \Lambda\le \mathcal L\right)
\ \le\ \frac{1}{N_k\,\eps_k^2}\,\mathbb{E}\!\left[\left\|\frac{(T{-}s)^k}{k!}\,Z_k(t,\omega)\right\|^{\!2}\ \Big|\ \Lambda\le \mathcal L\right]
\ \le\ \frac{(T{-}s)^{2k}}{(k!)^2}\,\frac{\mathcal L^{2k}}{N_k\,\eps_k^2}.
\]
Hence, enforcing \cref{eq:Nk-Cheb} guarantees
\[
\mathbb{P}\!\left(\Big\|\widehat{\mathcal{I}}_k-\mathcal{I}_k\Big\| \ge \eps_k\ \Big|\ \Lambda\le \mathcal L\right)\ \le\ \delta_k.
\]
We then choose $\delta_k = \delta/\mathcal{K}$ and $\eps_k = \eps/(2\mathcal{K})$ and then 
use the union bound over $k=1,\dots,\mathcal{K}$ to find
\[
\mathbb{P}\!\left(\ \sum_{k=1}^\mathcal{K} \Big\|\widehat{\mathcal{I}}_k-\mathcal{I}_k\Big\| \ \le\ \sum_{k=1}^\mathcal{K} \eps_k\ =\ \frac{\eps}{2}\ \Big|\ \Lambda\le \mathcal L\right)\ \ge\ 1-\sum_{k=1}^\mathcal{K} \delta_k\ =\ 1-\delta_{\mathrm{MC}}.
\]
Next we need to consider the truncation error in the Dyson series expansion.  We express this truncation error as $\mathcal{E}_{\mathcal{K}}$,  which is bounded by \cref{thm:TDS-path-K} to be $\eps/2$ for a value of $\mathcal{K}\in \mathcal{O}(\log(1/\eps))$.
Then by the triangle inequality, and~\cref{eq:tdsdef}
\[
\big\|U_H - \mathrm{TDS}^{\mathrm{MC}}_{\mathcal{K},\mathbf{N}}\big\|
\ \le\ \sum_{k=1}^\mathcal{K}\big\|\widehat{\mathcal{I}}_k-\mathcal{I}_k\big\| \ +\ \mathcal E_{\mathcal{K}}(\Lambda)
\ \le\ \sum_{k=1}^\mathcal{K} \big\|\widehat{\mathcal{I}}_k-\mathcal{I}_k\big\| \ +\ \mathcal E_{\mathcal{K}}(\mathcal L)
\ \le\ \eps,
\]
with conditional probability at least $1-\delta_{\mathrm{MC}}$ when $\Lambda\le \mathcal L$. Unconditioning and combining \cref{eq:event-L} yields
\[
\mathbb{P}\!\left(\ \big\|U_H - \mathrm{TDS}^{\mathrm{MC}}_{\mathcal K,\mathbf{N}}\big\| \le \eps\ \right)
\ \ge\ \mathbb{P}(\Lambda\le \mathcal L)\cdot \mathbb{P}\!\left(\sum_k \big\|\widehat{\mathcal{I}}_k-\mathcal{I}_k\big\| \le \frac{\eps}{2}\ \Big|\ \Lambda\le \mathcal L\right)
\ \ge\ (1-\delta_{\mathcal L})(1-\delta_{\mathrm{MC}})\ \ge\ 1-\delta_{\mathcal L}-\delta_{\mathrm{MC}}
\]
and proves \cref{eq:PAC-final} under the assumption that $\Delta_{\mathcal{L}}+ \delta_{\mathrm{MC}}= \delta$.
\end{proof}
This shows that we can use a MC estimate of the integrals in a Dyson series expansion in place of a quadrature formula to deal with the fact that quadrature formulas will not work for stochastic differential equations. Further, it is instructive to note that this bound also holds in cases where $\|H(t)\|$ is unbounded over a set of zero measure.

\section{Quantum implementation and complexity for NSDEs}\label{sec:quantum-impl}

In this section, we provide concrete quantum implementation and complexity of LCHS for the linearized SDE in \cref{eq:LODE}.
The key idea of the quantum algorithm of NSDE is to implement the LCHS representation of its solution via several layers application of LCU.

Our aim for the quantum algorithm is to simulate the Carleman linearized differential equation in~\cref{pre:cl}:
\begin{equation}\label{eq:carlemanDiffReminder}
    \frac{\d{y}}{\d{t}}={A}y + {b}(t)
\end{equation}
where $A$ is a random matrix distributed according to the OU process and where for convenience we define $A=iH+L$ for Hermitian $H,L$ and $L\prec 0$.  Specifically we wish to compute from our quantum algorithm a state $y(t)/\|y(t)\|$ such that $\|\Pi_0y(t) - x(t)\oplus [0]\|\le \eps$ where $\Pi_0$ is the projector onto the first Carleman block and $[0]$ is the zero vector of dimension equal to ${\mathrm{dim}}(y)$ minus the dimension of the first Carleman block: ${\mathrm{dim}}(y) - {\mathrm{dim}}(x)$.  
Our approach to solving this differential equation involves reducing the problem to our stochastic LCHS result on a random variable $\omega$ that specifies the instance of the OU process.

An important assumption that we make in  the following is that the force imposed on OU trajectories will be upper bounded by a constant for all but a small fraction of the evolutions.  Specifically, we quantify this in terms of the variable $\Xi$ via the following OU event
\begin{equation}\label{eqn:OU_trunc_event}
\mathcal{E}_\Xi\ :=\ \left\{\omega: \sup_{t\in[0,T]}\|\mathbf{F}_0(t,\omega)\|\ \le\ \Xi\ \right\}.
\end{equation}
and then define the probability of achieving such a deviation to be bounded below by
 \(\mathbb{P}(\mathcal{E}_\Xi)\ge 1-\delta_\Xi\).  Such bounds can be found by using Gaussian tail bounds in our context.  Note that alternative approaches can be considered using the concept of Hölder continuity of the OU process, which is not used here but is discussed in detail in the appendix.

\begin{definition}\label[definition]{def:prepSel}
    Let $A=iH+L$ be the Carleman linearized form of the differential operator in \cref{eq:carlemanDiffReminder} such that $H(t,\omega)=\sum_p \alpha_{H,p}(t,\omega) V_p$ and $L(t,\omega)=\sum_p \alpha_{L,p}(t) V_p$ for unitary matrices $V_p$ and non-negative $\alpha_{L,p}$ and $\alpha_{H,p}$.  Finally let $\xi$ be an integer that specifies a ``seed'' or index of a particular random process and a particular random selection of times for the Monte-Carlo integral used in LCHS and let $\omega_\xi$ represent the particular random noise instance specified by this index. Assume that block encodings of $H(t)$ and $L(t)$ are provided by the oracles 
  $\mathrm{PREP}_L$, $\mathrm{PREP}_H$, and $\mathrm{SELECT}$. These oracles are defined such that
    \begin{equation*}
        {\mathrm{PREP}}_L\ket{0}\ket{t}\ket{\xi}= \sum_p \frac{\sqrt{\alpha_{L,p}(t,\omega_{\xi}) }}{\sqrt{\alpha_L(t,\omega_{\xi})}}\ket{p}\ket{t}\ket{\xi},\quad {\mathrm{PREP}}_H\ket{0}\ket{t}\ket{\xi}= \sum_p \frac{\sqrt{\alpha_{H,p}(t,\omega_{\xi})}}{\sqrt{\alpha_H(t,\omega_{\xi})}}\ket{p}\ket{t}\ket{\xi}
    \end{equation*}
    and the SELECT operator which applies the element of the operator basis $V_p$ to the target state selected on a state that encodes the index. Here $\xi$ is an index over a finite set of sampled paths and $\omega_\xi$ is the corresponding sample of the Wiener process.
    \[
    {\mathrm{SELECT}}\ket{p} \ket{\psi} = \ket{p}V_p \ket{\psi}.
    \]
    Further, assume that we have access to an oracle ${\mathrm{PREP'}}_L$ and an oracle ${\mathrm{PREP'}}_H$ such that
\begin{align*}
{\mathrm{PREP'}}_L\ket{0}\ket{k_j}\ket{t}\ket{\xi} &= \left(\frac{\sqrt{\alpha_L(t,\omega_\xi)|k_j|}}{\sqrt{\alpha_H(\Xi)+\alpha_L(\Xi) \max_j|k_j|  }}\ket{0} + \sqrt{1-\frac{{\alpha_L(t,\omega_\xi)|k_j|}}{{\alpha_H(\Xi)+\alpha_L(\Xi) \max_j|k_j| }}}\ket{1} \right)\ket{k_j}\ket{t}\ket{\xi}\\
{\mathrm{PREP'}}_H\ket{0}\ket{t}\ket{\xi} &= \left(\frac{\sqrt{\alpha_H(t,\omega_\xi)}}{\sqrt{\alpha_H(\Xi)+\alpha_L(\Xi) \max_j|k_j| }}\ket{0} + \sqrt{1-\frac{{\alpha_H(t,\omega_\xi)}}{{\alpha_H(\Xi)+\alpha_L(\Xi) \max_j|k_j| }}}\ket{1} \right)\ket{t}\ket{\xi}
\end{align*}    
    for $\alpha_H(\Xi)\ge \sup_{\omega \in \mathcal{E}_\Xi}\sup_t \sum_j \alpha_{H,j}(t,\omega)$ and $\alpha_L(\Xi)\ge \sup_{\omega \in \mathcal{E}_\Xi}\sup_t \sum_j \alpha_{L,j}(t,\omega)$.
    We also assume that controlled versions and inverses of these operators can be performed at the cost of a single query.
\end{definition}
In the language of LCU, the PREP oracles can be thought of as a sub-prepare oracle that prepares a weighted distribution over the unitaries in the decomposition of the generator terms and the PREP' operators weight everything appropriately relative to a uniform upper bound used for the normalization.  In cases, such as the sparse access model, these decompositions can be computed in polynomial time and also implemented using a number of queries proportional to the sparsity of the underlying matrix.

\begin{lemma}\label[lemma]{lem:HAMT}
    Let $A=iH+L$ be the Carleman linearized form of the differential operator in \cref{eq:carlemanDiffReminder} such that $H(t,\omega)=\sum_p \alpha_{H,p}(t,\omega) V_p$ and $L(t,\omega)=\sum_p \alpha_{L,p}(t) V_p$ for unitary matrices $V_p$ and non-negative $\alpha_{L,p}$ and $\alpha_{H,p}$.  There exists a quantum algorithm that can implement a unitary $\textup{HAM-T}$ such that the oracle provides a block encoding of the form
\begin{align*}
(\bra{0}\bra{k_j}\bra{t}\bra{\xi}\otimes I) {\textup{HAM-T}} (\ket{0}\ket{k_j}\ket{t}\ket{\xi}\otimes I) = \frac{H(t,\omega_\xi)+|k_j|L(t,\omega_\xi)}{2({\alpha_H(\Xi)+\alpha_L(\Xi) \max_j|k_j| })}
\end{align*}
with zero error using $\mathcal{O}(1)$ queries to the $\mathrm{PREP}$, $\mathrm{PREP'}$ and $\mathrm{SELECT}$ oracles (as well as their adjoints and controlled variants of them). 
\end{lemma}
\begin{proof}
    We construct our block encoding in two phases.  First, note that
    \begin{align}
        {\mathrm{PREP'}}_L\ket{0}({\mathrm{PREP}}_L \ket{0}\ket{k_j}\ket{t} \ket{\xi}) = \frac{1}{\sqrt{\alpha_H(\Xi)+\alpha_L(\Xi) \max_j|k_j| }} \ket{0}\sum_p \sqrt{\alpha_{L,p}(t,\omega_\xi)|k_j|}\ket{p}\ket{k_j}\ket{t}\ket{\xi} + \alpha' \ket{1}\ket{\xi}
    \end{align}
    for some state $\ket{\chi}$ and coefficient $\alpha'$.  This implies from the LCU Lemma~\cite{childs_hamiltonian_nodate} that an $(\alpha_L(\Xi)+\alpha_L(\Xi) \max_j|k_j|, \cdot, 0)$-block encoding of $\sum_{k_j} \sum_t |k_j| L(t,\omega_\xi)\otimes  \ket{t}\!\bra{t}\otimes \ket{\xi}\!\bra{\xi}$ can be constructed given that $L(t,\omega) = \sum_p \alpha_p V_p$ from \cref{def:prepSel}.  
    This process requires $\mathcal{O}(1)$ query operations and incurs no error.
    Similarly, an identical argument provides an algorithm that constructs an $(\alpha_L(\Xi)+\alpha_L(\Xi) \max_j|k_j|, \cdot, 0)$ block encoding of $\sum_{t,\xi} H \otimes \ket{\xi}\!\bra{\xi}\otimes \ket{t}\!\bra{t}$ with no error (by definition) and using $\mathcal{O}(1)$ queries to $\mathrm{PREP}_H$, $\mathrm{PREP'}_H$ and SELECT.

    The above argument shows that the individual block encodings of the two terms can be constructed using $\mathcal{O}(1)$ queries.  The sum of the two block-encodings can be constructed using a further LCU using the results of~\cite{gilyen_quantum_2019} to produce a block-encoding of $H(t,\omega)+|k_j|L(t,\omega)$ with zero error a normalization factor equal to the sum of those of $H$ and $L$ which is simply $2({\alpha_H(\Xi)+\alpha_L(\Xi) \max_j|k_j| }$ as claimed.  We define this block encoding to be \textup{HAM-T}.
\end{proof}
Next we define an oracle that implements the inhomogeneity term for the Carleman linearized equation.
\begin{definition}\label{def:oracle_b}
    Let the oracle $O_b$ be a unitary operator acting on a time register, a path-seed register, an ancilla flag qubit $\ket{0}_a$, and a system register $\ket{0}_s$ such that 
    $$
    O_b \ket{t}\ket{\xi}\ket{0}_a\ket{0}_s = \ket{t}\ket{\xi}\left(\frac{\|\tilde{b}(t, \omega_\xi)\|}{b_{\max}}\ket{0}_a\ket{\widetilde{b}(t,\omega_\xi)}_s + \sqrt{1- \frac{\|\tilde{b}(t, \omega_\xi)\|^2}{b_{\max}^2}}\ket{1}_a \ket{c(t,\omega_\xi)}_s \right) 
    $$ 
    where $\ket{\widetilde{b}(t,\omega)} := \tilde{b}(t,\omega)/\|\tilde{b}(t,\omega)\|$ denotes the normalized state in the system register, and the truncated inhomogeneity is
    $$
    \tilde{b}(t,\omega) = \begin{cases} b(t,\omega), & \text{if $\|b(t,\omega)\|\le b_{\max}$}\\[4pt]
    b_{\max}\,\dfrac{b(t,\omega)}{\|b(t,\omega)\|},& \text{otherwise}
    \end{cases}
    $$
    Here $b$ is the inhomogeneous term from the Carleman linearized equation in \cref{eq:carlemanDiffReminder}, $b_{\max}$ is a norm bound satisfied for all $t\in[0,T]$ and all $\omega\in\mathcal{E}_\Xi$ (cf.~\cref{eqn:OU_trunc_event}), and $\ket{c(t,\omega)}_s$ is an unspecified state orthogonal to $\ket{0}_a$.  We further assume that $O_b^\dagger$ and controlled versions of $O_b$ can be implemented at the cost of a single query to $O_b$.
\end{definition}

We now have an additional oracle which we use to introduce the random times used in the evaluation of the time-dependent Hamiltonian integral.  In particular, the Monte-Carlo integration requires that we choose a sequence of $t^{(r)}_i$ such that $\sum_{i=1}^k H(t^{(r)}_1) \cdots H(t_k^{(r)})$ corresponds to the LCHS Hamiltonian present at degree $k$ in the time-dependent Hamiltonian that appears in the truncated Dyson series expansion.

\begin{definition}\label{def:oracle_Delta_k}
    For each Dyson order $k\in\{0,\ldots,\mathcal{K}\}$ and each Monte Carlo sample index $r\in\{1,\ldots,N_k\}$, let $t_1^{(r)}\le\cdots\le t_k^{(r)}$ denote the order statistics of $k$ i.i.d.\ $\mathrm{Uniform}(0,T)$ random variables (equivalently, a sample from the uniform density on the ordered simplex $\Delta_k$ of \cref{def:uniform-simplex-MC}).  Different sample indices $r\ne r'$ correspond to independent draws.  We define $O_{\Delta_k}$ to be a unitary that prepares a uniform superposition over the $N_k$ Monte Carlo samples:
    $$O_{\Delta_k}\ket{k}\ket{0} = \frac{\ket{k}}{\sqrt{N_k}}\sum_{r=1}^{N_k} \ket{r}\ket{t_1^{(r)},\ldots, t_{k}^{(r)}}\ket{0}_a,$$
    where $\ket{0}_a$ is an ancilla register for subsequent block-encoding operations.  
\end{definition}

Note that the $O_{\Delta_k}$ oracle is independent of the OU path $\omega$ and therefore does not require a $\ket{\xi}$ register.
Further, we assume that there exists a constant $C_{\alpha}$ such that $\alpha_L(\Xi) \le C_{\alpha} \sup_t \|L(t,\omega)\|$ and $\alpha_H(\Xi) \le C_{\alpha} \sup_t \|H(t,\omega)\|$.  In cases where a sparse decomposition is chosen, $C_\alpha =1$~\cite{gilyen_quantum_2019}, but in general a larger block-encoding constant will be needed and this will correspond to $C_{\alpha}>1 $.

Recall that we approximate each time-ordered exponential by LCHS over $\{(k_j,c_j)\}$ and, per LCHS term (Hamiltonian simulation), a MC truncated Dyson series (MC-TDS) over the ordered simplex $\Delta_k$ (\cref{def:uniform-simplex-MC}). The OU dependence of $A(t,\omega)$ is handled coherently by evaluating $L(t,\omega)$ and $H(t,\omega)$ at the sampled times.

We construct a block-encoding of the Carleman linearized dynamics $\mathcal{T}\exp\!\big(i\int_0^T (k_j L+ H)\,dt\big)$ via LCU across Dyson orders and MC samples, and, within each factor, an LCU to combine $k_jL$ and $H$. 
We state our first result about the implementation of the operator below.  Note that here we design our algorithms to work with a block encoding of the Carleman linearized operator.  This is important to remember as the coefficients of the Carleman linearized term will also include contributions from the \emph{stochastic} inhomogeneity terms of the original nonlinear equation.
\begin{lemma}\label[lemma]{lem:HAMT-oracle}
    There exists a quantum algorithm that accepts an input index $\ket{j}$ and a seed  $\xi$, evolution time $|T|\le \ln(2)/(2({\alpha_H(\Xi)+\alpha_L(\Xi) \max_j|k_j| }))$ and constructs a $(1+\mathcal{O}(\eps),\cdot,\eps)$ block encoding of $\mathcal{T}\exp\!\big(i\int_0^T (k_j L(t,\omega)+ H(t,\omega))\,dt\big)$ for any $k_j$ and $\omega$ within error $\eps$  with probability at least $1-\delta$ that further uses $\mathcal{O}(\mathrm{polylog}(N_k))$ additional two-qubit gates for the value of $N_k$ chosen in~\cref{lemma:PAC-MC}.
\end{lemma}

\begin{proof}

Our proof of the block encoding consists of several steps.  The truncated Dyson series using Monte-Carlo integration is of the form from~\cref{def:uniform-simplex-MC}

\begin{equation}
    V_{tds}:=\sum_{k=0}^{\mathcal K} (-i)^k\, \widehat{\mathcal{I}}_k[H](\omega) = \sum_{k=0}^{\mathcal K} (-i)^k\,  \frac{T^k}{k!}\cdot \frac{1}{N_k}\sum_{r=1}^{N_k} (k_j L(t^{(r)}_{k,k},\omega)+ H(t^{(r)}_{k,k},\omega))\cdots (k_j L(t^{(r)}_{1,k},\omega)+ H(t^{(r)}_{1,k},\omega))
\end{equation}
Here we explicitly place a double subscript to keep track of the value of $k$ that the term arises from specifically, $t_{2,3}^{(r)}$ would refer to here as the 2nd smallest term from a set of $k=3$ times uniformly drawn over the interval $[0,T]$.
Our aim is to block-encode this operator.  First, for brevity let us define
\begin{equation}
    \alpha_{\max} := 2({\alpha_H(\Xi)+\alpha_L(\Xi) \max_j|k_j| })
\end{equation}
We then have that
\begin{equation}
    V_{tds}  = \sum_{k=0}^{\mathcal K} (-i)^k\,  \frac{(\alpha_{\max}T)^k}{k!}\cdot \frac{1}{N_k}\sum_{r=1}^{N_k} \frac{(k_j L(t^{(r)}_{k,k},\omega)+ H(t^{(r)}_{k,k},\omega))}{\alpha_{\max}}\cdots \frac{(k_j L(t^{(r)}_{1,k},\omega)+ H(t^{(r)}_{1,k},\omega))}{\alpha_{\max}}
\end{equation}
This can be re-expressed as
\begin{align}
    V_{tds}  &= \left(\sum_{k'} \sqrt{\frac{(\alpha_{\max}T)^{k'}}{k'!}}\bra{k'}\otimes I\right) \frac{1}{N_k} \sum_{r,k} \ket{k}\!\bra{k} \otimes \prod_{p=0}^{k-1} \left(\frac{(k_j L(t^{(r)}_{p+1,k},\omega)+ H(t^{(r)}_{p+1,k},\omega))}{\alpha_{\max}} \right)^{\delta_{p\le k}}\nonumber\\
    &\quad\times \left(\sum_{k''} \sqrt{\frac{(\alpha_{\max}T)^{k''}}{k''!}}\ket{k''}\otimes I\right)
 \end{align}
Using the results of~\cite{gilyen_quantum_2019}, this result can be written as a sum and product of block encodings which in turn can be combined into a single block encoding.  Specifically, the sum of the block encodings can be implemented via a linear combination of unitaries circuit.
By definition HAM-T block encodes $(k_j L(t^{(r)}_k,\omega)+ H(t^{(r)}_k,\omega))$ with a constant $\alpha_{\max}$.  Thus the combination of them can be block encoded using $1$ call to controlled HAM-T, two calls to a comparator to determine whether $p\le \delta$ to control this block encoding $\mathcal{O}(\log(\mathcal{K}))$ qubits and the preparation of the Monte-Carlo integrator register which requires $\mathcal{O}(\log(N_k))$ qubits and polylogarithmic gates in $N_k$.  This requires one query to a multiply-controlled HAM-T for each of the $\mathcal{K}$ time registers.  This requires a single controlled HAM-T if an ancilla is used to hold the output from the comparison.  The comparison result is then computable with a Toffoli network on a logarithmic number of qubits.  Thus $\mathcal{O}(\mathcal{K})$ queries to HAM-T and a polynomial sized circuit to implement the comparisons.  Further two queries to $\mathcal{O}_{\Delta_k}$ are needed to compute the times needed for the interval.

The product of block encodings can be implemented by using a permutation to ensure that the garbage elements of two block encodings are multiplied by zero when forming the block encoding of the product of the encoded matrices.  This leads to a block encoding constant that is the product of the underlying block encoding constants.  Summing the results, we have that the block-encoding constants using the results of~\cite{gilyen_quantum_2019}
\begin{align}
    \sum_{k=0}^{\mathcal{K}} \frac{(\alpha_{\max} T)^{2k}}{k!} \le e^{2\alpha_{\max} T}.
\end{align}
 Thus by choosing $T\le \ln(2)$ we guarantee that the success probability is at least $1/4$.  Further, the result $\epsilon$ approximates unitary dynamics because $H+k_j L$ is Hermitian by assumption if the values of \cref{lemma:PAC-MC} are adopted. Thus under these assumptions fixed point oblivious amplitude amplification can be used to boost the block encoding constant $1- \mathcal{O}(\epsilon)$ using $\mathcal{O}(\log(1/\epsilon))$ repetitions of the fundamental block encodings.  These values will be valid with probability at least $1-\delta$ over resolutions of $\omega$. Thus we can prepare the block encoding with a normalization constant $1-\mathcal{O}(\epsilon)$ using $\mathcal{O}(1)$ queries to HAM-T and $\mathcal{O}(1)$ queries to $O_{\Delta_k}$.

 A single query to HAM-T requires $\mathcal{O}(1)$ queries to our fundamental oracles using the result of \cref{lem:HAMT}: $\mathrm{PREP}$, $\mathrm{PREP'}$ and $\mathrm{SELECT}$ oracles.  Thus the result holds.  The operation $\mathcal{O}_{\Delta k}$ can be implemented within error $\mathcal{O}(\epsilon)$ using a Hadamard circuit to draw $\mathcal{K}$ uniform samples from $[0,T]$ using a fixed point representation for values between $[0,2T]$.  Then conditioned on the value of $k$, the values are swapped with a register containing the bit representation of $2T$ if the index is greater than $k$.  We then sort the $\mathcal{K}$ registers based on their values.  The corresponding registers for $t_1^{(r)},\ldots,t_k^{(r)}$ then by construction will contain the desired order statistics.  This requires a number of two-qubit gates that is at most in $\mathcal{O}({\rm polylog}(1/\epsilon,1/\delta))$.  Then  using the fact that
 $$
 \log(N_k) \in O\!\left(\log\left(\frac{\mathcal{K}^2(T{-}s)^{2k} }{ (k!)^2 }\ \frac{\sup_q\left(\mathbb{E}_\omega\left({\sup_{\tau\in[s,T]}\ \|H(\tau,\omega)\|^q/\delta}\right)\right)^{2k/q} }{ \delta\ \eps^2 }\right)\right) \supset \rm{\polylog}(1/\epsilon,1/\delta).
 $$ which implies that and thus justifies the quoted scaling after taking $\alpha_{\max} = \Lambda(\omega;0,T)$.

\end{proof}

We now present a quantum complexity analysis for solving the NSDEs. The overall quantum algorithm comprises two pillars: Carleman linearizing the NSDEs and LCHS of the linearized SDEs. The former contributes to the quantum complexity via the block-encoding factor $\alpha_A$ related to the norm of the Carleman matrix $A$. The latter carries over most of the deterministic LCHS cost other than time discretization of the stochastic inhomogeneous term and the truncated Dyson series via MC sampling. However, the MC sampling overhead does \textit{not} contribute to the HAM-T query cost and only adds $O\left(\log(N_S)\right)$ overhead cost in ancillary registers, where $N_S:=\sum_{k=0}^{\mathcal K} N_k$ is the total number of MC samples across all Dyson orders. It contributes to the quantum gate complexity that we analyze later.

\begin{lemma}\label[lemma]{thm:complexity_ou_cart}
Consider the general quadratic ODE of \cref{eq:main_system}
driven by an OU process of \cref{eq:general_ou}. Under the definitions and assumptions in \cref{def:qdns} and \cref{def:uniform-simplex-MC}, let $\alpha:=\alpha_H(\Xi)+\alpha_L(\Xi)\max_j|k_j|$ denote the block-encoding normalization constant.
With probability at least $1-\delta$ over $\omega \in \Omega$, the algorithm prepares an $\eps$-approximation of the normalized state $\ket{U(T)}$ using, for each Dyson order $k\in\{1,\ldots,\mathcal{K}\}$,
$N_k \in O\!\left(\frac{\mathcal{K}^2(T{-}s)^{2k} }{ (k!)^2 }\ \frac{\sup_q\left(\mathbb{E}_\omega\left({\sup_{\tau\in[s,T]}\ \|H(\tau,\omega)\|^q/\delta}\right)\right)^{2k/q} }{ \delta\ \eps^2 }\right)$
Monte Carlo samples, with
\begin{align}
\widetilde{\mathcal{O}}\!\left( \frac{\|U_{\mathrm{in}}\|+
\mathbb{E}_\omega\!\big[\|b\|_{L^1}\big]
}{\|U(T)\|}\alpha T\,\left(\log\frac{1}{\eps}\right)^{1+1/\beta}\right) 
\label{eq:HAMT_queries_cart} 
\end{align}
queries to ${\mathrm{PREP}}_L$, ${\mathrm{PREP}}_H$, ${\mathrm{PREP'}}_L$, ${\mathrm{PREP'}}_H$, ${\mathrm{SELECT}}$,
        and queries to the state preparation oracles $O_u$ and $O_{b}$.
        The algorithm  also requires $\mathcal{O}\!\left( \frac{\|U_{\mathrm{in}}\|+
\mathbb{E}_\omega\!\big[\|b\|_{L^1}\big]
}{\|U(T)\|}\alpha T\cdot\mathrm{polylog}(N_k)\right)$ two-qubit operations.
\end{lemma}

\begin{proof}
The solution to the Carleman-linearized differential equation reads
\begin{equation}
    \mathbf{y}(t) = \mathcal{T}e^{-\int_0^t \mathbf{A}(s) \ud s} \mathbf{y}_0 + \int_0^t \mathcal{T}e^{-\int_s^t \mathbf{A}(s') \ud s'} \mathbf{b}(s) \ud s,
\end{equation}
Thus we can implement the normalized solution to the differential equation by summing the solutions to both the homogeneous and the inhomogeneous terms in turn.

  We start with counting the HAM-T queries for implementing the homogeneous term.
  By \cref{def:uniform-simplex-MC} and \cref{lemma:PAC-MC}, choosing the Dyson series truncation order $\mathcal K$ as in~\cref{thm:TDS-path-K} and MC sampling size $N_k$ as in~\cref{eq:Nk-Cheb} ensures
\begin{equation}
\label{eq:PAC-final2}
\mathbb{P}\!\left(\ \big\| \mathcal{T}e^{-i\int_s^T (k_j L(t',\omega)+H(t',\omega))\, dt'}- \mathrm{TDS}^{\mathrm{MC}}_{\mathcal K,\mathbf{N}}[k_jL+H](\omega;s,T)\big\| \ \le\ \eps\ \right)\ \ge\ 1-\delta.
\end{equation}
The core idea behind implementing this is the time-dependent simulation algorithm of~\cref{lem:HAMT-oracle} for a short time evolution 
\begin{equation}\tau\le \ln(2)/\Lambda(\omega,s,s+\tau) \le\ln(2)/(2({\alpha_H(\Xi)+\alpha_L(\Xi) \max_j|k_j| }))
\end{equation}
using a number of queries to the fundamental oracles that is in $\mathcal{O}(\log(1/\epsilon))$.  This needs to be repeated using the LCHS formalism a number of times that is in
\begin{equation}
    \mathcal{O}(KT/\tau) \subseteq \mathcal{O}(K({\alpha_H(\Xi)+\alpha_L(\Xi) \max_j|k_j| })T) = \mathcal{O}(\log^{1/\beta}(1/\epsilon)({\alpha_H(\Xi)+\alpha_L(\Xi) \max_j|k_j| })T), 
\end{equation}
which implies from Box 4.1 of~\cite{nielsen2010quantum}  that the total error per step needs to be shrunk by a corresponding factor leading to a cost per segment of
\begin{align}
    \mathcal{O}(\log(1/\epsilon)) \rightarrow \mathcal{O}(\log(\log^{1/\beta}(1/\epsilon)({\alpha_H(\Xi)+\alpha_L(\Xi) \max_j|k_j| })T/\epsilon) \subset \widetilde{O}\left( \log(({\alpha_H(\Xi)+\alpha_L(\Xi) \max_j|k_j| })T/\epsilon)\right)
\end{align}
Applying oblivious amplitude amplification as discussed in \cite[Corollary~4]{low_hamiltonian_2019}, for any $\eps_{\mathrm{TDS}} > 0$, the query cost for simulating the Hamiltonian in the LCHS formula, including oblivious amplitude amplification, is

\begin{equation}
\widetilde{\mathcal{O}}\!\left( \alpha K T\log\frac{1}{\epsilon}\right)=\widetilde{\mathcal{O}}\!\left( \alpha T \log^{1+1/\beta}(1/\epsilon) \right).
\end{equation}
Here $K$ is the truncation parameter in the LCHS formula and is chosen to be in $\mathcal{O}( \log^{1/\beta}(1/\eps_{\mathrm{TDS}}) )$ according to~\citet{an_quantum_2025}. 
So the overall HAM-T query cost for implementing the homogeneous term becomes 
\begin{equation}\label{eq:LW_cost}
\widetilde{\mathcal{O}}\!\left( \alpha T\,\Big(\log\frac{1}{\eps_{\mathrm{TDS}}}\Big)^{1+1/\beta}\right). 
\end{equation}

For the inhomogeneous term, the Duhamel integral $\int_0^T \mathcal{T}e^{-\int_s^T A(s')\,ds'}\,b(s)\,ds$ is discretized via LCHS over start times $\{s_{j'}\}$ and, for each propagator $\mathcal{T}e^{-\int_{s_{j'}}^T A\,ds'}$, a MC-TDS block encoding identical in structure to the homogeneous case. Each such propagator block encoding therefore also costs \cref{eq:LW_cost} queries.

We now determine the parameter choices for designed error $\eps$ for given probability $\delta$. Let $v$ be the unnormalized vector produced after combining the homogeneous and inhomogeneous contributions by a final one-qubit LCU (rotation $R$ with weights $\|U_{\mathrm{in}}\|$ and $\|c'\|_1$):
\begin{equation}\label{eq:v_cart}
v=\left(\sum_{j=0}^{M-1} c_j\,W_j\right)\|U_{\mathrm{in}}\|\ket{U_{\mathrm{in}}}
+\sum_{j'=0}^{M'-1}\sum_{j=0}^{M-1} c'_{j'}c_j\,W_{j,j'}\,\ket{b(s_{j'},\omega)}.
\end{equation}
The error decomposes as
\begin{align}\label{eq:error_decomp_cart}
\|v-U(T)\|
&\le \underbrace{\left\|\sum_j c_j W_j U_{\mathrm{in}} - \sum_j c_j U_j(T) U_{\mathrm{in}}\right\|}_{\text{homogeneous MC-TDS}}\
+\ \underbrace{\left\|\sum_{j'} c'_{j'}\sum_j c_j W_{j,j'} b(s_{j'}) - \sum_{j'} c'_{j'}\sum_j c_j U_j(T,s_{j'}) b(s_{j'})\right\|}_{\text{inhomogeneous MC-TDS}}\\
&\quad +\ \underbrace{\left\|\sum_j c_j U_j(T) U_{\mathrm{in}} - \mathcal{T}e^{-\int_0^T A}\,U_{\mathrm{in}}\right\|}_{\text{LCHS discretization (homog.)}}\
+\ \underbrace{\left\|\sum_{j'} c'_{j'}\sum_j c_j U_j(T,s_{j'}) b(s_{j'}) - \int_0^T \mathcal{T}e^{-\int_s^T A}\,b(s)\,ds\right\|}_{\text{LCHS discretization (inhom.)}}.\nonumber
\end{align}
By \cref{lemma:PAC-MC} and a union bound over the indices, the first two terms are bounded by $\|c\|_1\|U_{\mathrm{in}}\|\,\eps_1$ and $\|c\|_1\|c'\|_1\,\eps_2$, respectively, with probability at least $1-\delta_{\mathcal L}-\delta_{\mathrm{MC}}$, where $\delta_{\mathcal L}$ is the failure probability of the LCHS discretization and $\delta_{\mathrm{MC}}$ is the failure probability of the MC-TDS approximation (both chosen so that $\delta_{\mathcal L}+\delta_{\mathrm{MC}}\le \delta$). We set $\eps_1=\eps_2=\eps_{\mathrm{TDS}}$.
The last two terms (LCHS discretization errors) are bounded by $2\eps_3$ by choosing the LCHS sizes $N_U$ and $M'$ as in~\cref{eq:homo-M} and \cref{eq:M_choice_lemma}, respectively. Thus
\begin{equation}\label{eq:error_v_cart}
\|v-U(T)\|\ \le\ \|c\|_1\|U_{\mathrm{in}}\|\,\eps_1\ +\ \|c\|_1\|c'\|_1\,\eps_2\ +\ 2\eps_3.
\end{equation}
The error in the normalized state satisfies
\begin{equation}\label{eq:state_error_cart}
\left\|\ket{v}-\ket{U(T)}\right\|\ \le\ \frac{2}{\|U(T)\|}\,\|v-U(T)\|\ \le\ \frac{2\|c\|_1\|U_{\mathrm{in}}\|}{\|U(T)\|}\eps_1\ +\ \frac{2\|c\|_1\|c'\|_1}{\|U(T)\|}\eps_2\ +\ \frac{4}{\|U(T)\|}\eps_3.
\end{equation}
To achieve a total error $\eps$, pick
\begin{equation}\label{eq:eps123_cart}
\eps_1=\frac{\|U(T)\|}{8\|c\|_1\|U_{\mathrm{in}}\|}\,\eps,\quad
\eps_2=\frac{\|U(T)\|}{8\|c\|_1\|c'\|_1}\,\eps,\quad
\eps_3=\frac{\|U(T)\|}{8}\,\eps,
\end{equation}
which fixes $\eps_{\mathrm{TDS}}$ in~\cref{eq:Nk-Cheb}.

The good subspace amplitude in the final output is proportional to $\|v\|/\big(\|c\|_1(\|U_{\mathrm{in}}\|+\|c'\|_1)\big)$. From~\cref{eq:error_v_cart} and the triangle inequality,
\begin{equation}\label{eq:v_norm_lb_cart}
\|v\|\ \ge\ \|U(T)\|-\big(\|c\|_1\|U_{\mathrm{in}}\|\eps_1+\|c\|_1\|c'\|_1\eps_2+2\eps_3\big)\ \ge\ \|U(T)\|\,(1-\eps/2).
\end{equation}
Therefore, the number of repetitions needed to boost success to a constant is
\begin{equation}\label{eq:repetitions_cart}
\mathcal{O}\!\left(\frac{\|c\|_1(\|U_{\mathrm{in}}\|+\|c'\|_1)}{\|v\|}\right)\ =\ \mathcal{O}\!\left(\frac{\|U_{\mathrm{in}}\|+\|b\|_{L^1}}{\|U(T)\|}\right).
\end{equation} 
Multiplying the  HAM-T cost  per repetition by the number of repetitions given by~\cref{eq:repetitions_cart} and setting $\eps_{\mathrm{TDS}}$ proportional to $\eps$ yields the \textit{pathwise} overall HAM-T query complexity
\begin{equation}\label{eq:pathwise_HAMT_cost}
\widetilde{\mathcal{O}}\!\left(\frac{\|U_{\mathrm{in}}\|+\|b\|_{L^1}}{\|U(T)\|}\alpha T\,\left(\log\frac{1}{\eps}\right)^{1+1/\beta}\right).
\end{equation}
Note that the number of queries to our fundamental oracles per query to ${\textup{HAM-T}}$ is in $\mathcal{O}(1)$ as argued above in the proof of \cref{lem:HAMT}.  This means that the HAM-T query complexity is asymptotically the same as our final query complexity.  Then,
after taking the expectation over the sampling paths $\omega$, we arrive at the expected query complexity in \cref{eq:HAMT_queries_cart} and completes the proof.  

Finally, we need to consider the gate complexity incurred in the non-query operations in the algorithm.  These specifically involve the comparisons performed in HAM-T, which are arithmetic operations performed on $\mathcal{O}(\log_2(N_k))$ qubits via~\cref{lem:HAMT}.  This gives the final gate complexity of $\mathcal{O}\!\left( \frac{\|U_{\mathrm{in}}\|+\mathbb{E}_\omega[\|b\|_{L^1}]}{\|U(T)\|}\alpha T\cdot\mathrm{polylog}(N_k)\right)$ two-qubit operations as stated.
\end{proof}

We now give a probabilistic bound for the Carleman matrix norm $\|A\|$ to make the dependence of block-encoding factor $\alpha$ in \cref{thm:complexity_ou_cart} explicit in terms of Carleman linearization truncation and OU parameters.

\begin{theorem}\label{thm:main}
  Under the assumptions of \cref{thm:complexity_ou_cart}, and assuming that a block encoding is chosen such that for all instances of the noise, the block encoding constant $\alpha_A$ of the Carleman matrix $\mathbf{A}_N$ truncated at order $N$ and its normalization constant $C_\alpha$ satisfies  $\alpha_A=C_\alpha \sup_t\|\mathbf{A}_N\|$,
we have that there exists an algorithm that prepares an $\eps$-approximation of the normalized state $\ket{U(T)}$ with probability at least $1-\delta$ that uses
\begin{align} \label{eq:query_HAM-T_alpha}
N_Q\in\widetilde{\mathcal{O}}\!\left(
\frac{\|U_{\mathrm{in}}\|+\,\sqrt{\ \frac{T \|\boldsymbol{\Sigma}\|_{F}^{2}}{2\,\lambda_{\min}\delta}\,\left(\,T\ -\ \frac{1-e^{-2\lambda_{\min} T}}{2\,\lambda_{\min}}\,\right)}}{\|U(T)\|}
\;N C_{\alpha}\left(\|\mathbf{F}_1\|+\|\mathbf{F}_2\|+\sqrt{\frac{1-e^{-2\lambda_{\min} T}}{2\lambda_{\min}\delta}\,\|\boldsymbol{\Sigma}\|_{F}^{2}}\right)\;
T\;\left(\log\frac{1}{\eps}\right)^{1+1/\beta}
\right)
\end{align}
queries to the oracles ${\mathrm{PREP}}_L, {\mathrm{PREP}}_H,{\mathrm{PREP'}}_L, {\mathrm{PREP'}}_H$ and SELECT as well as the oracles for preparing the initial state and $b(t)$.  The algorithm further requires an additional number of two-qubit gates that do not depend on the inputs encoded in the oracles for $H,L$ that scale as  
\begin{equation}
\mathcal{O}\left(N_Q\mathrm{polylog}\left(N_k \right) \right)
\end{equation}
two-qubit operations.
Further, here $\lambda_{\min}$ is the smallest eigenvalue of the drift matrix $\boldsymbol{\Theta}$ of the OU process defined in \cref{def:OU-LPDE} and $N_k \in O\left(\frac{\mathcal{K}^2(T{-}s)^{2k} }{ (k!)^2 }\ \frac{\sup_q\left(\mathbb{E}_\omega\left({\sup_{\tau\in[s,T]}\ \|H(\tau,\omega)\|^q/\delta}\right)\right)^{2k/q} }{ \delta\ \eps^2 }\right)$ is the MC integration size for Dyson series truncation $\mathcal{K}$ in~\cref{lemma:PAC-MC}.  

\end{theorem}
\begin{proof}
  The pathwise (fixed sampling path $\omega$) norm of the Carleman matrix at truncation order $N$ in \cref{pre:cl} can be bounded as
\begin{equation}\label{eq:M-norm-form}
\|\mathbf{A}_N(\omega, t)\| = \biggl\|   
\sum_{j=2}^N|j\rangle\langle j-1|\otimes \mathbf{A}_{j-1}^j + 
\sum_{j=1}^N|j\rangle\langle j|\otimes A_{j}^j+
\sum_{j=1}^{N-1}|j\rangle\langle j+1|\otimes \mathbf{A}_{j+1}^j\biggr\| \le N(\|\mathbf{F}_0(\omega, t)\|+ \|\mathbf{F}_1\|+\|\mathbf{F}_2\|),
\end{equation}
using the fact from~\cref{pre:cl} that
\begin{align}
\mathbf{A}_{j+1}^j &= \mathbf{F}_2\otimes I^{\otimes j-1}+I\otimes \mathbf{F}_2\otimes I^{\otimes j-2}+\cdots+I^{\otimes j-1}\otimes \mathbf{F}_2, \label{eq:tensor22} \\
\mathbf{A}_j^j &= \mathbf{F}_1\otimes I^{\otimes j-1}+I\otimes \mathbf{F}_1\otimes I^{\otimes j-2}+\cdots+I^{\otimes j-1}\otimes \mathbf{F}_1, \label{eq:tensor12} \\
\mathbf{A}_{j-1}^j &= \mathbf{F}_0(t)\otimes I^{\otimes j-1}+I\otimes \mathbf{F}_0(t)\otimes I^{\otimes j-2}+\cdots+I^{\otimes j-1}\otimes \mathbf{F}_0(t). \label{eq:tensor02}
\end{align}
Then, a probabilistic bound on $\|\mathbf{A}\|$ can be established via $\mathbb{P}\bigl(\|\mathbf{F}_{0}(t)\| < \Delta_{F_0}\bigr)$ in \cref{eq:Prob-F0} under the boundedness assumption of the OU process by \cref{eqn:OU_trunc_event}. 
By the pathwise inequality in \cref{eq:M-norm-form} we have that for any threshold value $\Delta_{F_0}>0$, 
\[
\bigl\{\omega:\|\mathbf{F}_{0}(\omega,t)\| < \Delta_{F_0}\bigr\} \;\subseteq\; \bigl\{\omega:\|\mathbf{A}_N(\omega,t)\| \le N(\Delta_{F_0} + \|\mathbf{F}_1\|+\|\mathbf{F}_2\|)\bigr\}.
\]
Hence, for any threshold gap $\Delta_A$ chosen such that $\Delta_A  := N\Delta_{F_0} + N(\|\mathbf{F}_1\|+\|\mathbf{F}_2\|) > 0$ gives
\begin{equation}\label{eq:DeltaAFbd}
\mathbb{P}\bigl(\|\mathbf{A}_N(\omega,t)\| < \Delta_A\bigr) \;\ge\; \mathbb{P}\bigl(\|\mathbf{F}_{0}(t)\| < \Delta_{F_0}\bigr).
\end{equation}

Invoking the probabilistic bound \cref{eq:Prob-F0},
\[
\mathbb{P}\bigl(\|\mathbf{F}_{0}(t)\| < \Delta_{F_0}\bigr) \;\ge\; 1 - \frac{\frac{1-e^{-2\lambda_{\min} t}}{2\lambda_{\min}}\,\|\boldsymbol{\Sigma}\|_{F}^{2}}{\Delta_{F_0}^2}
\;=\; 1 - \frac{\frac{1-e^{-2\lambda_{\min} t}}{2\lambda_{\min}}\,\|\boldsymbol{\Sigma}\|_{F}^{2}}{\bigl(\Delta_A/N - \|\mathbf{F}_1\|-\|\mathbf{F}_2\|\bigr)^2},
\]
Then invoking the fact that the probability of a large deviation in $\|\mathbf{A}_N\|$ bounds that of $\|\mathbf{F}_0(t)\|$ via~\cref{eq:DeltaAFbd} \begin{equation}\label{eq:P_alphaA_DeltaA}
\mathbb{P}\bigl(\|\mathbf{A}_N(\omega,t)\| < \Delta_A \bigr)\;\ge\;1 - \frac{\frac{1-e^{-2\lambda_{\min} t}}{2\lambda_{\min}}\,\|\boldsymbol{\Sigma}\|_{F}^{2}}{\bigl(\Delta_A/N - \|\mathbf{F}_1\|-\|\mathbf{F}_2\|\bigr)^2}.
\end{equation}

Let $\alpha_A=C_\alpha \sup_t\|\mathbf{A}_N\|$ where $\alpha_A$ is the block-encoding constant for $A$ and $C_\alpha\ge 1$ representing the ratio between the block encoding constant and the matrix norm of the block encoded operator, then the probability of a large total deviation in the block-encoding constant $\alpha_A$ over the random variable $\omega$ can be given by
\begin{equation}\label{eq:P-alpha_Delta-alpha}
	\mathbb{P}\bigl(\alpha_A < \Delta_\alpha \bigr)\;\ge\;1 - \underbrace{\frac{\frac{1-e^{-2\lambda_{\min} t}}{2\lambda_{\min}}\,\|\boldsymbol{\Sigma}\|_{F}^{2}}{\bigl(\Delta_\alpha/(C_\alpha N) - \|\mathbf{F}_1\|-\|\mathbf{F}_2\|\bigr)^2}}_{\delta/3} = 1-\delta/3
\end{equation}
by the substitution $\Delta_\alpha = C_\alpha \Delta_A$ in \cref{eq:P_alphaA_DeltaA}, where $\delta$ is a given confidence level. We choose $\delta/3$ here so that subsequent contributions to the failure probability are bounded above by $2\delta/3$.  
Solving for $\Delta_\alpha$ in~\cref{eq:P-alpha_Delta-alpha} yields 
\begin{equation}\label{eq:DeltaA}
\Delta_\alpha = C_\alpha N\left(\|\mathbf{F}_1\|+\|\mathbf{F}_2\| + \sqrt{\frac{3}{\delta}\cdot\frac{1-e^{-2\lambda_{\min} t}}{2\lambda_{\min}}\,\|\boldsymbol{\Sigma}\|_{F}^{2}} \, \right).
\end{equation}
Since $\frac{1-e^{-2\lambda_{\min}t}}{2\lambda_{\min}}$ is increasing in $t$, the worst case is $t=T$, so
\begin{equation}\label{eq:alphaA}
\alpha_A \le C_\alpha\,N\left(\|\mathbf{F}_1\|+\|\mathbf{F}_2\| + \sqrt{\frac{3}{\delta}\cdot\frac{1-e^{-2\lambda_{\min} T}}{2\lambda_{\min}}\,\|\boldsymbol{\Sigma}\|_{F}^{2}} \, \right)
\quad\text{with probability at least }1-\delta/3.
\end{equation}

Next we turn our attention to
bounding \(\mathbb{E}_\omega[\|b\|_{L^1}]\) in the query complexity estimate in~\cref{thm:complexity_ou_cart} via its relation with \(V\) of \cref{eq:VE} and provide a high-probability bound by the Markov inequality.
By \cref{eq:VE}, we obtain
\begin{equation}\label{eq:l1-second-moment}
\mathbb{E}_\omega\!\big[\|b\|_{L^1}^2\big]
= \mathbb{E}_\omega\!\Big[\Big(\int_0^T \|b(s)\|\,ds\Big)^{\!2}\,\Big]
\le T\int_0^T \mathbb{E}_\omega\!\big[\|b(s)\|^2\big]\,ds
= T^2\,V.
\end{equation}
Applying Markov's inequality to the nonnegative random variable \(\|b\|_{L^1}^2\) gives
\begin{equation}\label{eq:l1-markov}
\mathbb{P}\!\left(\ \|b\|_{L^1}^2\ \ge\ \frac{3\,T^2\,V}{\delta}\ \right)
\ \le\ \frac{\mathbb{E}[\|b\|_{L^1}^2]}{3\,T^2 V/\delta}
\ \le\ \frac{\delta}{3},
\end{equation}
which, after taking the square root, implies
\begin{equation}\label{eq:l1-highprob}
\|b\|_{L^1}\ \le\ T\sqrt{\frac{3V}{\delta}}
\quad\text{with probability at least }1-\delta/3.
\end{equation}
Thus from the union bound and \cref{eq:P-alpha_Delta-alpha} we obtain
\begin{equation}
    \mathbb{P}\!\left(\left[\|b\|_{L^1} \le T\sqrt{\frac{3V}{\delta}}\right]\ \wedge\ \left[\alpha_A \le C_\alpha N\!\left(\|\mathbf{F}_1\|+\|\mathbf{F}_2\| + \sqrt{\frac{3}{\delta}\cdot\frac{1-e^{-2\lambda_{\min} T}}{2\lambda_{\min}}\,\|\boldsymbol{\Sigma}\|_{F}^{2}} \right)\right]\right) \ge 1-2\delta/3.\label{eq:unionBd1}
\end{equation}

\cref{thm:complexity_ou_cart} states that an $\eps$-approximation to the normalized state $\ket{U(T)}$ can be prepared with probability of failure at most $\delta/3$, and that the number of queries to all oracles in the problem is bounded above by
\begin{align}
N_Q\in\widetilde{\mathcal{O}}\!\left( \frac{\|U_{\mathrm{in}}\|+
\mathbb{E}_\omega\!\big[\|b\|_{L^1}\big]
}{\|U(T)\|}\alpha_A T\,\left(\log\frac{1}{\eps}\right)^{1+1/\beta}\right),
\end{align}
where the dependence on $\delta$ enters through $N_k$, which only explicitly appears in the two-qubit gate complexity.   We then have that for all but a $2\delta/3$ fraction of the results that the values of $\|b\|_{L^1}$ and $\alpha_A$ are bounded above by the values given in~\cref{eq:unionBd1}, which implies that with probability at least $1-2\delta/3$ over the realizations of $\omega$
\begin{equation}
N_Q\in\widetilde{\mathcal{O}}\!\left( \frac{\|U_{\mathrm{in}}\|+
T\sqrt{V}/\delta
}{\|U(T)\|}\left( C_\alpha N\left(\|\mathbf{F}_1\|+\|\mathbf{F}_2\| + \sqrt{\frac{\frac{1-e^{-2\lambda_{\min} T}}{2\lambda_{\min}}\,\|\boldsymbol{\Sigma}\|_{F}^{2}}{\delta/2}} \, \right)\right)T\,\left(\log\frac{1}{\eps}\right)^{1+1/\beta}\right) \label{eq:queryBd}
\end{equation}
Then from the union bound, we have that the total probability of failure over the algorithm and the resolutions of $\omega$ is upper bounded by $\delta_{tot} \le \delta/3 + 2\delta/3 = \delta$ showing that these choices lead to a total failure probability of $\delta$ as claimed.  The claimed query complexity then follows from \cref{lem:MC_error_inhomo} by substituting the upper bound on $V$ contained in the lemma into~\cref{eq:queryBd}.  

The gate complexity is multiplicative with the query complexity per step as noted in~\cref{lemma:PAC-MC}.  The gate complexity then differs only by factors that are polylogarithmic in $N_k$ which leads to the final result.  

\end{proof}

This result shows that the overall simulation cost in terms of the number of queries to our fundamental oracles scales, up to polylogarithmic factors, quadratically in the evolution time and polylogarithmically with the error tolerance $\eps$.  This latter scaling follows from the use of LCHS constructions for the integral and we expect it cannot be substantially improved.  The former scaling at first may seem to be a major deficiency of the algorithm.  After all, near-linear time quantum algorithms exist for dissipative dynamics.  However, as the cumulative impact on the standard deviation of the OU process increases, linear time scaling is not expected for an SDE of this form, and this matches other quadratic scaling results seen for open systems simulations~\cite{pocrnic2025quantum}.  For this reason, the results here are not clearly generalizable.  It should be noted that the optimal $\beta$ scaling can be improved using the methods of~\cite{low_optimal_2025} to achieve logarithmic scaling with $1/\eps$ through the improved choice of contours in LCHS.

\section{Conclusion}\label{sec:conclusion}
This work has provided a quantum algorithm for solving nonlinear stochastic differential equations, with a definite initial condition, for systems that follow a multivariate Ornstein-Uhlenbeck process.  The central challenge of this work stems from the nature of the driving Wiener noise process, which is continuous and nowhere differentiable.  This creates issues with the use of existing methods for solving such problems with LCHS because the quadratures used to estimate the integrals are no longer well defined.  Our work addresses this by providing a formalism that we show allows nonlinear stochastic differential equations to be solved using Carleman linearization techniques under the assumptions of weak nonlinearity and dissipative dynamics.  We show that the dissipative assumption and the weakness of the nonlinearity can be relaxed by subtracting an approximate equilibrium distribution for inhomogeneous differential equations that reach a steady state.

We find that for multivariate OU processes, for constant failure probability the number of queries to the oracles (that provide us with block encodings of the linear, nonlinear, and stochastic parts of the differential equations) scales, for a fixed differential operator, with the evolution time $T$, solution norm $\|U(T)\|$, and error tolerance $\eps$ as 
\begin{equation}
    \widetilde{\mathcal{O}}\left(\frac{T^2 \log^{1+1/\beta}(1/\eps)}{\|U(T)\|} \right)
\end{equation}
where $\beta$ is a universal constant specific to the contour chosen for LCHS\@.   The scaling of the algorithm with all parameters is given in~\cref{thm:main}, which in turn shows that our algorithm can efficiently sample from the normalized solutions of the corresponding stochastic differential equations under reasonable assumptions as specified in detail in~\cref{thm:main}.  
At first glance, this scaling seems substantially worse than the near-linear scaling with $T$ that is achievable for linear differential equations.  However, for OU processes, the variance of the noise term tends to grow with time.  This effect means that we cannot expect that linear scaling with $T$ is achievable, although the ultimate limitations of more general stochastic differential equations remains unstudied.

Further improvements in the performance of simulation schemes may be attainable by combining our approach with the results of~\cite{bravyi2025quantum} which provides an alternative strategy for simulating stochastic differential equations that does not require Carleman linearization but also does not permit the use of a definite initial configuration.  Providing methods that smoothly transition between our approach and that of~\cite{bravyi2025quantum} would not only lead to better algorithms for stochastic differential equations but may also provide a better understanding of the limitations and opportunities that quantum algorithms afford for differential equations. 

A further issue in our work arises from the use of Markov inequalities to prove large deviation bounds for the random process.  While the Markov inequality is highly general, it leads to inverse polynomial scaling with the desired failure probability.  Tighter concentration bounds may be achievable by making stronger assumptions about the distribution.

One of the most significant challenges that this work addresses stems from the fact that stochastic differential equations need not have continuous or even bounded generators.  The lessons learned in this context may be valuable in addressing broader issues in quantum simulation and differential equations where one might encounter differential operators from broader continuity classes(in time) or in cases where the operators themselves may be unbounded.  Further, incorporating and extending ideas from the Feynman-Kac formula may also open up new connections between stochastic differential equations, stochastic processes and elliptic differential equations in a quantum context.  More broadly, we believe that by stretching quantum algorithms to tackle stochastic differential equations, we not only extend the range of applications of quantum algorithms but also provide fundamentally new ways to expand our toolbox of quantum algorithmic design primitives.

\section*{Acknowledgments}
XYL and NW acknowledge the support from DOE, Office of Science, National Quantum Information Science Research Centers, Co-design Center for Quantum Advantage (C2QA) under Contract No.~DE-SC0012704 (Basic Energy Sciences, PNNL FWP 76274) and Pacific Northwest National Laboratory's Quantum Algorithms and Architecture for Domain Science (QuAADS) Laboratory Directed Research and Development (LDRD) Initiative. 
JPL acknowledges support from Quantum Science and Technology--National Science and Technology Major Project under Grant No.~2024ZD0300500, start-up funding from Tsinghua University and Beijing Institute of Mathematical Sciences and Applications.

\clearpage

\appendix

\section{Probability inequalities}\label{app:prob-ineq}

\begin{theorem}[Borell-Tsirelson-Ibragimov-Sudakov (Borell-TIS) inequality]\label{prelim:BorellTIS}
Let $\{X_t : t\in T\}$ be a centered, separable Gaussian process indexed by a metric space $(T,d)$, and assume
\begin{equation}\label{eq:var-radius}
\sigma^2 \;:=\; \sup_{t\in T} \mathrm{Var}(X_t) \;<\; \infty.
\end{equation}
Let $S:=\sup_{t\in T} X_t$ and suppose $S$ is almost surely finite and $m:=\mathbb{E}[S]<\infty$ (this holds, e.g., if $T$ is compact and the process admits a continuous modification; see Dudley’s entropy criterion below). Then the Borell-TIS inequality asserts sub-Gaussian deviations of $S$ about its mean, with variance proxy $\sigma^2$:
\begin{equation}\label{eq:Borell-TIS}
\mathbb{P}\big(S - m \ge r\big) \;\le\; \exp\!\left(-\frac{r^2}{2\,\sigma^2}\right),
\qquad
\mathbb{P}\big(S - m \le -r\big) \;\le\; \exp\!\left(-\frac{r^2}{2\,\sigma^2}\right),
\quad \text{for all } r \ge 0.
\end{equation}
Equivalently \citep[Theorem~2.1.1]{robert_j_adler__jonathan_e_taylor_random_2007},
\begin{equation}\label{eq:Borell-TIS-two-sided}
\mathbb{P}\big(|S - \mathbb{E}[S]| \ge r\big) \;\le\; 2\,\exp\!\left(-\frac{r^2}{2\,\sigma^2}\right),\qquad r\ge 0.
\end{equation}
This form is particularly useful when $X_t$ is a linear functional of a Gaussian field and $T$ is a compact index set. A common specialization is:
\begin{equation}\label{eq:specialization}
X_{(s,u)} := \langle u, Z(s)\rangle,\qquad (s,u)\in T:=[0,t]\times S^{n-1},
\end{equation}
for a centered Gaussian field $Z(s)\in\mathbb{R}^n$; then
\begin{equation}\label{eq:sigma-special}
\sigma^2 \;=\; \sup_{(s,u)\in T} \mathrm{Var}\big(\langle u, Z(s)\rangle\big) \;=\; \sup_{s\in[0,t]}\big\| \mathrm{Cov}(Z(s))\big\|,
\end{equation}
and Borell-TIS applies to $S=\sup_{(s,u)\in T} X_{(s,u)} = \sup_{s\in[0,t]}\|Z(s)\|$.
For the finiteness of $m=\mathbb{E}[S]$ and measurability (separability) of $S$ on compact index sets, one can use the following classical results:
\begin{enumerate}
\item Dudley’s entropy integral criterion: if the canonical metric $d_X(s,t):=\big(\mathbb{E}|X_s-X_t|^2\big)^{1/2}$ has finite entropy integral on $T$ (e.g., $T$ compact with appropriate covering numbers), then the process admits a continuous modification and $S$ is finite \citet[Theorem~1.3.3]{adler_gaussian_2007}.
\item Under separability and bounded variance radius \cref{eq:var-radius}, $m=\mathbb{E}[S]$ is finite on compact $T$ by chaining arguments \citet[Chapter~7]{adler_gaussian_2007}.
\end{enumerate}
\end{theorem}

\begin{theorem}[Hanson-Wright-type inequalities]\label{prelim:HW}
Let $A\in\mathbb{C}^{n\times n}$, and let $X=(X_1,\ldots,X_n)^\top$ be a random vector with independent, mean-zero, sub-Gaussian coordinates. Denote by $\|\cdot\|$ the operator (spectral) norm and by $\|\cdot\|_{\mathrm{F}}$ the Frobenius norm. We write $\|X_i\|_{\psi_2}$ for the sub-Gaussian Orlicz norm of $X_i$ and set $K:=\max_i \|X_i\|_{\psi_2}$. Then, for all $t\ge 0$,
\begin{equation}\label{eq:HW-finite}
\mathbb{P}\Big(\big|X^\top A X - \mathbb{E}(X^\top A X)\big| \;\ge\; t\Big)
\;\le\;
2\,\exp\!\left(
- c\, \min\!\left\{ \frac{t^2}{K^4\, \|A\|_{\mathrm{F}}^2},\; \frac{t}{K^2\, \|A\|}\right\}
\right),
\end{equation}
for a universal constant $c>0$. In particular, for a standard Gaussian vector $X\sim\mathcal{N}(0,I_n)$ one may take $K=1$, and \cref{eq:HW-finite} holds with the same structure \cite[Theorem~1.1]{rudelson_hanson-wright_2013}.

More generally, let $Z$ be an isonormal Gaussian process on a separable Hilbert space $\mathcal{H}$ (i.e., a centered Gaussian random element with covariance the identity on $\mathcal{H}$), and let $A:\mathcal{H}\to\mathcal{H}$ be a self-adjoint trace-class operator. Then the centered quadratic form satisfies, for all $t\ge 0$,
\begin{equation}\label{eq:HW-Hilbert}
\mathbb{P}\Big(\big|\langle Z, A Z\rangle - \mathbb{E}\langle Z, A Z\rangle\big| \;\ge\; t\Big)
\;\le\;
2\,\exp\!\left(
- c\, \min\!\left\{ \frac{t^2}{\|A\|_{\mathrm{HS}}^2},\; \frac{t}{\|A\|}\right\}
\right),
\end{equation}
where $\|\cdot\|_{\mathrm{HS}}$ denotes the Hilbert-Schmidt norm and $\|\cdot\|$ the operator norm on $\mathcal{H}$.

The inequalities \cref{eq:HW-finite}-\cref{eq:HW-Hilbert} are commonly referred to as Hanson-Wright-type inequalities and provide sub-exponential (Bernstein-type) tail bounds for quadratic forms of sub-Gaussian (resp., Gaussian) random vectors/processes. They strengthen classical large deviation results for $\chi^2$-type functionals and are widely used to control fluctuations of quadratic functionals of Gaussian processes and random vectors.
\end{theorem}

\section{Proof of \cref{lemma:carleman-delta-f}}
\label{sec:proof-carleman-delta-f}

This section provides a deterministic Carleman linearization error bound by employing a
Lyapunov-induced norm to handle non-normal linear parts
\cite{jennings_quantum_2025}. We start by introducing stability
criterion for the Lyapunov systems and providing Carleman linearization error for stable Lyapunov
systems as in \citet{jennings_quantum_2025}, which is then applied to the proof of \cref{lemma:carleman-delta-f}.

\begin{definition}[Stable Lyapunov systems]\label[definition]{def:stable-ODE}
Consider the quadratic $n$-dimensional ODE
\begin{equation}\label{eq:ODE-main}
\dot{\mathbf{x}}(t) \;=\; \mathbf{F}_0 \;+\; \mathbf{F}_1\, \mathbf{x}(t) \;+\; \mathbf{F}_2\, \mathbf{x}(t)^{\otimes 2},\qquad \mathbf{x}(0)=\mathbf{x}_{\mathrm{in}}\in\mathbb{C}^n,
\end{equation}
where $\mathbf{F}_0\in\mathbb{C}^n$, $\mathbf{F}_1\in\mathbb{C}^{n\times n}$, and $\mathbf{F}_2\in\mathbb{C}^{n\times n^2}$ is linear in $\mathbf{x}^{\otimes 2}$. For a positive-definite $P\in\mathbb{C}^{n\times n}$, define $\langle u,v\rangle_P:=u^\dag P v$, $\|u\|_P:=\sqrt{u^\dag P u}$ and the generalized logarithmic norm
\begin{align}
\label{eq:muP}
  \mu_P(\mathbf{F}_1) : & = \max_{\mathbf{x} \neq 0 } \mathrm{Re} \frac{\langle \mathbf{F}_1 \mathbf{x}, \mathbf{x} \rangle_P}{\langle \mathbf{x}, \mathbf{x} \rangle_P} < 0 \\
	       &=\tfrac12\lambda_{\max}(P^{-1}(\mathbf{F}_1^\dag P+P \mathbf{F}_1)) \\
	       \label{eq:muP-eig}
	       &=\; \lambda_{\max}\!\left(\tfrac12\big(P^{1/2}F_1P^{-1/2} + (P^{1/2}F_1P^{-1/2})^\dagger\big)\right)\;<\; 0,
\end{align}
Assume there exists $P>0$ (Hermitian, positive-definite) such that
\begin{equation}\label{eq:Lyapunov-ineq}
P \mathbf{F}_1 + \mathbf{F}_1^\dag P \;<\; 0,\qquad \mu_P(\mathbf{F}_1)<0.
\end{equation}
Then, \cref{eq:ODE-main} is stable with the following Lyapunov $R$-number satisfying
\begin{equation}\label{eq:R-P}
R_P \;:=\; \frac{1}{-\mu_P(\mathbf{F}_1)}\left(\|\mathbf{F}_2\|_P\,\|x(0)\|_P \;+\; \frac{\|\mathbf{F}_0\|_P}{\|x(0)\|_P}\right) < 1,
\end{equation}
where the operator norm induced by $P$ is defined as
\begin{align}
    \| \mathbf{F}_2\|_P = \left\| P^{1/2} \mathbf{F}_2\left( P^{-1/2} \otimes P^{-1/2}\right)\right\|.
\end{align}
\end{definition}

\begin{lemma}[Lemma 3.3 of \cite{jennings_quantum_2025}]
Given $P>0$, let $\mu_P(\mathbf{F}_1)$ be defined as in \cref{eq:muP} and $R_P$ as in \cref{eq:R-P}. Then
\begin{align}
\label{eq:xt-norm}
  \mu_P(\mathbf{F}_1)<0,  R_P<1 \quad  \Rightarrow \quad \|x(t)\|_P \le \|x(0)\|_P.
\end{align}
\end{lemma}
\begin{proof}
Under \cref{eq:Lyapunov-ineq}, consider the $P$-energy $V(x):=\tfrac12\|x\|_P^2$. Along \cref{eq:ODE-main},
\begin{equation}\label{eq:dV}
\dot{V}(t) \;=\; \Re\,\langle x, \mathbf{F}_1 x\rangle_P + \Re\,\langle x, \mathbf{F}_2 x^{\otimes 2}\rangle_P + \Re\,\langle x, \mathbf{F}_0\rangle_P
\;\le\; \mu_P(\mathbf{F}_1)\|x\|_P^2 + \|\mathbf{F}_2\|_P\,\|x\|_P^3 + \|\mathbf{F}_0\|_P\,\|x\|_P.
\end{equation}
Fix the radius $r:=\|x(0)\|_P$ and the sphere $\mathcal{S}:=\{\|x\|_P=r\}$. At points on $\mathcal{S}$, \cref{eq:dV} yields
\[
\dot{V}\le (-\mu_P(\mathbf{F}_1))\, r^2\,(R_P-1),
\]
with $R_P$ as in \cref{eq:R-P}. If $R_P<1$, any $\gamma\in(0,1/\|x(0)\|_P)$ and the continuity of the flow imply $\|x(t)\|_P\le \|x(0)\|_P$ for all $t\in[0,T]$ (no exit from $\mathcal{S}$). This closes the trajectory norm bound needed in the Carleman lifting.
\end{proof}

\begin{lemma}[Carleman linearization error for stable Lyapunov systems]\label[lemma]{lem:Carleman}
  Consider the stable quadratic $n$-dimensional ODE defined in \cref{def:stable-ODE}. Let $\hat{\mathbf{y}}^{(N)}(t)=[\hat{\mathbf{y}}_1(t),\ldots,\hat{\mathbf{y}}_N(t)]^\top\in\mathbb{C}^{n+\cdots+n^N}$ denote the order-$N$ truncated Carleman vector associated with \cref{eq:ODE-main}, with blocks $\hat{\mathbf{y}}_j(t)\in\mathbb{C}^{n^j}$ ($j=1,\ldots,N$), and define the truncation error components
$
\boldsymbol{\eta}_j(t) \;:=\; \mathbf{x}(t)^{\otimes j} - \hat{\mathbf{y}}_j(t),\qquad j=1,\ldots,N,
$
and the Carleman linearization error vector
$
\boldsymbol{\eta}^{(N)}(t):=[\boldsymbol{\eta}_1(t),\ldots,\boldsymbol{\eta}_N(t)]^\top.
$
If $R_P<1$, then there exists a rescaling $x\mapsto \gamma x$ with $0<\gamma<1/\|\mathbf{x}(0)\|_P$ such that, for all $t\in[0,T]$ and all $j=1,\ldots,N$,
\begin{equation}\label{eq:stable-simplified}
\|\boldsymbol{\eta}_j(t)\| \;\le\; \frac{2}{-\xi_P}\; N\, \|\mathbf{F}_2\|_P\, \|P^{-1}\|^{j/2}\, \|\mathbf{x}(0)\|_P^{\,N+1},
\end{equation}
where the Gershgorin bound $\xi_P$ is given by \cref{eq:xiP}.
\end{lemma}

\begin{proof}
By the Carleman linearization error dynamics \cref{eq:eta-evolution}, specialized to the deterministic system \cref{eq:ODE-main}, the error vector $\boldsymbol{\eta}^{(N)}(t)$ evolves under the truncated Carleman generator $A$ with residual $\mathbf{R}_N(t)$ from \cref{lemma:RNt}, whose last Carleman block is $A_{N+1}^N x(t)^{\otimes (N+1)}$. Let $P_N := \bigoplus_{j=1}^N P^{\otimes j}$ and equip $\boldsymbol{\eta}^{(N)}(t)$ with the $P_N$-norm $\|\cdot\|_{P_N}$. Then, exactly as in \cref{eq:d-norm-squared},
\[
\frac{d}{dt}\|\boldsymbol{\eta}^{(N)}(t)\|_{P_N}^2 = {\boldsymbol{\eta}^{(N)}}^\dag(\mathbf{A}^\dag P_N + P_N \mathbf{A})\boldsymbol{\eta}^{(N)} + 2\,\Re\,\langle \boldsymbol{\eta}^{(N)},\mathbf{R}_N(t)\rangle_{P_N}.
\]
Introduce the $N\times N$ symmetric “matrix of norms’’ $\widehat{G}:=(G+G^\dag)/2$ with diagonal $G_{j,j}=2 j\mu_P(\mathbf{F}_1)$, sub-diagonal $G_{j,j-1}=2 j\|\mathbf{F}_0\|_P$ and super-diagonal $G_{j,j+1}=2 j\|\mathbf{F}_2\|_P$. Then
\begin{align*}
&{\boldsymbol{\eta}^{(N)}}^\dag(\mathbf{A}^\dag P_N+P_N\mathbf{A})\boldsymbol{\eta}^{(N)} \;\le\; \lambda_1(\widehat{G})\, \|\boldsymbol{\eta}^{(N)}\|_{P_N}^2,\qquad\\
& 2\,\Re\,\langle \boldsymbol{\eta}^{(N)},\mathbf{R}_N(t)\rangle_{P_N} \;\le\; 2\,\|\mathbf{A}_{N+1}^N\|_P\, \|\mathbf{x}\|_P^{N+1}\, \|\boldsymbol{\eta}^{(N)}\|_{P_N} \;\le\; 2 N\,\|\mathbf{F}_2\|_P\, \|\mathbf{x}(0)\|_P^{N+1}\,\|\boldsymbol{\eta}^{(N)}\|_{P_N},
\end{align*}
where \cref{eq:xt-norm} and the bound $\|A_{N+1}^N\|_P \le N\,\|\mathbf{F}_2\|_P$ (cf. \cref{lemma:RNt}) are applied.
Hence,
\[
\frac{d}{dt}\|\boldsymbol{\eta}^{(N)}(t)\|_{P_N} \;\le\; \tfrac12\lambda_1(\widehat{G})\, \|\boldsymbol{\eta}^{(N)}(t)\|_{P_N} + N\,\|\mathbf{F}_2\|_P\, \|\mathbf{x}(0)\|_P^{N+1},
\]
and Grönwall yields
\[
\|\boldsymbol{\eta}^{(N)}(t)\|_{P_N} \;\le\; \frac{1-e^{\lambda_1(\widehat{G})t/2}}{-\lambda_1(\widehat{G})/2}\; N\,\|\mathbf{F}_2\|_P\, \|\mathbf{x}(0)\|_P^{N+1}.
\]

Projecting to the $j$-th block and using $\|u\| \le \|P^{-1}\|^{j/2}\,\|u\|_{P^{\otimes j}}$ gives 
\begin{equation}\label{eq:stable-bound}
\|\boldsymbol{\eta}_j(t)\| \;\le\; \frac{1 - e^{\lambda_1(\widehat{G})\,t/2}}{-\lambda_1(\widehat{G})/2}\; N\, \|\mathbf{F}_2\|_P\, \|P^{-1}\|^{j/2}\, \|\mathbf{x}(0)\|_P^{\,N+1},
\end{equation}

The largest eigenvalue of $\widehat{G}$ satisfies the Gershgorin bound
\begin{equation}\label{eq:xiP}
\lambda_1(\widehat{G}) \;\le\; \xi_P \;:=\; 4\,\mu_P(\mathbf{F}_1) \;+\; 5\,\|\mathbf{F}_0\|_P \;+\; 3\,\|\mathbf{F}_2\|_P \;<\; 0.
\end{equation}
Since $\frac{1 - e^{\lambda_1(\widehat{G})t/2}}{-\lambda_1(\widehat{G})/2}\le 2/(-\lambda_1(\widehat{G}))\le 2/(-\xi_P)$,
combining \cref{eq:xiP} and \cref{eq:stable-bound} yields \cref{eq:stable-simplified} and completes the proof. 
\end{proof}

\begin{lemma}[Carleman linearization error for $\delta \mathbf{x}$ in the Lyapunov-stable condition]\label[lemma]{lem:LyapunovStableNoMonotone}
  Under the same settings and assumptions of \cref{def:stable-ODE}, assume $\mathbf{F}_0 = \mathbf{0}$ in \cref{eq:ODE-main}. Define the truncation error components $\boldsymbol{\eta}_j(t):=(\delta \mathbf{x}(t))^{\otimes j}-\hat{\mathbf{y}}_j(t)$ and the error vector $\boldsymbol{\eta}^{(N)}(t):=(\boldsymbol{\eta}_1(t),\ldots,\boldsymbol{\eta}_N(t))^\top$. Then for all $t\in[0,T]$,
\begin{equation}\label{eq:monotone}
\|\delta \mathbf{x}(t)\|_P \;\le\; \|\delta \mathbf{x}_{\mathrm{in}}\|_P,
\end{equation}
and the Carleman linearization errors satisfy
\begin{equation}\label{eq:etaN-bound}
\|\boldsymbol{\eta}_N(t)\|_P \;\le\; \|\delta \mathbf{x}_{\mathrm{in}}\|_P^{\,N+1}\, \frac{\|\mathbf{F}_2\|_P}{|\mu_P(\mathbf{F}_1)|}\, \big(1-e^{N\,\mu_P(\mathbf{F}_1)\, t}\big),
\end{equation}
and, for each $j=1,2,\ldots,N$,
\[
\|\boldsymbol{\eta}_j(t)\|_P \;\le\; \|\delta \mathbf{x}_{\mathrm{in}}\|_P^{\,N+1}\, \frac{\|\mathbf{F}_2\|_P^{\,N+1-j}}{|\mu_P(\mathbf{F}_1)|^{\,N+1-j}}.
\]
In particular, for $j=1$, the nested integral can be evaluated exactly to give
\[
\|\boldsymbol{\eta}_1(t)\|_P \;\le\; \|\delta \mathbf{x}_{\mathrm{in}}\|_P\, \left(\frac{\|\mathbf{F}_2\|_P}{|\mu_P(\mathbf{F}_1)|}\right)^N\, \big(1-e^{\mu_P(\mathbf{F}_1)\, t}\big)^{N}.
\]
\end{lemma}

\begin{proof}
The Dini derivative of the $P$-norm along \cref{eq:perturbation_dynamics} satisfies
\begin{equation}\label{eq:dini}
\frac{d}{dt}\,\|\delta \mathbf{x}(t)\|_P \;\le\; \mu_P(\mathbf{F}_1)\,\|\delta \mathbf{x}(t)\|_P \;+\; \|\mathbf{F}_2\|_P\, \|\delta \mathbf{x}(t)\|_P^2,
\end{equation}
since $\|\mathbf{F}_2(\delta \mathbf{x}\otimes \delta \mathbf{x})\|_P \le \|\mathbf{F}_2\|_P\, \|\delta \mathbf{x}\|_P^2$. Set $z(t):=\|\delta \mathbf{x}(t)\|_P$, $\mu:=\mu_P(\mathbf{F}_1)<0$, $c:=\|\mathbf{F}_2\|_P$, so \cref{eq:dini} is
\begin{equation}\label{eq:logistic}
\dot{z}(t) \;\le\; \mu\, z(t) \;+\; c\, z(t)^2 \;=\; z(t)\big(\mu + c\, z(t)\big).
\end{equation}
The exact logistic equation $\dot{y}=\mu y + c y^2$ with $y(0)=z(0)$ admits explicit solution
\begin{equation}\label{eq:logistic-sol}
y(t) \;=\; \frac{z(0)\, e^{\mu t}}{1 - \frac{c}{\mu}\, z(0)\, (1-e^{\mu t})}
\end{equation}
for $\mu\neq 0$. Combining \cref{eq:logistic} with comparison principles (e.g., Bihari-LaSalle inequality), one has $z(t)\le y(t)$ for all $t\ge 0$. The small-gain condition \cref{eq:R-P} is equivalent to $z(0)<(-\mu)/c$, hence $1 - \frac{c}{\mu} z(0) (1-e^{\mu t})>0$ and $\mu + c\, z(0)<0$. It follows from \cref{eq:logistic-sol} that $y(t)\le z(0)$ for all $t\ge 0$, hence $z(t)\le z(0)$ and \cref{eq:monotone} holds:
\begin{equation*}
\|\delta \mathbf{x}(t)\|_P \;\le\; \|\delta \mathbf{x}_{\mathrm{in}}\|_P \qquad \text{for all } t\in[0,T].
\end{equation*}

For the Carleman bounds, use the ordered lift. The degree-preserving blocks $\mathbf{A}_j^j$ commute, so by \cite[Proposition~4.2.8]{horn_topics_1991}
\begin{equation}\label{eq:exp-factor}
e^{\mathbf{A}_j^j t} \;=\; \big(e^{\mathbf{F}_1 t}\big)^{\otimes j}.
\end{equation}
Combining \cref{eq:muP} with \cref{eq:exp-factor} yields
\begin{equation}\label{eq:Ajj-exp}
\|e^{\mathbf{A}_j^j t}\|_P \;\le\; \big(\|e^{\mathbf{F}_1 t}\|_P\big)^j \;\le\; e^{j\, \mu_P(\mathbf{F}_1)\, t} \qquad \forall t\ge 0.
\end{equation}
For the degree-raising blocks, the ordered monomial basis gives (each of the $j$ slots can host the action of $\mathbf{F}_2$)
\begin{equation}\label{eq:raising}
\|\mathbf{A}_{j+1}^j\|_P \;\le\; j\, \|\mathbf{F}_2\|_P.
\end{equation}
The top-block error obeys the variation-of-constants formula with $\boldsymbol{\eta}_N(0)=0$:
\begin{equation}\label{eq:etaN-VoC}
\boldsymbol{\eta}_N(t) \;=\; \int_0^t e^{\mathbf{A}_N^N (t-s)}\, \mathbf{A}_{N+1}^N\, \delta \mathbf{x}(s)^{\otimes (N+1)}\, ds.
\end{equation}
Combining \cref{eq:Ajj-exp} with \cref{eq:raising} and \cref{eq:monotone} yields
\begin{align}
\|\boldsymbol{\eta}_N(t)\|_P
&\le \int_0^t e^{N\, \mu_P(\mathbf{F}_1)\,(t-s)}\, \|\mathbf{A}_{N+1}^N\|_P\, \|\delta \mathbf{x}(s)\|_P^{N+1}\, ds \nonumber\\
&\le N\, \|\mathbf{F}_2\|_P\, \|\delta \mathbf{x}_{\mathrm{in}}\|_P^{N+1}\, \int_0^t e^{N\, \mu_P(\mathbf{F}_1)\,(t-s)}\, ds \nonumber\\
&= \|\delta \mathbf{x}_{\mathrm{in}}\|_P^{\,N+1}\, \frac{\|\mathbf{F}_2\|_P}{|\mu_P(\mathbf{F}_1)|}\, \big(1-e^{N\,\mu_P(\mathbf{F}_1)\, t}\big),
\label{eq:etaN-bound-deriv}
\end{align}
which is \cref{eq:etaN-bound}.

For $j<N$, iterate variation of constants backwards:
\begin{equation}\label{eq:etaj-VoC}
\boldsymbol{\eta}_j(t) \;=\; \int_0^t e^{\mathbf{A}_j^j(t-s_{N-j})}\, \mathbf{A}_{j+1}^j\, \boldsymbol{\eta}_{j+1}(s_{N-j})\, ds_{N-j},
\end{equation}
and substitute \cref{eq:etaN-bound-deriv} recursively. The nested structure comprises integrals of the form
\[
\int_0^{s_{k+1}} e^{(j+k)\,\mu_P(\mathbf{F}_1)\,(s_{k+1}-s_k)}\, ds_k,\qquad k=0,1,\ldots,N-j,
\]
with $(j+k)\,\mu_P(\mathbf{F}_1)<0$. Using the elementary bound, for $a<0$,
\begin{equation}\label{eq:inner-int}
\int_0^{s_{k+1}} e^{a\, (s_{k+1}-s_k)}\, ds_k \;=\; \frac{1-e^{a\, s_{k+1}}}{|a|} \;\le\; \frac{1}{|a|},
\end{equation}
and combining \cref{eq:inner-int} across the $N+1-j$ nested integrals in \cref{eq:etaj-VoC} yields
\[
\int_{0}^{t} \cdots \int_{0}^{s_1}(\cdots)\, ds_0 \cdots ds_{N-j} \;\le\; \frac{(j-1)!}{N!\, |\mu_P(\mathbf{F}_1)|^{\,N+1-j}}.
\]
Multiplying the $\mathbf{A}_{j+1}^j$ bounds \cref{eq:raising} across the chain gives the prefactor $N!/(j-1)!\,\|\mathbf{F}_2\|_P^{N+1-j}$, and hence
\begin{equation*}
\|\boldsymbol{\eta}_j(t)\|_P \;\le\; \|\delta \mathbf{x}_{\mathrm{in}}\|_P^{\,N+1}\, \frac{\|\mathbf{F}_2\|_P^{\,N+1-j}}{|\mu_P(\mathbf{F}_1)|^{\,N+1-j}},
\end{equation*}
which is \cref{eq:etaj-bound}. For $j=1$, the nested integral in \cref{eq:etaj-VoC} may be evaluated exactly via
\begin{equation}\label{eq:Forets}
\int_0^{s_{N}} \cdots \int_0^{s_2}\int_0^{s_1} e^{a(-Ns_0+\sum_{k=1}^N s_k)}\,ds_0\,ds_1\cdots ds_{N-1}
\;=\; \frac{\big(e^{a s_N}-1\big)^N}{N!\, a^N}.
\end{equation}
Applying \cref{eq:Forets} with $a:=\mu_P(\mathbf{F}_1)$ and $s_N:=t$ in the $j=1$ case of \cref{eq:etaj-VoC}, together with the prefactor $\|\delta \mathbf{x}_{\mathrm{in}}\|_P\, \|\mathbf{F}_2\|_P^N$, gives \cref{eq:eta1-exact}.
\end{proof}

\section{Statistical properties of the multivariate OU process and the resulting Carleman matrix}\label{app:ou-stats}

\subsection{Statistical properties of the multivariate OU process}

Below, we derive explicit bounds for the supremum variance and expectation of the multivariate OU process defined in \cref{eq:general_ou}, which are then used to establish sub-Gaussian tail bounds via the Borell-TIS/Gaussian concentration inequality in \cref{lem:BTIS-OU}.
\begin{lemma}\label[lemma]{lemma:sup-OU}
  Under the definition and assumptions in \cref{def:OU-LPDE}, the supremum variance and expectation of the multivariate OU process satisfy
\begin{align}\label{eq:sigma-star-bound}
&  \sigma_*^2\  :=\ \sup_{t}\lambda_{\max}\big(\mathrm{Cov}(\mathbf{F}_0(t))\big)\ \le\ \frac{1-e^{-2\lambda_{\min} T}}{2\,\lambda_{\min}}\ \|\boldsymbol{\Sigma}\|^{2}, \\
\label{eq:m-compact}
&  \mathbb{E}\Big[\sup_{t\in[0,T]}\|\mathbf{F}_0(t)\|\Big]  \le\ C''\ \|\boldsymbol{\Sigma}\|\ \sqrt{\frac{1-e^{-2\lambda_{\min} T}}{2\,\lambda_{\min}}}\ \Big(\ \sqrt{n}\ +\ \sqrt{1+\lambda_{\min}T}\ \Big).
\end{align}
\end{lemma}

\begin{proof}
We first provide a bound for the supremum variance of the OU process.
For matrix–norm bounds, it is convenient to assume that the symmetric part
$
\boldsymbol{S}\ :=\ \tfrac{1}{2}\big(\boldsymbol{\Theta}+\boldsymbol{\Theta}^\top\big)
$
obeys \(\lambda_{\min}(\boldsymbol{S})>0\), and we write \(\lambda_{\min}:=\lambda_{\min}(\boldsymbol{S})\). Then the standard logarithmic–norm estimate gives
\begin{equation}\label{eq:exp-bound}
\|e^{-\boldsymbol{\Theta} t}\|\ \le\ e^{-\lambda_{\min} t}\qquad (t\ge 0).
\end{equation}
Define the centered Gaussian process
\begin{equation}\label{eq:Ztv}
Z(t,v)\ :=\ v^\top \mathbf{F}_0(t),\qquad (t,v)\in\mathcal{T}:=[0,T]\times \mathbb{S}^{n-1},
\end{equation}
where 
\(
\mathbb{S}^{n-1}\ :=\ \{\,v\in\mathbb{R}^n:\ \|v\|=1\,\}
\)
denotes the unit Euclidean sphere in \(\mathbb{R}^n\).
Then
\[
\sigma_*^2\ :=\ \sup_{t\in[0,T]}\ \sup_{v\in\mathbb{S}^{n-1}}\ \mathrm{Var}\big[Z(t,v)\big]
\ =\ \sup_{t\in[0,T]}\ \lambda_{\max}\big(\mathrm{Cov}(\mathbf{F}_0(t))\big),
\]
because for any positive semidefinite matrix \(C\), \(\sup_{\|v\|=1} v^\top C v=\lambda_{\max}(C)\).
Assuming \(\mathbf{F}_0(0)=0\), the explicit solution is
\[
\mathbf{F}_0(t)\ =\ \int_0^t e^{-\boldsymbol{\Theta}(t-s)}\,\boldsymbol{\Sigma}\,d\mathbf{W}(s),
\]
hence
\[
\mathrm{Cov}(\mathbf{F}_0(t))\ =\ \int_0^t e^{-\boldsymbol{\Theta}(t-s)}\,\boldsymbol{\Sigma}\,\boldsymbol{\Sigma}^\top\,e^{-\boldsymbol{\Theta}^\top(t-s)}\,ds.
\]
Using the operator-norm inequality \(\lambda_{\max}(A X A^\top) \le \|A\|^2\,\lambda_{\max}(X)\) for \(X\succeq 0\) and \(\|AB\|\le \|A\|\|B\|\), we obtain, together with \cref{eq:exp-bound},
\begin{align*}
\lambda_{\max}\big(\mathrm{Cov}(\mathbf{F}_0(t))\big)
&\le \int_0^t \big\|e^{-\boldsymbol{\Theta}(t-s)}\big\|^2\ \big\|\boldsymbol{\Sigma}\boldsymbol{\Sigma}^\top\big\|\ ds\\
&\le \|\boldsymbol{\Sigma}\|^2 \int_0^t e^{-2\lambda_{\min}(t-s)}\,ds
\ =\ \|\boldsymbol{\Sigma}\|^2\ \frac{1-e^{-2\lambda_{\min} t}}{2\,\lambda_{\min}}\\
&\le \|\boldsymbol{\Sigma}\|^2\ \frac{1-e^{-2\lambda_{\min} T}}{2\,\lambda_{\min}}.
\end{align*}
Taking the supremum over \(t\in[0,T]\) yields
\cref{eq:sigma-star-bound} as claimed.

We then provide a bound for \(m=\mathbb{E}\big[\sup_{t\in[0,T]}\|\mathbf{F}_0(t)\|\big]\) by applying Dudley’s entropy integral to the Gaussian process $Z(t,v)$ in \cref{eq:Ztv}
with canonical (pseudo) metric \citep[Definition~1.3.2, Page 13]{robert_j_adler__jonathan_e_taylor_random_2007}
\[
d\big((t,v),(s,w)\big)^2\ :=\ \mathbb{E}\big[\,\big(Z(t,v)-Z(s,w)\big)^2\,\big]\ =\ \mathrm{Var}\big(Z(t,v)-Z(s,w)\big).
\]
We first upper bound \(d\). For any \(v,w\in\mathbb{S}^{n-1}\) and \(t,s\in[0,T]\), by the triangle inequality and the variance bound for sums,
\[
d\big((t,v),(s,w)\big)
\ \le\ \sqrt{\mathrm{Var}\big(Z(t,v)-Z(t,w)\big)}\ +\ \sqrt{\mathrm{Var}\big(Z(t,w)-Z(s,w)\big)}.
\]
The first term is bounded by
\[
\sqrt{\mathrm{Var}\big(Z(t,v)-Z(t,w)\big)}\ \le\ \sqrt{\mathrm{Var}(Z(t,v))}+\sqrt{\mathrm{Var}(Z(t,w))}
\ \le\ 2\,\sigma_*\,\tfrac{1}{2}\|v-w\|
\ \le\ \sigma_*\,\|v-w\|,
\]
where we also used the inequality \(|v^\top X - w^\top X|\le \|v-w\|\,\|X\|\) and \(\sup_t \sqrt{\mathrm{Var}(Z(t,\cdot))}\le \sigma_*\). For the second term, consider one fixed \(w\), and write
\[
Z(t,w)-Z(s,w)\ =\ w^\top\!\int_s^t e^{-\boldsymbol{\Theta}(t-u)}\boldsymbol{\Sigma}\,d\mathbf{W}(u)\ +\ w^\top\!\Big(\big(e^{-\boldsymbol{\Theta}(t-s)}-\mathbf{I}\big)\mathbf{F}_0(s)\Big).
\]
Taking variances and using independence of the increment and the past, plus \(\|e^{-\boldsymbol{\Theta} r}-\mathbf{I}\|\le \int_0^r \|\,\boldsymbol{\Theta}\,e^{-\boldsymbol{\Theta} u}\|\,du \le \|\boldsymbol{\Theta}\|\,r\), we obtain (by Itô isometry and \cref{eq:exp-bound})
\begin{align*}
\mathrm{Var}\big(Z(t,w)-Z(s,w)\big)
&\le \int_s^t \|e^{-\boldsymbol{\Theta}(t-u)}\boldsymbol{\Sigma}\|^2\,du\ +\ \|e^{-\boldsymbol{\Theta}(t-s)}-\mathbf{I}\|^2\ \mathrm{Var}\big(Z(s,w)\big)\\
&\le \|\boldsymbol{\Sigma}\|^2\int_0^{|t-s|} e^{-2\lambda_{\min} u} du\ +\ \|\boldsymbol{\Theta}\|^2\,|t-s|^2\ \sigma_*^2\\
&\le \Big(\tfrac{\|\boldsymbol{\Sigma}\|^2}{2\lambda_{\min}}\ +\ \|\boldsymbol{\Theta}\|^2\,T\,\sigma_*^2\Big)\,|t-s|.
\end{align*}
Hence there exists a constant
\[
L_t^2\ :=\ \frac{\|\boldsymbol{\Sigma}\|^2}{2\lambda_{\min}}\ +\ \|\boldsymbol{\Theta}\|^2\,T\,\sigma_*^2
\]
such that
\[
\sqrt{\mathrm{Var}\big(Z(t,w)-Z(s,w)\big)}\ \le\ L_t\,\sqrt{|t-s|}.
\]
Combining the two terms, we have the canonical metric bound
\begin{equation}\label{eq:canonical-metric}
d\big((t,v),(s,w)\big)\ \le\ \sigma_*\,\|v-w\|\ +\ L_t\,\sqrt{|t-s|}.
\end{equation}
We now estimate the covering numbers of \((\mathcal{T},d)\). For the unit sphere,
\[
\mathcal{N}\big(\mathbb{S}^{n-1},\|\cdot\|,\eta\big)\ \le\ \Big(1+\frac{2}{\eta}\Big)^{\!n}\qquad (0<\eta\le 2).
\]
For the time interval with the \(|\cdot|^{1/2}\)-metric, a uniform \(\eps\)-net has size
\[
\mathcal{N}\big([0,T],\,|\cdot|^{1/2},\,\eta\big)\ \le\ 1+\frac{T}{\eta^2}.
\]
Using the product structure implied by \cref{eq:canonical-metric}, a standard covering argument yields
\[
\mathcal{N}\big(\mathcal{T},d,\eps\big)\ \le\ \mathcal{N}\!\Big(\mathbb{S}^{n-1},\|\cdot\|,\,\frac{\eps}{2\sigma_*}\Big)\ \cdot\ \mathcal{N}\!\Big([0,T],|\cdot|^{1/2},\,\frac{\eps}{2L_t}\Big)
\ \le\ \Big(1+\frac{4\,\sigma_*}{\eps}\Big)^{\!n}\ \cdot\ \Big(1+\frac{4\,L_t^2\,T}{\eps^2}\Big).
\]
By Dudley’s entropy integral \citep[Theorem~1.3.3, Page 14]{robert_j_adler__jonathan_e_taylor_random_2007}, there exists a universal constant \(C>0\) such that
\begin{equation}\label{eq:E-Dudley}
\mathbb{E}\Big[\sup_{(t,v)\in\mathcal{T}} Z(t,v)\Big]
\ \le\ C\ \int_0^{\mathrm{diam}(\mathcal{T},d)} \sqrt{\log \mathcal{N}(\mathcal{T},d,\eps)}\,d\eps.
\end{equation}
Bounding the diameter by \(\mathrm{diam}(\mathcal{T},d)\le 2\sigma_*+2\sqrt{L_t^2 T}\) and evaluating the integral with the above covering estimate gives the explicit (albeit conservative) bound
\begin{equation}\label{eq:m-explicit}
\mathbb{E}\Big[\sup_{t\in[0,T]}\|\mathbf{F}_0(t)\|\Big]
\ =\ \mathbb{E}\Big[\sup_{(t,v)\in\mathcal{T}} Z(t,v)\Big]
\ \le\ C'\,\Big(\ \sigma_*\,\sqrt{n}\ +\ \sqrt{L_t^2\,T}\ \sqrt{\log\big(2+4L_t^2 T\big)}\ \Big),
\end{equation}
for some universal constant \(C'>0\). Using \cref{eq:sigma-star-bound} and the definition of \(L_t\), we further simplify \cref{eq:m-explicit} into \cref{eq:m-compact},
where \(C''>0\) is a universal constant. The precise constants \(C',C''\) can be traced through the Dudley integral; the important point is that \(m\) is finite and admits an explicit upper bound in terms of \(\lambda_{\min},\ \|\boldsymbol{\Sigma}\|,\ n,\ T\).
\end{proof}
These bounds, together with Borell-TIS in \cref{lem:BTIS-OU}, yield the sub-Gaussian tail control used to produce a high–probability threshold for \(\Lambda(\omega)=\sup_{t\in[0,T]}\|L(t,\omega)\|\) when \(L\) depends on $\mathbf{F}_0$ via Carleman transfer matrix.

\begin{remark}[Universal constants \(C\) and \(C''\)]
The constant \(C\) in the Dudley entropy integral \cref{eq:E-Dudley}
is an absolute numerical constant that arises from the proof of Dudley’s inequality. It does not depend on the dimension, the covariance of the Gaussian process, the index set, or any parameters of the OU process. Different references give slightly different numerical values (e.g., \(C=12\) or \(C=24\)) depending on normalization \citep[Theorem~1.3.3]{robert_j_adler__jonathan_e_taylor_random_2007}.
\(C\) is universal in the sense that it is a fixed finite number independent of the problem data.
The constant \(C''\) in \cref{eq:m-compact} is also universal because it is a finite absolute constant independent of the OU process parameters and the dimension, determined only by the chosen version of Dudley’s inequality.
\end{remark}

With the expectation and supremum variance bounds in \cref{lemma:sup-OU} in hand, we can now state the Borell-TIS/Gaussian concentration result for the supremum of the multivariate OU process that will be used to bound the solution norm of the OU-driven quadratic ODE in \cref{cor:borell_xP} and ultimately, the Carleman linearization error in \cref{coro:P-eta}.

\begin{lemma}[Borell-TIS/Gaussian concentration for the OU supremum]
\label[lemma]{lem:BTIS-OU}
Under the assumptions of \cref{lem:lyap-path-solution-norm}, suppose there exist constants $M\ge 1$ and $\mu>0$ such that $\|e^{-\Theta s}\| \le M\, e^{-\mu s}$ for all $s\ge 0$. Then the supremum of the multivariate OU process (\cref{eq:general_ou})
\[
S_{t,P} \;:=\; \sup_{0\le s\le t}\, \|\mathbf{F}_0(s)\|_{P}
\]
is sub-Gaussian about its mean:
\begin{equation}
\label{eq:BTIS-OU-tail}
\mathbb{P}\big(S_{t,P} \ge \mathbb{E}[S_{t,P}] + r\big)
\;\le\;
\exp\!\left(-\frac{r^2}{2\,Q_{*,P}}\right),
\qquad r>0,
\end{equation}
with variance proxy
\begin{equation}
\label{eq:QstarP-def}
Q_{*,P}\ :=\ \sup_{0\le s\le t}\big\|P^{1/2}\,\mathrm{Cov}(\mathbf{F}_0(s))\,P^{1/2}\big\|
\ \le\ \|P^{1/2}\|^2\,\frac{M^2\,\|\Sigma\|^2}{2\mu}\,.
\end{equation}
\end{lemma}

\begin{proof}
The mild solution of the OU process (\cref{eq:general_ou}) is
\[
\mathbf{F}_0(s) \;=\; e^{-\Theta s}\,\mathbf{F}_0(0) \;+\; \int_0^s e^{-\Theta(s-\tau)}\,\Sigma\, dW(\tau)
\;=:\; m(s) \;+\; Z(s),\qquad s\in[0,t],
\]
where $m(s):=e^{-\Theta s}\mathbf{F}_0(0)$ is deterministic and $Z$ is a centered Gaussian process with continuous sample paths and covariance \cite[Section~7.5]{oksendal_stochastic_2003}
\begin{equation}
\label{eq:OU-cov}
\mathrm{Cov}(Z(s))
\;=\;
\int_0^s e^{-\Theta(s-\tau)}\,\Sigma \Sigma^\top\, e^{-\Theta^\top(s-\tau)}\,d\tau .
\end{equation}
We apply the dual-norm representation \citep[Chapter~IV, \S2, Th\'eor\`eme~3, p.~55]{banach1932theorie}; \citep[Chapter~3, Example~3.11, p.~82]{boyd2004convex} to express the norm as a supremum over the dual unit ball:
\[
\|\mathbf{F}_0(s)\|_{P} \;=\; \sup_{\|u\|_{P,*}\le 1}\,\langle u,\, m(s)+Z(s)\rangle.
\]
Therefore,
\[
S_{t,P} \;=\; \sup_{(u,s)\in T}\, Y_{u,s},\qquad
T:=\{(u,s): \|u\|_{P,*}\le 1,\ s\in[0,t]\},\quad
Y_{u,s}:=\langle u, m(s)\rangle + \langle u, Z(s)\rangle.
\]
Through the dual-norm representation, the path supremum $S_{t,P}=\sup_{0\le s\le t}\|\mathbf{F}_0(s)\|_P$ of the non-Gaussian process $s\mapsto\|\mathbf{F}_0(s)\|_P$ is re-expressed as the supremum $\sup_{(u,s)\in T}Y_{u,s}$ of the process $\{Y_{u,s}\}_{(u,s)\in T}$ over the compact index set~$T$. Each $Y_{u,s}=\langle u,\mathbf{F}_0(s)\rangle$ is a Gaussian random variable---being a linear functional of the Gaussian vector~$\mathbf{F}_0(s)$---and every finite sub-collection is jointly Gaussian (as linear functionals of the joint Gaussian vector $(\mathbf{F}_0(s_1),\dots,\mathbf{F}_0(s_k))$), so $\{Y_{u,s}\}$ is a Gaussian process and the Borell--TIS inequality becomes applicable.

The centered part $X_{u,s}:=\langle u, Z(s)\rangle$ is a separable centered Gaussian process indexed by $T$ (finite-dimensional index $u$ and continuous $s\mapsto Z(s)$ suffice). Its pointwise variance satisfies, for any $(u,s)\in T$,
\[
\mathrm{Var}(X_{u,s})
\;=\;
u^\top\, \mathrm{Cov}(Z(s))\, u
\;\le\;
\|P^{-1/2}u\|^2\, \big\|P^{1/2}\,\mathrm{Cov}(Z(s))\,P^{1/2}\big\|
\;\le\;
\big\|P^{1/2}\,\mathrm{Cov}(Z(s))\,P^{1/2}\big\|,
\]
since $\|P^{-1/2}u\|\le 1$ whenever $\|u\|_{P,*}\le 1$. Using \cref{eq:OU-cov} and the semigroup bound $\|e^{-\Theta r}\|\le M e^{-\mu r}$, we obtain uniformly in $s\in[0,t]$,
\[
\big\|P^{1/2}\,\mathrm{Cov}(Z(s))\,P^{1/2}\big\|
\;\le\;
\|P^{1/2}\|^2\, \|\Sigma\|^2
\int_0^s \|e^{-\Theta(s-\tau)}\|^2\,d\tau
\;\le\;
\|P^{1/2}\|^2\, \|\Sigma\|^2\, \frac{M^2}{2\mu}.
\]
Thus the variance proxy $\sigma^2:=\sup_{(u,s)\in T}\mathrm{Var}(X_{u,s})$ satisfies $\sigma^2\le Q_{*,P}$ with $Q_{*,P}$ given by \cref{eq:QstarP-def}. Moreover, $\mathbb{E}[S_{t,P}]<\infty$ by standard entropy bounds for Gaussian suprema on compact index sets (Dudley's integral bound) \cite[Theorem~2.3.1]{adler_gaussian_2007}.

The field $Y_{u,s}=\mu_{u,s}+X_{u,s}$ (with deterministic shift $\mu_{u,s}:=\langle u,m(s)\rangle$) is non-centered, whereas the Borell--TIS inequality is classically stated for centered processes.
Since $\mu_{u,s}$ is deterministic, $S_{t,P}=\sup_{(u,s)\in T}(\mu_{u,s}+X_{u,s})$ is a pointwise supremum of affine functionals of the centered Gaussian field~$X$, hence a convex, Lipschitz function of the underlying Gaussian vector with Lipschitz constant $L=\sigma$: the pointwise supremum of affine forms $\{\mu_\alpha+\langle a_\alpha,\cdot\rangle\}_\alpha$ has Lipschitz constant $\sup_\alpha\|a_\alpha\|$, and here $\|a_{u,s}\|^2=\mathrm{Var}(X_{u,s})$, so $L=\sup_{(u,s)}\sqrt{\mathrm{Var}(X_{u,s})}=\sigma$.
By the Gaussian concentration of measure principle \cite[Theorem~7.1, Corollary~7.12]{ledoux_concentration_2005}, for every $r>0$,
\[
\mathbb{P}\!\left(S_{t,P} \ge \mathbb{E}[S_{t,P}] + r\right) \;\le\; \exp\!\left(-\frac{r^2}{2\,\sigma^2}\right) \;\le\; \exp\!\left(-\frac{r^2}{2\,Q_{*,P}}\right),
\]
which proves \cref{eq:BTIS-OU-tail}.
\end{proof}

\begin{remark}
  \begin{enumerate}
    \item If $\mathbf{F}_0(0)=0$, then $m(s)\equiv 0$ and $S_{t,P}=\sup_{(u,s)\in T}
      X_{u,s}$ is the supremum of a centered separable Gaussian process. In
      this case the Borell--TIS inequality applies directly to $\{X_{u,s}\}$
      with the same variance proxy $Q_{*,P}$ (\cref{eq:QstarP-def}), yielding the same tail bounds (\cref{eq:BTIS-OU-tail}).

    \item If $\Sigma = \sigma I$, \cref{eq:QstarP-def} can be simplified to
\begin{equation}
\label{eq:QstarP-explicit-bound}
Q_{*,P} \;\le\; \sigma^2\, \|P^{1/2}\|^2\, \frac{M^2}{2\mu}.
\end{equation}
If, moreover, $\Theta$ is normal (i.e., $\Theta^\top \Theta=\Theta \Theta^\top$) with $\min_i \mathrm{Re}\,\lambda_i(\Theta)\ge \mu>0$, then $\|e^{-\Theta r}\|=e^{-\mu r}$ and
\begin{equation}
\label{eq:QstarP-normal-tight}
Q_{*,P}
\;\le\;
\sigma^2\, \|P^{1/2}\|^2\, \frac{1-e^{-2\mu t}}{2\mu}
\;\le\;
\sigma^2\, \|P^{1/2}\|^2\, \frac{1}{2\mu}.
\end{equation}
In addition, $\Sigma = \sigma I$ leads to an explicit upper bound for the mean of $S_{t,P}$:
\begin{equation}
\label{eq:mean-StP-upper}
\mathbb{E}[S_{t,P}]
\;\le\;
\sup_{0\le s\le t}\,\|e^{-\Theta s} \mathbf{F}_0(0)\|_{P}
\;+\;
\sup_{0\le s\le t}\, \mathbb{E}\big[\|Z(s)\|_{P}\big]
\;\le\;
\sup_{0\le s\le t}\,\|e^{-\Theta s} \mathbf{F}_0(0)\|_{P}
\;+\;
\sqrt{n\, Q_{*,P}},
\end{equation}
where $Z(s):=\int_0^s e^{-\Theta(s-u)}\sigma I\, dW(u)$ is the centered OU component and the last inequality uses $\mathbb{E}\|X\|\le \sqrt{\mathbb{E}\|X\|^2}$ and $\mathbb{E}\|P^{1/2}X\|^2 = \mathrm{tr}(P^{1/2}\mathrm{Cov}(Z(s))P^{1/2})\le n\,\|P^{1/2}\mathrm{Cov}(Z(s))P^{1/2}\|$.

  \end{enumerate}
\end{remark}

\subsection{Tail behaviour of the Carleman matrix}

With the covariance and expectation bounds for the supremum of the multivariate OU process in \cref{lemma:sup-OU}, we can now establish the sub-Gaussian tail bound for the Carleman transfer matrix norm \(\Lambda(\omega)=\sup_{t\in[0,T]}\|L(t,\omega)\|\) defined in \cref{eq:LM}. This result is crucial for deriving high-probability error bounds for the Carleman linearization truncation of the OU-driven quadratic ODE in \cref{coro:P-eta} and MC integration for time discretization in \cref{lem:MC_error_inhomo} for the stochastic inhomogeneous term of LCHS.

\begin{lemma}\label[lemma]{lemma:Lambda-subG}
  Consider the Carleman linearization of the stochastic quadratic ODE driven by the multivariate OU process, 
\(\Lambda(\omega):=\sup_{t\in[0,T]}\|L(t,\omega)\|\) follows a sub-Gaussian tail bound:
\begin{equation}\label{eq:Lambda-subG}
\mathbb{P}\!\left(\ \Lambda(\omega)\ \ge\ x\ \right)
\ \le\ \exp\!\left(-\frac{(x - N\Big(\|\mathbf{F}_1\| + \|\mathbf{F}_2\| + \mathbb{E}\Big[\sup_{t\in[0,T]}\|\mathbf{F}_0(t)\|\Big]\Big))^2}{2\,N^2\,\sigma_*^2}\right)
\end{equation}
for any $x\ge N\Big(\|\mathbf{F}_1\| + \|\mathbf{F}_2\| + \mathbb{E}\Big[\sup_{t\in[0,T]}\|\mathbf{F}_0(t)\|\Big]\Big)$
with a high-probability threshold being
\begin{equation}\label{eq:Lambda-threshold-subG}
\Lambda_{\mathrm{threshold}}\ :=\ N\Big(\|\mathbf{F}_1\| + \|\mathbf{F}_2\| + \mathbb{E}\Big[\sup_{t\in[0,T]}\|\mathbf{F}_0(t)\|\Big] + \sqrt{2\,\sigma_*^2\,\log(1/\delta)}\Big)
\quad\Rightarrow\quad
\mathbb{P}\!\left(\ \Lambda(\omega)\ \le\ \Lambda_{\mathrm{threshold}}\ \right)\ \ge\ 1-\delta
\end{equation}
for any \(\delta\in(0,1)\), where $\sigma_*^2$ and $\mathbb{E}\Big[\sup_{t\in[0,T]}\|\mathbf{F}_0(t)\|\Big]$ are given by \cref{eq:sigma-star-bound} and \cref{eq:m-compact}, respectively. 
\end{lemma}
\begin{proof}
Combining \cref{eq:LM} and \cref{eq:M-norm-form} and using \(\|L(t,\omega)\|\le \|A(t,\omega)\|\) (since \(\|(A+A^\dagger)/2\|\le \|A\|\) in operator norm), we immediately obtain
\begin{equation}\label{eq:L-simplified}
\|L_N(t,\omega)\|\ \le\ \|\mathbf{A}_N(t,\omega)\|\ \le\ N\big(\,\|\mathbf{F}_0(t,\omega)\| + \|\mathbf{F}_1\| + \|\mathbf{F}_2\|\,\big).
\end{equation}
Taking the supremum of \cref{eq:L-simplified} over \(t\in[0,T]\) yields
\begin{equation}\label{eq:Lambda-sup}
\Lambda(\omega)\ :=\ \sup_{t\in[0,T]}\|L_N(t,\omega)\|\ \le\ N (\|\mathbf{F}_1\| + \|\mathbf{F}_2\| + \sup_{t\in[0,T]}\|\mathbf{F}_0(t,\omega)\|).
\end{equation}
Consider the centered Gaussian process defined in \cref{eq:Ztv}. We have \(\sup_{t}\|\mathbf{F}_0(t)\|=\sup_{t,v} Z(t,v)\). By Borell-TIS (Gaussian isoperimetry for separable Gaussian processes \citep[Theorem~2.1.1, Page 50]{robert_j_adler__jonathan_e_taylor_random_2007}), for any \(u>0\),
\begin{equation}\label{eq:Borell}
\mathbb{P}\!\left(\ \sup_{t,v} Z(t,v)\ \ge\ \mathbb{E}\big[\sup_{t,v} Z(t,v)\big]\ +\ u\ \right)\ \le\ \exp\!\left(-\frac{u^2}{2\,\sigma_*^2}\right),
\end{equation}
where \(\sigma_*^2:=\sup_{t,v}\mathrm{Var}[Z(t,v)]\) is given by \cref{eq:sigma-star-bound}.
Combining \cref{eq:Lambda-sup,eq:Borell}: for any \(x\ge N\big(\|\mathbf{F}_1\|+\|\mathbf{F}_2\|+\mathbb{E}[\sup_t\|\mathbf{F}_0(t)\|]\big)\), the event \(\{\Lambda\ge x\}\subseteq\{\sup_t\|\mathbf{F}_0(t)\|\ge x/N - \|\mathbf{F}_1\|-\|\mathbf{F}_2\|\}\).
Setting \(u=x/N - \|\mathbf{F}_1\|-\|\mathbf{F}_2\| - \mathbb{E}[\sup_t\|\mathbf{F}_0(t)\|]\ge 0\) in \cref{eq:Borell} gives \cref{eq:Lambda-subG}.
Substituting \(x=\Lambda_{\mathrm{threshold}}\) with \(u=\sqrt{2\sigma_*^2\log(1/\delta)}\) yields \cref{eq:Lambda-threshold-subG}.
\end{proof}

\section{Expectation and tail bounds for the solution norm of the OU-driven quadratic ODE}\label{app:x-tail}
We now provide expectation and tail bounds for the $P$-norm of the solution $\mathbf{x}(t)$ of the stochastic ODE \cref{eq:LODE} driven by the multivariate OU process $\mathbf{F}_0$ (\cref{eq:general_ou}) using the Borell-TIS inequality for $S_{t, P}$ (\cref{lemma:sup-OU}).
\begin{corollary}[Expectation and tail bounds of $\|\mathbf{x}(t)\|_P$]\label[corollary]{cor:borell_xP}
Under the assumptions and definitions of \cref{lem:lyap-path-solution-norm}, the following expectation and high-probability tail bounds hold for the $P$-norm of $\|\mathbf{x}(t)\|_P$:
\begin{equation}\label{eq:ExP-borell}
\mathbb{E}\big[\|\mathbf{x}(t)\|_P\big]\ \le\ \sqrt{a_t}\ +\ \sqrt{b_t}\,m_P
\ \le\ \sqrt{a_t}\ +\ \sqrt{b_t}\,C''\ \big\|P^{1/2}\boldsymbol{\Sigma}\big\|\ 
\sqrt{\frac{1-e^{-2\lambda_{\min} t}}{2\,\lambda_{\min}}}\ \Big(\sqrt{n}\ +\ \sqrt{1+\lambda_{\min} t}\Big),
\end{equation}
\begin{equation}\label{eq:Prob_xP_borell}
\mathbb{P}\!\left(\ \|\mathbf{x}(t)\|_P\ \le\ \sqrt{a_t}\ +\ \sqrt{b_t}\,\big(m_P+u_\delta\big)\ \right)\ \ge\ 1-\delta,
\end{equation}
where $\delta \in[0, 1]$ is a success probability and other parameters (e.g.,  $C^{''}>0, \lambda_{\min}>0$) are defined in \cref{lemma:sup-OU}.
\end{corollary}
\begin{proof}
\Cref{eq:ExP-borell} is a direct result of applying elementary calculus inequality $\sqrt{a+b s^2}\le \sqrt{a}+\sqrt{b}\,s$ to \cref{eq:xP-bound} and \cref{eq:m-compact} to $Z(s,v)=v^\top P^{1/2} \mathbf{F}_0(s)$ (i.e., with $\boldsymbol{\Sigma}$ replaced by $P^{1/2}\boldsymbol{\Sigma}$).
For any $\delta\in(0,1)$, define
\begin{equation}\label{eq:u_delta_borell}
u_\delta\ :=\ \sqrt{\,2\,\sigma_{*,P}^2\,\log(1/\delta)\,}\qquad\text{with}\qquad
\sigma_{*,P}^2\ \le\ \frac{1-e^{-2\lambda_{\min} t}}{2\,\lambda_{\min}}\ \big\|P^{1/2}\boldsymbol{\Sigma}\big\|^{2}
\end{equation}
by \cref{eq:sigma-star-bound} applied to $P^{1/2}\boldsymbol{\Sigma}$. Then Borell-TIS yields
\cref{eq:Prob_xP_borell}.
In particular, substituting the explicit $m_P$-bound from \cref{eq:m-compact} and the explicit $\sigma_{*,P}^2$-bound from \cref{eq:sigma-star-bound} (with $P^{1/2}\boldsymbol{\Sigma}$) yields fully explicit expectation and tail bounds for $\|\mathbf{x}(t)\|_P$ in terms of $(\lambda_{\min},\,\|P^{1/2}\boldsymbol{\Sigma}\|,\,n,\,t)$.
\end{proof}

\section{Bounds for the perturbation dynamics}\label{app:perturbation-bounds}

\Cref{lem:deltax-path} is the same as \cref{lem:lyap-path-solution-norm} but for the norm bound of the perturbation dynamics $\delta \mathbf{x}$ described in \cref{eq:stationary_condition} and is used in \cref{lem:deltax-carleman-path} and therefore \cref{thm:deltax-tail} to provide a probabilistic tail bound for the Carleman linearization error of $\delta \mathbf{x}$. 

\begin{lemma}[Pathwise norm bound of the perturbation]\label[lemma]{lem:deltax-path}
Let $\xst(t)$ be a pathwise stationary state of \cref{eq:main_system} in the sense of \cref{eq:stationary_condition}, and let $\delta \mathbf{x}:=\mathbf{x}-\xst$. The perturbation dynamics \cref{eq:perturbation_dynamics} hold with $J(t)=\mathbf{F}_1+\mathbf{F}_2 B_1(\xst(t))$.
Assume the Lyapunov inequalities of \cref{eq:lyapunov-ineq} and define
$
\kappa_0 \;:=\; 2\mu_P(\mathbf{F}_1) + 2\beta_P.
$
Then, the perturbation norm obeys the pathwise bound
\begin{equation}\label{eq:deltax-path-bound}
\|\delta \mathbf{x}(t)\|_P^2
\;\le\;
\exp\!\Big(\big[\kappa_0 + 2 C_{P,B}\, \sup_{0\le s\le t}\|\xst(s)\|_P\big]\, t\Big)\, \|\delta \mathbf{x}(0)\|_P^2,
\end{equation}
where the deterministic Jacobian growth constant
\begin{equation}\label{eq:CPB-def}
C_{P,B} \;:=\; \sup_{\|a\|_P=1}\, \Big\|\, P^{-1/2}\, \mathbf{F}_2\, B_1(a)\, P^{1/2} \,\Big\| <\infty
\end{equation}
is determined by the time-independent quadratic operator $\mathbf{F}_2$ in \cref{eq:main_system},
bilinear operator $B_1(\cdot)$ in \cref{def:bo}, and the Lyapunov matrix $P>0$ in \cref{eq:lyapunov-ineq}.
\end{lemma}

\begin{proof}
  The proof follows the same procedure as in \cref{lem:lyap-path-solution-norm} for the pathwise norm bound of $\mathbf{x}$. The main difference is that the governing equation of $\delta \mathbf{x}$, \cref{eq:perturbation_dynamics} in \cref{lemma:perturbation_dynamics}, is homogeneous and its linear coefficient is represented by the Jacobian $J(t)=\mathbf{F}_1+\mathbf{F}_2 B_1(\xst(t))$ in \cref{lemma:perturbation_dynamics}. Thus, applying the energy equation \cref{eq:dV-expanded} of \cref{lem:lyap-path-solution-norm} to $\delta \mathbf{x}$ yields
\begin{align}
\frac{d}{dt}\,\|\delta \mathbf{x}(t)\|_P^2
&= 2\,\Re\,\langle \delta \mathbf{x},\, J \,\delta \mathbf{x}\rangle_P \;+\; 2\,\Re\,\langle \delta \mathbf{x},\, \mathbf{F}_2((\delta \mathbf{x})^{\otimes 2})\rangle_P.
\label{eq:dx2}
\end{align}
To bound the first term in \cref{eq:dx2}, we first expand the Lyapunov stability criterion of $\mathbf{F}_1$ in \cref{eq:lyapunov-ineq-F1} to $J$.
For each $a\in\C^n$, define the linear operator
$
M(a) \;:=\; P^{-1/2} \mathbf{F}_2 B_1(a) P^{1/2}.
$
By the linearity of $B_1(\cdot)$ in $a$ and boundedness of $\mathbf{F}_2$, $M(a)$ is linear in $a$ and $\|M(a)\|<\infty$ for all $a$. For any $a\neq 0$ write $a=\|a\|_P\,\hat a$ with $\|\hat a\|_P=1$. Then
$
M(a) = \|a\|_P M(\hat a),
\|M(a)\| = \|a\|_P\, \|M(\hat a)\|
\;\le\; C_{P,B}\,\|a\|_P,
$
where $C_{P,B}$ is the supremum in \cref{eq:CPB-def}. Thus for all $a,z\in\C^n$,
\begin{align}
\big|\Re\langle z, P \mathbf{F}_2 B_1(a) z\rangle\big|
&= \big|\Re (P^{1/2}z)^\dagger M(a) (P^{1/2}z)\big| \\
&\le \|M(a)\|\, \|P^{1/2}z\|^2
\;\le\; C_{P,B}\,\|a\|_P\,\|z\|_P^2.
\label{eq:PB1-bound}
\end{align}
From \cref{eq:lyapunov-ineq-F1}, for any $z$,
\begin{equation}\label{eq:F1-LMI-compact}
z^\dagger(P\mathbf{F}_1+\mathbf{F}_1^\dagger P)z \;\le\; 2\mu_P(\mathbf{F}_1)\,\|z\|_P^2.
\end{equation}
For the Jacobian $J(t)=\mathbf{F}_1+\mathbf{F}_2 B_1(\xst(t))$, combine \cref{eq:F1-LMI-compact} with \cref{eq:PB1-bound} (with $a=\xst(t)$) to obtain, for all $z$,
\begin{equation}\label{eq:J-LMI-final}
z^\dagger\big(PJ(t)+J(t)^\dagger P\big)z
\;\le\;
\big(2\mu_P(\mathbf{F}_1) + 2 C_{P,B}\|\xst(t)\|_P\big)\,\|z\|_P^2.
\end{equation}
The second term in \cref{eq:dx2} can be bounded directly by the 
quadratic bound in \cref{eq:lyapunov-ineq} conditional on \cref{eq:J-LMI-final} with $z=\delta \mathbf{x}(t)$,
\[
2\,\Re\langle \delta \mathbf{x}, P \,\mathbf{F}_2((\delta \mathbf{x})^{\otimes 2})\rangle \;\le\; 2\beta_P\,\|\delta \mathbf{x}\|_P^2,
\]
which, together with \cref{eq:J-LMI-final} and \cref{eq:dx2} leads to
\[
\frac{d}{dt}\|\delta \mathbf{x}(t)\|_P^2
\;\le\;
\big(\kappa_0 + 2 C_{P,B}\|\xst(t)\|_P\big)\,\|\delta \mathbf{x}(t)\|_P^2.
\qquad \kappa_0=2\mu_P(\mathbf{F}_1)+2\beta_P.
\]
Grönwall’s inequality with $\sup_{0\le s\le t}\|\xst(s)\|_P$ yields \cref{eq:deltax-path-bound}.
\end{proof}

\begin{remark}[Deterministic $C_{P,B}$ and role of the OU process]
The Jacobian growth constant $C_{P,B}$ in \cref{eq:CPB-def} depends only on time-independent $(\mathbf{F}_2,B_1,P)$; it is defined via a supremum over the \emph{deterministic} unit $P$-sphere $\{a:\|a\|_P=1\}$ and an operator norm. 
For the special case $P=I$, the $P$-norm is the Euclidean norm, and
$
C_{I,B} \;:=\; \sup_{\|a\|=1}\, \big\|\,\mathbf{F}_2 B_1(a)\big\|_{\mathrm{op}}.
$
Using the explicit structure of $B_1(a)$ from \cref{def:bo}, one can bound $C_{I,B}$ in terms of $\|\mathbf{F}_2\|$ alone. Indeed, for any $a,b\in\mathbb{R}^n$,
$
\|B_1(a)b\| \;\le\; 4\,\|a\|\,\|b\|,
$
so for $\|a\|=1$,
$
\|\mathbf{F}_2 B_1(a)\|
= \sup_{\|b\|=1}\|\mathbf{F}_2 B_1(a)b\|
\;\le\; \|\mathbf{F}_2\| \sup_{\|b\|=1}\|B_1(a)b\|
\;\le\; 4\,\|\mathbf{F}_2\|.
$
Hence, a concrete deterministic bound is
\begin{equation}\label{eq:CI-explicit}
C_{I,B} \;\le\; 4\,\|\mathbf{F}_2\|.
\end{equation}
In the perturbation dynamics,
$
\dot{\delta \mathbf{x}} = \mathbf{F}_2((\delta \mathbf{x})^{\otimes 2}) + J(t)\,\delta \mathbf{x},
J(t)=\mathbf{F}_1 + \mathbf{F}_2 B_1(\xst(t)),
$
the OU process $\mathbf{F}_0$ enters only through the pathwise stationary state $\xst(t)$, which solves
$
\mathbf{F}_2(\xst(t)^{\otimes 2}) + \mathbf{F}_1 \xst(t) + \mathbf{F}_0(t) = 0.
$
Thus the randomness from the OU process affects the bound on $\|\delta \mathbf{x}(t)\|_P^2$ only via the factor
$
\sup_{0\le s\le t}\|\xst(s)\|_P,
$
appearing in \cref{eq:deltax-path-bound}. The constant $C_{P,B}$ itself remains deterministic and does not depend on the OU realization.
Combining \cref{eq:CI-explicit} with \cref{eq:deltax-path-bound} shows explicitly how the deterministic nonlinearity $\mathbf{F}_2$ (through $C_{I,B}$) and the random OU forcing (through $\xst$, hence $\sup_{s\le t}\|\xst(s)\|$) jointly determine the growth of $\|\delta \mathbf{x}(t)\|^2$.
\end{remark}

\begin{lemma}[Stationary state norm in terms of OU forcing]\label[lemma]{lem:xstar-bound}
Under \cref{eq:stationary_condition} and the Lyapunov bounds \cref{eq:lyapunov-ineq}, for each $t$,
\begin{equation}\label{eq:xstar-pointwise}
(-\mu_P(\mathbf{F}_1)-\beta_P)\, \|\xst(t)\|_P \;\le\; \|\mathbf{F}_0(t)\|_P,
\end{equation}
\end{lemma}

\begin{proof}
Take the $P$-inner product of \cref{eq:stationary_condition} with $\xst$ and apply \cref{eq:lyapunov-ineq}:
\[
0 = \Re\!\langle \xst, P(\mathbf{F}_2(\xst^{\otimes 2})+\mathbf{F}_1\xst+\mathbf{F}_0)\rangle
\le \mu_P(\mathbf{F}_1) \|\xst\|_P^2 + \beta_P \|\xst\|_P^2 + \|\xst\|_P\, \|\mathbf{F}_0\|_P.
\]
Rearrange to obtain \cref{eq:xstar-pointwise}.
\end{proof}

We now bound the pathwise Carleman linearization error for $\delta \mathbf{x}$ in \cref{lem:deltax-carleman-path}, which will be used for its probabilistic tail bound in \cref{thm:deltax-tail}.

\begin{lemma}[Pathwise Carleman linearization error for $\delta \mathbf{x}$]\label[lemma]{lem:deltax-carleman-path}
Let $\boldsymbol{\eta}^{(N)}(t)=\mathbf{y}^{(N)}(t)-\hat{\mathbf{y}}^{(N)}(t)$ be the order-$N$ Carleman linearization error for the perturbation dynamics \cref{eq:perturbation_dynamics} with the block-lift of \cref{pre:cl}. Assume the lifted generator satisfies the logarithmic norm bound $\mu_{P_N}(\mathbf{A}_N(\tau))\le \chi_P$ on $[0,t]$ as used in \cref{eq:lognorm-Pk-def} of \cref{lem:lyap-path}. Then
\begin{equation}\label{eq:deltax-eta-path}
\|\boldsymbol{\eta}^{(N)}(t)\|_{P_N}^2
\;\le\;
\Phi_t^2\, \|A_{N+1}^N\|_P^2\,
\Big(\exp\!\big([\kappa_0+2 C_{P,B}\, \sup_{0\le s\le t}\|\xst(s)\|_P] \,t\big)\, \|\delta \mathbf{x}(0)\|_P^2\Big)^{N+1},
\end{equation}
where $\Phi_t$ is defined in \cref{eq:params}, the Jacobian growth constant $C_{P, B}$ is defined in \cref{eq:CPB-def}, and the Lyapunov growth rate $\kappa_0$ is defined in \cref{lem:deltax-path}.
\end{lemma}

\begin{proof}
Repeat the variation-of-constants and energy estimate in \cref{lem:lyap-path}, replacing $\|\mathbf{x}(s)\|_P$ by $\|\delta \mathbf{x}(s)\|_P$ and bounding $\sup_{0\le s\le t}\|\delta \mathbf{x}(s)\|_P^2$ by \cref{eq:deltax-path-bound}. The residual bound is the same as \cref{eq:xi-bound2} because $A_{N+1}^N$ and $\mathbf{y}^{(N+1)}$ depend only on $\mathbf{F}_2$ and the monomials of the state. Squaring the resulting inequality gives \cref{eq:deltax-eta-path}.
\end{proof}

\section{OU increment and continuous bounds}
\label{sec:ou-holder}

The OU process is continuous everywhere but nowhere differentiable. We review several canonical continuity theorems for the OU process that will be used for the Riemann integration of the time discretization in \cref{sec:riemann-integral}.  
\begin{preliminaries}[Kolmogorov--Chentsov Continuity Theorem.]\label[preliminary]{preli:KCC}
  Let $\{X_t\}_{t \in [0,T]}$ be a stochastic process taking values in a Banach space. If there exist constants $p>0$, $\delta>0$, and $C>0$ such that for all $s,t \in [0,T]$ \citep[Theorem~2.8, Page 53]{karatzas_brownian_1998},
\begin{equation}\label{eq:KC-condition}
    \EE\norm{X_t - X_s}^p \le C |t-s|^{1+\delta},
\end{equation}
then there exists a modification of $X_t$, say $\tilde{X}_t$, which is almost surely locally Hölder continuous with any exponent $\alpha \in (0, \delta/p)$. Specifically, for any such $\alpha$, the random variable
\[ K_\alpha' := \sup_{0 \le s < t \le T} \frac{\norm{\tilde{X}_t - \tilde{X}_s}}{|t-s|^\alpha} \]
is almost surely finite, and its $p$-th moment is bounded, i.e., $\EE[(K_\alpha')^p] < \infty$ \citep[equation 2.21, Page 55]{karatzas_brownian_1998}.
\end{preliminaries}

\begin{preliminaries}[Burkholder-Davis-Gundy (BDG) Inequality.]
Let $M_t$ be a continuous local martingale taking values in $\mathbb{R}^d$. For any $p \ge 1$, there exists a universal constant $c_p > 0$ such that for any stopping time $\tau$,
\begin{equation}\label{eq:BDG}
    \EE\left[\sup_{0 \le t \le \tau} \norm{M_t}^p\right] \le c_p \EE\left[ \langle M \rangle_\tau^{p/2} \right],
\end{equation}
where $\langle M \rangle_t$ is the quadratic variation process of $M_t$. For a stochastic integral of the form $M_t = \int_0^t \Phi_s \,dW_s$, where $\Phi_s$ is a predictable matrix-valued process and $W_s$ is a standard Brownian motion, the quadratic variation is $\langle M \rangle_t = \int_0^t \norm{\Phi_s}_{\mathrm{HS}}^2 \,ds$, where $\norm{\cdot}_{\mathrm{HS}}$ is the Hilbert-Schmidt norm \citep[Theorem~4.1, Page 160]{revuz_continuous_1999}.
\end{preliminaries}

\begin{assumption}[OU semigroup stability and moment bounds]\label{ass:OU}
  For the multivariate OU process given by \cref{eq:general_ou}, there exist constants $M\ge 1$ and $\vartheta_0>0$ such that
\begin{equation}\label{eq:semigroup-bound}
  \norm{e^{-\Theta t}} \le M e^{-\vartheta_0 t}, \qquad \forall t\ge 0.
\end{equation}
Moreover, for some $p_\star\ge 2$, the OU process admits bounded $L^{p_\star}$ moments on $[0,T]$:
\begin{equation}\label{eq:F-moment}
  M_{F,p_\star}:=\sup_{0\le t\le T} \left(\EE \norm{\mathbf{F}_0(t)}^{p_\star}\right)^{1/p_\star}<\infty.
\end{equation}
\end{assumption}

In finite dimension, every matrix generates a $C_0$-semigroup with a growth bound $\|e^{At}\|\le M e^{\omega t}$; exponential stability ($\omega<0$) is equivalent to the spectral bound being negative \cite[Theorem~II.1.10]{engel_semigroups_2000}  \cite[Section~1.2]{pazy_semigroups_1983}. For the OU process, existence, Gaussianity, and the explicit mild form \cref{eq:MOUsemiclosed} follow from standard SDE theory \cite[Section~3.2; Theorem~5.2.1]{oksendal_stochastic_2003} \cite[Section~5.6]{karatzas_brownian_1998}. Finite-time $L^p$-moment bounds for solutions of SDEs with globally Lipschitz/linear growth coefficients (satisfied here) are standard \cite[Chapter~4, Theorem~4.5.4]{kloeden_numerical_1992}. If $\mathbf{F}_0$ is stationary and $\Theta$ is Hurwitz, the covariance $P$ solves the Lyapunov equation $\Theta P+P\Theta^\top=\Sigma\Sigma^\top$ \cite[Proposition~3.5]{pavliotis_stochastic_2014}.

\begin{lemma}[OU increments in $L^p$]\label[lemma]{lem:OU-increments}
Under \cref{ass:OU}, for any $p\in[2,p_\star]$ there exists a constant $C_{\mathrm{OU},p}<\infty$ (depending on $\Theta,\Sigma,M,\vartheta_0,p$, and $M_{F,p}$) such that for all $0\le s\le t\le T$,
\begin{equation}\label{eq:OU-inc-Lp}
  \left(\EE \norm{\mathbf{F}_0(t)-\mathbf{F}_0(s)}^{p}\right)^{1/p} \le C_{\mathrm{OU},p}\, |t-s|^{1/2}.
\end{equation}
Moreover, for $p=2$ the following explicit bound holds:
\begin{equation}\label{eq:OU-inc-L2-const}
  C_{\mathrm{OU},2} \;\le\; \frac{M\norm{\Sigma}_F}{\sqrt{2\vartheta_0}} \;+\; \frac{M\|\Theta\|}{\vartheta_0}\, T^{1/2}\, M_{F,2}.
\end{equation}
\end{lemma}

\begin{proof}
Fix $0\le s<t\le T$ and set $\Delta:=t-s$. From \cref{eq:MOUsemiclosed}, we obtain the increment decomposition
\begin{equation}\label{eq:inc-decomp}
  \mathbf{F}_0(t)-\mathbf{F}_0(s) \;=\; \underbrace{\left(e^{-\Theta \Delta}-I\right)\mathbf{F}_0(s)}_{=:D_\Delta} \;+\; \underbrace{\int_s^t e^{-\Theta (t-r)} \Sigma\, dW(r)}_{=:M_\Delta}.
\end{equation}
We first establish the $L^2$ estimate, then extend to $L^p$ for $p>2$.

\emph{Bound for $M_\Delta$ in $L^2$.} By It\^o isometry for vector-valued stochastic integrals \cite[Corollary~3.1.7 \& Lemma 3.1.5]{oksendal_stochastic_2003}, we have
\begin{equation}\label{eq:ito-isometry}
  \EE \norm{M_\Delta}^2 \;=\; \int_s^t \norm{e^{-\Theta (t-r)} \Sigma}_F^2\, dr.
\end{equation}
With the change of variable $u:=t-r$ in \cref{eq:ito-isometry} and using the semigroup bound \cref{eq:semigroup-bound},
\begin{align}
  \EE \norm{M_\Delta}^2
  &= \int_0^\Delta \norm{e^{-\Theta u}\Sigma}_F^2\, du
   \;\le\; \int_0^\Delta \norm{e^{-\Theta u}}^2\, \norm{\Sigma}_F^2\, du
   \;\le\; M^2\norm{\Sigma}_F^2 \int_0^\Delta e^{-2\vartheta_0 u}\, du \nonumber\\
  &\le\; \frac{M^2\norm{\Sigma}_F^2}{2\vartheta_0}\, \Delta. \label{eq:ito-bound}
\end{align}
Taking square roots in \cref{eq:ito-bound} yields
\begin{equation}\label{eq:M-L2}
  \left(\EE \norm{M_\Delta}^2\right)^{1/2} \;\le\; \frac{M\norm{\Sigma}_F}{\sqrt{2\vartheta_0}}\, \Delta^{1/2}.
\end{equation}

\emph{Bound for $D_\Delta$ in $L^2$.} Using the identity
\begin{equation}\label{eq:exp-diff}
  e^{-\Theta \Delta}-I \;=\; -\int_0^\Delta \Theta e^{-\Theta r}\, dr,
\end{equation}
which follows from differentiating $e^{-\Theta r}$ and integrating \cite[Proposition~II.1.3]{engel_semigroups_2000}, together with \cref{eq:semigroup-bound}, we obtain
\begin{align}
  \norm{e^{-\Theta \Delta}-I}
  &\le \int_0^\Delta \norm{\Theta}\, \norm{e^{-\Theta r}}\, dr
   \;\le\; M\|\Theta\| \int_0^\Delta e^{-\vartheta_0 r}\, dr
   \;\le\; \frac{M\|\Theta\|}{\vartheta_0}\, \Delta. \label{eq:Ctheta}
\end{align}
Combining \cref{eq:inc-decomp} and \cref{eq:Ctheta} gives
\begin{equation}\label{eq:D-L2-pre}
  \norm{D_\Delta} \;=\; \norm{\left(e^{-\Theta \Delta}-I\right)\mathbf{F}_0(s)} \;\le\; \frac{M\|\Theta\|}{\vartheta_0}\, \Delta\, \norm{\mathbf{F}_0(s)}.
\end{equation}
Taking $L^2$ norms in \cref{eq:D-L2-pre} and using \cref{eq:F-moment} yields
\begin{equation}\label{eq:D-L2}
  \left(\EE \norm{D_\Delta}^2\right)^{1/2} \;\le\; \frac{M\|\Theta\|}{\vartheta_0}\, \Delta\, \left(\EE \norm{\mathbf{F}_0(s)}^2\right)^{1/2}
  \;\le\; \frac{M\|\Theta\|}{\vartheta_0}\, \Delta\, M_{F,2}.
\end{equation}

\emph{$L^2$ bound for the increment.} By the triangle inequality in $L^2$ (Minkowski), applied to \cref{eq:inc-decomp},
\begin{equation}\label{eq:L2-triangle}
  \left(\EE \norm{\mathbf{F}_0(t)-\mathbf{F}_0(s)}^2\right)^{1/2} \;\le\; \left(\EE \norm{M_\Delta}^2\right)^{1/2} \;+\; \left(\EE \norm{D_\Delta}^2\right)^{1/2}.
\end{equation}
Plugging \cref{eq:M-L2} and \cref{eq:D-L2} into \cref{eq:L2-triangle} gives
\begin{equation}\label{eq:L2-sum}
  \left(\EE \norm{\mathbf{F}_0(t)-\mathbf{F}_0(s)}^2\right)^{1/2} \;\le\; \frac{M\norm{\Sigma}_F}{\sqrt{2\vartheta_0}}\, \Delta^{1/2} \;+\; \frac{M\|\Theta\|}{\vartheta_0}\, M_{F,2}\, \Delta.
\end{equation}
Since $\Delta \le T^{1/2}\Delta^{1/2}$ for $0\le \Delta \le T$, it follows from \cref{eq:L2-sum} that
\begin{equation}\label{eq:L2-final}
  \left(\EE \norm{\mathbf{F}_0(t)-\mathbf{F}_0(s)}^2\right)^{1/2} \;\le\; \left(\frac{M\norm{\Sigma}_F}{\sqrt{2\vartheta_0}} \;+\; \frac{M\|\Theta\|}{\vartheta_0}\, T^{1/2}\, M_{F,2}\right)\, \Delta^{1/2}.
\end{equation}
Setting $C_{\mathrm{OU},2}$ equal to the parenthesis in \cref{eq:L2-final} proves \cref{eq:OU-inc-Lp} for $p=2$, together with \cref{eq:OU-inc-L2-const}.

\emph{Extension to $L^p$, $p>2$.} Taking $L^p$ norms in \cref{eq:inc-decomp} and using the triangle inequality in $L^p$,
\begin{equation}\label{eq:Lp-triangle}
  \norm{\mathbf{F}_0(t)-\mathbf{F}_0(s)}_{L^p} \;\le\; \norm{M_\Delta}_{L^p} \;+\; \norm{D_\Delta}_{L^p}.
\end{equation}
For the martingale term, we use the Burkholder--Davis--Gundy (BDG) inequality \cite[Chapter~IV, Theorem~4.1]{revuz_continuous_1999} combined with \cref{eq:ito-isometry}:
\begin{align}
  \norm{M_\Delta}_{L^p}
  &\le C_p \left\| \left(\int_s^t \norm{e^{-\Theta (t-r)}\Sigma}_F^2 dr\right)^{1/2} \right\|_{L^p}
   \;=\; C_p \left(\int_0^\Delta \norm{e^{-\Theta u}\Sigma}_F^2 du\right)^{1/2} \nonumber\\
  &\le C_p\, M \norm{\Sigma}_F \left(\int_0^\Delta e^{-2\vartheta_0 u} du\right)^{1/2}
   \;\le\; \frac{C_p\, M\norm{\Sigma}_F}{\sqrt{2\vartheta_0}}\, \Delta^{1/2}. \label{eq:M-Lp}
\end{align}
For the drift term, by \cref{eq:D-L2-pre} (which holds pathwise) and \cref{eq:F-moment},
\begin{equation}\label{eq:D-Lp}
  \norm{D_\Delta}_{L^p} \;\le\; \frac{M\|\Theta\|}{\vartheta_0}\, \Delta\, \norm{\mathbf{F}_0(s)}_{L^p}
  \;\le\; \frac{M\|\Theta\|}{\vartheta_0}\, \Delta\, M_{F,p}.
\end{equation}
Combining \cref{eq:Lp-triangle}, \cref{eq:M-Lp}, and \cref{eq:D-Lp}, we get
\begin{equation}\label{eq:Lp-sum}
  \norm{\mathbf{F}_0(t)-\mathbf{F}_0(s)}_{L^p} \;\le\; \frac{C_p\, M\norm{\Sigma}_F}{\sqrt{2\vartheta_0}}\, \Delta^{1/2} \;+\; \frac{M\|\Theta\|}{\vartheta_0}\, M_{F,p}\, \Delta.
\end{equation}
Using again $\Delta \le T^{1/2}\Delta^{1/2}$, we infer from \cref{eq:Lp-sum} that
\begin{equation}\label{eq:Lp-final}
  \norm{\mathbf{F}_0(t)-\mathbf{F}_0(s)}_{L^p} \;\le\; \left(\frac{C_p\, M\norm{\Sigma}_F}{\sqrt{2\vartheta_0}} \;+\; \frac{M\|\Theta\|}{\vartheta_0}\, T^{1/2}\, M_{F,p}\right)\, \Delta^{1/2}.
\end{equation}
Setting $C_{\mathrm{OU},p}$ equal to the parenthesis in \cref{eq:Lp-final} proves \cref{eq:OU-inc-Lp} for $p\in[2,p_\star]$.
\end{proof}

\begin{lemma}[Kolmogorov--Chentsov continuity for OU]\label[lemma]{lem:KC}
Under \cref{ass:OU}, if $p_\star > 2$, then for any $\alpha \in \left(0, \frac{1}{2} - \frac{1}{p_\star}\right)$, there exists a continuous modification of the process $\mathbf{F}_0(t)$ and a finite random variable $K_\alpha(\omega)$ along a sampling path $\omega$ (realization) such that
\begin{equation}\label{eq:KC-modulus}
  \sup_{\substack{0\le s<t\le T\\ t-s\le h}} \norm{\mathbf{F}_0(t,\omega)-\mathbf{F}_0(s,\omega)} \le K_\alpha(\omega)\, h^\alpha,\qquad \forall h\in(0,T].
\end{equation}
Moreover, $K_\alpha(\omega)$ has finite moments up to order $p_\star$, i.e., $\EE[K_\alpha(\omega)^q] < \infty$ for all $q \le p_\star$.
\end{lemma}

\begin{proof}
  The moment condition \cref{eq:KC-condition} of the Kolmogorov--Chentsov theorem for the OU process $\mathbf{F}_0(t)$ is given by \cref{lem:OU-increments}. To apply the Kolmogorov--Chentsov theorem \cref{eq:KC-condition}, we set the exponent $p/2 = 1 + \delta$. This requires $\delta = p/2 - 1 > 0$ implying $p/2 > 1$, or $p>2$.
 Preliminary~\ref{preli:KCC} then guarantees the existence of a modification of $\mathbf{F}_0(t)$ that is almost surely Hölder continuous for any exponent $\alpha$ satisfying
\[
    \alpha < \frac{\delta}{p} = \frac{p/2 - 1}{p} = \frac{1}{2} - \frac{1}{p}.
\]
This holds for any $p \in (2, p_\star]$. To find the largest possible range for $\alpha$, we can take $p$ arbitrarily close to $p_\star$. Thus, the conclusion holds for any $\alpha < \sup_{p \in (2, p_\star]} (1/2 - 1/p) = 1/2 - 1/p_\star$.
Finally, the Preliminary~\ref{preli:KCC} also states that the random Hölder constant $K_\alpha$ has a finite $p$-th moment, $\EE[K_\alpha^p] < \infty$ \cite[Theorem~4.23, Page~95]{olav_kallenberg_foundations_1997}. As this holds for any $p \in (2, p_\star]$, by Jensen's inequality, it implies that $\EE[K_\alpha^q] < \infty$ for all $q \le p_\star$. This completes the proof.
\end{proof}

\begin{lemma}[Lipschitz continuity of the evolution in $s$]\label[lemma]{lem:U-lip}
Let $G(\cdot,\omega)$ be a bounded measurable family of Hermitian matrices on $[0,T]$. 
Let $U(T,s,k,\omega)$ be the (stochastic) unitary evolution arising from a time-dependent bounded Hermitian generator $G(s,\omega):=k L(s,\omega)+H(s,\omega)$ in finite dimension, so that $U(T,s,k,\omega)$ is unitary for each $s$ and $\omega$
\begin{equation}\label{eq:U-lip}
  \norm{U(T,s,k,\omega)-U(T,s',k,\omega)} \le \int_{s'}^{s} \norm{G(r,\omega)}\, dr,\qquad 0\le s'\le s\le T.
\end{equation}
Then for all $0\le s'\le s\le T$,
\begin{equation}\label{eq:U-diff}
  \norm{U(T,s,\omega)-U(T,s',\omega)} \le \int_{s'}^{s} \norm{G(r,\omega)}\, dr.
\end{equation}
In particular, if $L_U(\omega):=\sup_{r\in[0,T]}\norm{G(r,\omega)}<\infty$, then
\begin{equation}\label{eq:U-lip-sup}
  \norm{U(T,s,\omega)-U(T,s',\omega)} \le L_U(\omega)\, |s-s'|.
\end{equation}
\end{lemma}

\begin{proof}
  In finite dimension, $U(T,\cdot,\omega)$ is strongly differentiable and satisfies $\partial_s U(T,s,\omega)= i\, U(T,s,\omega)\, G(s,\omega)$ (\citep[Equation 178]{an_quantum_2025}). By the fundamental theorem of calculus for Bochner integrals,
\[
  U(T,s,\omega)-U(T,s',\omega) = \int_{s'}^{s} \partial_r U(T,r,\omega)\, dr
  = i \int_{s'}^{s} U(T,r,\omega) G(r,\omega)\, dr.
\]
Taking norms and using $\norm{U(T,r,\omega)}=1$ gives \cref{eq:U-diff}. The bound \cref{eq:U-lip-sup} follows immediately \citep[Theorem~5.3.1]{pazy_semigroups_1983}.
\end{proof}
 
\section{Discretizing the inhomogeneous LCHS for SDEs under Hölder continuity}
\label{sec:riemann-integral}

Since both $A(s^\prime)$ and $b(s)$ in \cref{eq:b-LCHS} are stochastic in time, we use a Riemann sum to discretize the time integral instead of Gaussian quadrature, which requires sufficient smoothness of $b(s)$. Specifically, we replace derivative bounds with probabilistic regularity estimates under Hölder continuity. The Riemann-sum approach serves as an alternative to the MC sampling in \cref{lem:MC_error_inhomo}. 

\begin{lemma}[Time discretization for stochastic inhomogeneous term of LCHS]\label[lemma]{lem:quadrature_error_inhomo}
    Consider the discretization in~\cref{eq:lchs-b}. 
    For any $\eps > 0$, in order to bound the approximation error of~\cref{eq:lchs-b} by $\eps$, it suffices to choose
\begin{align}
K &= \mathcal{O}\left( \left(\log\left(1+\frac{\E[\|b\|_{L^1}]}{\eps}\right)\right)^{1/\beta} \right) \\
h_1 &= \frac{1}{eT \max_t\E[\|L(t,\omega)\|]} \\
Q_1 &= \mathcal{O}\left( \log\left( 1+\frac{\E[\|b\|_{L^1}]}{\eps} \right) \right) \\
h_2 &= \mathcal{O}\left(\left(\frac{\eps}{\E[C_\gamma(\omega)^2]}\right)^{1/(2\gamma)}\right)
\end{align}
where $h_2$ is a time step for the Riemann sum.
\end{lemma}

\begin{proof}
We first discretize the $k$ space and bound the discretization error probabilistically. For a given path $w$ and time $s$, the $k$-discretization error is the same as \cref{eq:gk} since it only depends on the deterministic kernel function \cref{eqn:quadrature}. By applying \citet[Equations 207\&208]{an_quantum_2025}, we have:
\begin{align}
& \E\left[\left\| \int_0^T \mathcal{T}e^{-\int_s^T A(s',\omega)\ud s'} b(s,\omega) \ud s - \int_0^T \sum_{m_1 = -K/h_1}^{K/h_1-1} \sum_{q_1=0}^{Q_1-1} c_{q_1,m_1} U(T,s,k_{q_1,m_1},\omega)  b(s,\omega)  \ud s \right\|^2\right] \\
&\quad \leq \E[\|b(\cdot,\omega)\|_{L^1}^2] \max_s \left\| \mathcal{T}e^{-\int_s^T A(s',\omega)\ud s'} - \sum_{m_1 = -K/h_1}^{K/h_1-1} \sum_{q_1=0}^{Q_1-1} c_{q_1,m_1} U(T,s,k_{q_1,m_1},\omega)  \right\|^2,
\end{align}
where
\begin{equation}
\E[\|b(\cdot,\omega)\|_{L^1}^2] = \E\left[\left(\int_0^T \|b(s,\omega)\| ds\right)^2\right] \leq T \int_0^T \E[\|b(s,\omega)\|^2] ds = T^2 \E[\|b(0)\|^2]
\end{equation}

Next, we use a Riemann sum for time discretization (as opposed to the Gaussian quadrature discretization of \citet[Equation~212]{an_quantum_2025}) for a given time step $h_2$. The corresponding Riemann-sum error is:
\begin{align}
\mathcal{E}(\omega)
&:=
\left\| \int_0^T U(T,s,k,\omega)  b(s,\omega)  \ud s - \sum^{T/h_2-1}_{m_2 = 0} h_2 U(T,s_{m_2},k,\omega)  b(s_{m_2},\omega) \right\| \\
\label{eq:riemann-sum-diff}
&\leq \sum_{m_2=0}^{T/h_2-1} \int_{s_{m_2}}^{s_{m_2+1}} \|U(T,s,k,\omega) b(s,\omega) - U(T,s_{m_2},k,\omega) b(s_{m_2},\omega)\| ds
\end{align}

Using the triangle inequality and the fact that $\|U(T,s,k,\omega)\| \equiv 1$ (unitary):
\begin{align}
  \text{\cref{eq:riemann-sum-diff}} & \leq \sum_{m_2=0}^{T/h_2-1} \int_{s_{m_2}}^{s_{m_2+1}} \left[\underbrace{\norm{U(T,s,k,\omega)}\|b(s,\omega) - b(s_{m_2},\omega)\|}_{\text{first term}} + \underbrace{\|b(s_{m_2},\omega)\| \|U(T,s,k,\omega) - U(T,s_{m_2},k,\omega)\|}_{\text{second term}} \right] ds \\
&= \sum_{m=0}^{N-1} \int_{s_m}^{s_{m+1}} \norm{b(s,\omega)-b(s_m,\omega)}\, ds \nonumber\\
&\quad + \sum_{m=0}^{N-1} \int_{s_m}^{s_{m+1}} \norm{b(s_m,\omega)}\norm{U(T,s,k,\omega)-U(T,s_m,k,\omega)}\, ds \nonumber\\
&=: A_1(\omega) + A_2(\omega), \label{eq:err-split}
\end{align}

For the first term, using OU Hölder continuity of \cref{eq:KC-modulus}:
\begin{align}
A_1(\omega) = \sum_{m_2=0}^{T/h_2-1} \int_{s_{m_2}}^{s_{m_2+1}} \|b(s,\omega) - b(s_{m_2},\omega)\| ds &\leq \sum_{m_2=0}^{T/h_2-1} \int_{s_{m_2}}^{s_{m_2+1}} C_\gamma(\omega) |s - s_{m_2}|^\gamma ds \\
&= C_\gamma(\omega) \sum_{m_2=0}^{T/h_2-1} \int_0^{h_2} t^\gamma dt \\
&= C_\gamma(\omega) \frac{T}{h_2} \cdot \frac{h_2^{\gamma+1}}{\gamma+1} \\
&= \frac{T C_\gamma(\omega)}{\gamma+1} h_2^\gamma.  \label{eq:A1-path}
\end{align}

Squaring \cref{eq:A1-path} and taking expectations yields
\begin{equation}\label{eq:A1-L2}
  \EE[A_1^2] \le \left(\frac{T}{\gamma+1}\right)^2\, \EE\!\left[C_\gamma(\omega)^2\right]\, h_2^{2\gamma}.
\end{equation}

For the second term, 
from \cref{lem:U-lip}, for $s\in[s_m,s_{m+1}]$,
\[
  \norm{U(T,s,k,\omega)-U(T,s_m,k,\omega)} \le L_U(\omega)\, (s-s_m)
  \quad\text{with}\quad L_U(\omega):=\sup_{r\in[0,T]}\norm{G(r,\omega)}.
\]
Hence
\begin{align}
A_2(\omega)
&= \sum_{m=0}^{N-1} \int_{s_m}^{s_{m+1}} \norm{b(s_m,\omega)}\, \norm{U(T,s,k,\omega)-U(T,s_m,k,\omega)}\, ds \nonumber\\
&\le L_U(\omega)\, \sum_{m=0}^{N-1} \norm{b(s_m,\omega)} \int_{s_m}^{s_{m+1}} (s-s_m)\, ds
= \frac{L_U(\omega)\, h_2^2}{2} \sum_{m=0}^{N-1} \norm{b(s_m,\omega)}. \label{eq:A2-sup-1}
\end{align}
Applying Cauchy--Schwarz to the sum gives
\begin{equation}\label{eq:A2-sup-2}
  \sum_{m=0}^{N-1} \norm{b(s_m,\omega)} \le \sqrt{N}\, \left(\sum_{m=0}^{N-1} \norm{b(s_m,\omega)}^2\right)^{1/2}
  = \sqrt{\frac{T}{h_2}}\, \left(\sum_{m=0}^{N-1} \norm{b(s_m,\omega)}^2\right)^{1/2}.
\end{equation}
Combining \cref{eq:A2-sup-1} and \cref{eq:A2-sup-2} yields the pathwise bound
\begin{equation}\label{eq:A2-sup-3}
  A_2(\omega) \le \frac{\sqrt{T}}{2}\, L_U(\omega)\, h_2^{3/2}\, \left(\sum_{m=0}^{N-1} \norm{b(s_m,\omega)}^2\right)^{1/2}.
\end{equation}
Squaring \cref{eq:A2-sup-3} and taking expectations, we obtain
\begin{equation}\label{eq:A2-sup-4}
  \EE[A_2^2] \le \frac{T}{4}\, h_2^{3}\, \EE\!\left[ L_U(\omega)^2 \sum_{m=0}^{N-1} \norm{b(s_m,\omega)}^2 \right].
\end{equation}
Using Cauchy–Schwarz/Hölder with p = q = 2,
\begin{equation}\label{eq:A2-sup-5}
  \EE\!\left[ L_U(\omega)^2 \sum_{m=0}^{N-1} \norm{b(s_m,\omega)}^2 \right]
  \le \left(\EE[L_U(\omega)^2]\right)^{1/2} \left(\EE\left[\left(\sum_{m=0}^{N-1} \norm{b(s_m,\omega)}^2\right)^2\right]\right)^{1/2},
\end{equation}
or more crudely, using $\EE\sum \norm{b(s_m)}^2 \le N\, \sup_t \EE\norm{b(t)}^2$,
\begin{equation}\label{eq:A2-sup-6}
  \EE\!\left[ L_U(\omega)^2 \sum_{m=0}^{N-1} \norm{b(s_m,\omega)}^2 \right]
  \le \EE[L_U(\omega)^2]\, N\, M_{F,2}^2
  = \EE[L_U(\omega)^2]\, \frac{T}{h_2}\, M_{F,2}^2.
\end{equation}
Plugging \cref{eq:A2-sup-6} into \cref{eq:A2-sup-4} gives
\begin{equation}\label{eq:A2-sup-final}
  \EE[A_2^2] \le \frac{T}{4}\, h_2^{3}\, \frac{T}{h_2}\, \EE[L_U(\omega)^2]\, M_{F,2}^2
  = \frac{T^2}{4}\, \EE[L_U(\omega)^2]\, M_{F,2}^2\, h_2^{2}.
\end{equation}
Thus $A_2$ contributes order $h_2^2$ in mean square.

From \cref{eq:err-split} and $(x+y)^2\le 2x^2+2y^2$,
\begin{equation}\label{eq:err-ms}
  \EE[\mathcal{E}^2] \le 2\, \EE[A_1^2] + 2\, \EE[A_2^2].
\end{equation}
Using \cref{eq:A1-L2} and \cref{eq:A2-sup-final}, we obtain the canonical form
\begin{equation}\label{eq:err-ms-final}
  \EE[\mathcal{E}^2] \;\le\; K_1\, h_2^{2\gamma} \;+\; K_2\, h_2^{2},
\end{equation}
with explicit constants, 
\[
  K_1 = 2 \left(\frac{T}{\gamma+1}\right)^2\, \EE[C_\gamma^2],\qquad
  K_2 = \frac{T^2}{2}\, \EE[L_U(\omega)^2]\, M_{F,2}^2.
\]
Since $\gamma<1/2$, the leading term as $h_2\downarrow 0$ is $h_2^{2\gamma}$ in \cref{eq:err-ms-final}. To enforce $\EE[\mathcal{E}^2]\le \eps$ it suffices to choose
\begin{equation}\label{eq:h-choice}
  h_2 \;\le\; \min\left\{ \left(\frac{\eps}{2K_1}\right)^{1/(2\gamma)},\ \left(\frac{\eps}{2K_2}\right)^{1/2} \right\}.
\end{equation}
In particular, if $\eps$ is small, the $h_2^{2\gamma}$ constraint dominates, and one can take
\begin{equation}\label{eq:h-dominant}
  h_2 \le \left(\frac{\eps}{2K_1}\right)^{1/(2\gamma)} \quad\Longleftrightarrow\quad h_2 \asymp \eps^{1/(2\gamma)},\qquad \gamma<1/2,
\end{equation}
which is slower than the deterministic Lipschitz case (where $\gamma=1$ would give $h_2\asymp \eps^{1/2}$). We remark that the query complexity of implementing LCU scales as $\log(h_2)$.
\end{proof}

We denote
\begin{equation}\label{eq:LambdaCalpha}
\Lambda(\omega):=\sup_{t\in[0,T]} \norm{H(t,\omega)}<\infty,\qquad
\text{and}\qquad
\norm{H(t,\omega)-H(s,\omega)} \le C_\alpha(\omega)\, |t-s|^\alpha
\end{equation}
for some $\alpha\in(0,1]$ and a finite random variable $C_\alpha(\omega)$ (pathwise Hölder modulus); see  \cref{sec:ou-holder} for  discussion on Hölder continuity of the OU process.
This bound is necessary to show that the forcing term for the OU process varies smoothly.
\begin{lemma}[Pathwise $k$-fold Riemann integral discretization error under Hölder continuity]\label[lemma]{lem:Ik-Qk}
Assume \cref{eq:LambdaCalpha} holds. Then, for each fixed $\omega$ and $k\ge 1$, the Riemann integral discretization error under Hölder continuity can be bounded by

\begin{align}
  \mathcal{E}_M(\omega) & = \sum_{k=1}^{K} \norm{\mathcal{I}_k[H]-\mathcal{Q}_{k,M}[H]} \\
& \le \eps_2
\end{align}
by choosing any time step $h$ satisfying
\begin{equation}\label{eq:M-choice}
h \le \left(\frac{\eps_2}{C_\alpha(\omega)\, T\, e^{\Lambda(\omega) T}}\right)^{1/\alpha},
\end{equation}
for a given error $\eps_2 > 0$.
\end{lemma}

\begin{proof}
For $t\in[0,T)$, let $\tau(t):= h \lfloor t/h\rfloor$ be the left gridpoint, so $0\le t-\tau(t)<h$. Given $(t_1,\dots,t_k)\in[0,T]^k$, define $\tau_j:=\tau(t_j)$. On the set $\{0\le t_1\le\cdots\le t_k\le T\}$, we have $\tau_1\le \cdots\le \tau_k$ and $\tau_j\in\{t_0,\dots,t_{M-1}\}$.
For $(t_1,\dots,t_k)$ in the ordered set, apply \cref{prop:telescoping} with $A_j=H(t_j,\omega)$ and $B_j=H(\tau_j,\omega)$:
\begin{equation}\label{eq:pointwise-product-diff}
\big\| H(t_k,\omega)\cdots H(t_1,\omega) - H(\tau_k)\cdots H(\tau_1)\big\|
\le \Lambda^{k-1} \sum_{j=1}^{k} \norm{H( t_j,\omega)-H(\tau_j,\omega)}.
\end{equation}
Using the Hölder condition in \cref{eq:LambdaCalpha} with $|t_j-\tau_j|\le h$, we have
\begin{equation}\label{eq:holder-cell}
\norm{H(\omega, t_j)-H(\omega, \tau_j)} \le C_\alpha (\omega)\, |t_j-\tau_j|^\alpha \le C_\alpha (\omega)\, h^\alpha.
\end{equation}
Combining \cref{eq:pointwise-product-diff} with \cref{eq:holder-cell} yields
\begin{equation}\label{eq:pointwise-bound}
\big\| H(t_k)\cdots H(t_1) - H(\tau_k)\cdots H(\tau_1)\big\|
\le k\, \Lambda(\omega)^{k-1}\, C_\alpha (\omega)\, h^\alpha.
\end{equation}
Integrating \cref{eq:pointwise-bound} over the $k$-simplex and using that its volume is $T^k/(k!)$, we obtain
\begin{align}
\norm{\mathcal{I}_k[H](\omega)-\mathcal{Q}_{k,M}[H](\omega)}
&= \left\| \int_{0\le t_1\le\cdots\le t_k\le T}
\Big(H(t_k)\cdots H(t_1) - H(\tau_k)\cdots H(\tau_1)\Big)\, dt\right\| \nonumber\\
&\le \int_{0\le t_1\le\cdots\le t_k\le T}
k\, \Lambda^{k-1}\, C_\alpha\, h^\alpha\, dt \nonumber\\
&= k\, \Lambda^{k-1}\, C_\alpha\, h^\alpha\, \frac{T^k}{k!}
= \frac{C_\alpha(\omega)\, \Lambda(\omega)^{k-1}\, T^k}{(k-1)!}\, h^\alpha,
\label{eq:IkQk-proof}
\end{align}
Summing \cref{eq:IkQk-proof} over $k=1,\dots,K$ gives
\begin{align}
\sum_{k=1}^{K} \norm{\mathcal{I}_k[H]-\mathcal{Q}_{k,M}[H]}
&\le C_\alpha\, h^\alpha \sum_{k=1}^{K} \frac{\Lambda^{k-1} T^k}{(k-1)!}
= C_\alpha\, h^\alpha\, T \sum_{k=1}^{K} \frac{(\Lambda T)^{k-1}}{(k-1)!} \nonumber\\
\label{eq:error-TDS-M}
&\le C_\alpha(\omega)\, T\, h^\alpha\, e^{\Lambda(\omega) T},
\end{align}
which completes the proof after substituting in the choice for $h$ in~\cref{eq:M-choice}.
\end{proof}

\begin{lemma}[Pathwise total TDS error under Hölder continuity]\label[lemma]{thm:TDS-path}
Assume \cref{eq:LambdaCalpha}. For any $\omega$, $T \le (\ln 2)/\Lambda(\omega)$ and $\eps>0$, there exist $K,h$ such that
\begin{equation}\label{eq:TDS-path}
 \mathcal{E}_{K,M}(\omega) \;\le\; \eps
\end{equation}
where specifically it suffices to choose
\begin{enumerate}
  \item $K = \left\lceil-1 + \frac{2\ln(1/\eps)}{\ln\ln(1/\eps) +1}\right\rceil$
\item $h \le \left(\frac{\eps}{C_\alpha(\omega)\, T\, e^{\Lambda(\omega) T}}\right)^{1/\alpha}$
\end{enumerate}
\end{lemma}

\begin{proof}
By \cref{eq:Dyson} and \cref{eq:TDS-KM}, the total error \cref{eq:total-error-def} is bounded by
\begin{align}\label{eq:error-decomp}
  \mathcal{E}_{K,M} & = \left\|\mathcal{T}\left[e^{-i\int_0^t H(s) \mathrm{d}s}\right]- \sum^K_{k=0} \mathcal{Q}_{k,M}[H] \right\| \\ 
		 \label{eq:error-decomp-cauchy}
		    &=  \left\|\mathcal{T}\left[e^{-i\int_0^t H(s) \mathrm{d}s}\right]- \sum^K_{k=0} (-i)^k \mathcal{I}_k + \sum^K_{k=0} (-i)^k \mathcal{I}_k - \sum^K_{k=0} \mathcal{Q}_{k,M}[H]\right\|\\ 
		    \label{eq:error-decomp-ineq}
		    & \le \left\|\mathcal{T}\left[e^{-i\int_0^t H(s) \mathrm{d}s}\right]- \sum^K_{k=0} (-i)^k \mathcal{I}_k\right\| + \sum_{k=1}^{K} \norm{\mathcal{I}_k[H]-\mathcal{Q}_{k,M}[H]} \\
		    \label{eq:error-decomp-K1}
		      & \le \underbrace{\sum_{k=K+1}^{\infty} \norm{\mathcal{I}_k[H]}}_{\mathcal{E}_K(\omega)}
+ \underbrace{\sum_{k=1}^{K} \norm{\mathcal{I}_k[H] - \mathcal{Q}_{k,M}[H]}}_{\mathcal{E}_M(\omega)} \\
		    \label{eq:TDS-pathwise}
		      & = \underbrace{\frac{(T \Lambda(\omega) e)^{K+1}}{K+1}}_{\mathcal{E}_K(\omega)}
\;+\; \underbrace{C_\alpha(\omega)\, T\, e^{\Lambda(\omega) T}\, h^\alpha}_{\mathcal{E}_M(\omega)}
\end{align}
for any given $\eps > 0$ and let $\mathcal{E}_K(\omega) \le \eps/2$ and $\mathcal{E}_M(\omega) \le \eps/2$.
\Cref{eq:error-decomp-ineq} is obtained by applying triangle inequality to \cref{eq:error-decomp-cauchy}. Plugging \cref{eq:e_K-Ik} into \cref{eq:error-decomp-ineq} yields \cref{eq:error-decomp-K1}. Combining \cref{eq:TDS-path-K} and \cref{eq:error-decomp-K1},  yields the result of the lemma for the cited values of $K$ and $h$ by substitution.
\end{proof}

We now provide probabilistic versions of \cref{thm:TDS-path}, allowing $\Lambda(\omega)$ and $C_\alpha(\omega)$ to be random.

\begin{lemma}[Probabilistic bounds for TDS under Hölder's continuity condition]\label[lemma]{lem:expectation}
Assume there exists a small perturbation parameter $\delta>0$ such that
$\EE\!\left[e^{(1+\delta) T\, \Lambda}\right] < \infty$
and also
$  \EE\!\left[C_\alpha^2\right] < \infty,\, 
  \EE\!\left[e^{2 T\, \Lambda}\right] < \infty,
$
we have
\begin{equation}\label{eq:EKM-exp-delta}
  \EE\!\left[\mathcal{E}_{K,M}\right]
  \;\le\; \left(\frac{1}{1+\delta}\right)^{K+1}\, \frac{1+\delta}{\delta}\, \EE\!\left[e^{(1+\delta)T \Lambda}\right]
  \;+\; T h^\alpha\, \sqrt{\EE[C_\alpha^2]\; \EE[e^{2T\Lambda}]}.
\end{equation}
Consequently, for any $\eps>0$,
\begin{equation}\label{eq:EKM-tail-delta}
  \PP\!\left(\mathcal{E}_{K,M} \ge \eps\right)
  \;\le\; \frac{\EE\!\left[\mathcal{E}_{K,M}\right]}{\eps}.
\end{equation}
Given a target mean error $\eps^*\in(0,1)$, the choices
\begin{equation}\label{eq:design-K-h}
  K \;\ge\; \left\lceil \frac{\log\!\Big(2\, \frac{1+\delta}{\delta}\, \EE[e^{(1+\delta)T\Lambda}]\Big) - \log \eps^*}{\log(1+\delta)} \right\rceil - 1,
  \qquad
  h \;\le\; \left(\frac{\eps^*}{2T\, \sqrt{\EE[C_\alpha^2]\, \EE[e^{2T\Lambda}]}}\right)^{\!1/\alpha}
\end{equation}
ensure $\EE[\mathcal{E}_{K,M}] \le \eps^*$. For tail control at level $\PP(\mathcal{E}_{K,M}\ge \eps)\le \rho$ with $\rho\in(0,1)$, replace $\eps^*$ by $\rho\,\eps$ in \cref{eq:design-K-h}.
\end{lemma}

\begin{proof}
 Taking the expectation of \cref{eq:TDS-pathwise}, we provide bounds for $\EE[\mathcal{E}_K]$ and $\EE[\mathcal{E}_M]$.  
Taking the expectation of \cref{eq:e_K-Ik}, by Tonelli and the exponential series domination,
\begin{equation}\label{eq:Tonelli-delta}
  \EE[\mathcal{E}_K] \;=\; \sum_{k=K+1}^{\infty} \frac{T^k}{k!}\, \EE[\Lambda^k]
  \;\le\; \sum_{k=K+1}^{\infty} \frac{T^k}{k!}\, \frac{k!}{\big((1+\delta)T\big)^k}\, \EE\!\left[e^{(1+\delta)T\Lambda}\right],
\end{equation}
since $e^{(1+\delta)T\Lambda} \ge \frac{((1+\delta)T\Lambda)^k}{k!}$ implies $\EE[\Lambda^k] \le \frac{k!}{((1+\delta)T)^k} \EE[e^{(1+\delta)T\Lambda}]$. Combining \cref{eq:Tonelli-delta} with the geometric sum gives
\begin{equation}\label{eq:AK-delta}
  \EE[\mathcal{E}_K] \;\le\; \EE[e^{(1+\delta)T\Lambda}] \sum_{k=K+1}^{\infty} \left(\frac{1}{1+\delta}\right)^k
  \;=\; \left(\frac{1}{1+\delta}\right)^{K+1}\, \frac{1+\delta}{\delta}\, \EE\!\left[e^{(1+\delta)T\Lambda}\right]
\end{equation}
Take the expectation of \cref{eq:error-TDS-M} and apply Cauchy--Schwarz inequality yields
\begin{equation}\label{eq:CS-EM-delta}
  \EE[\mathcal{E}_M] \;\le\; T h^\alpha\, \sqrt{\EE[C_\alpha^2]\; \EE[e^{2T\Lambda}]}.
\end{equation}
Combining \cref{eq:AK-delta} and \cref{eq:CS-EM-delta} proves \cref{eq:EKM-exp-delta}. Markov’s inequality for the nonnegative $\mathcal{E}_{K,M}$ then gives \cref{eq:EKM-tail-delta}. The design formulas in \cref{eq:design-K-h} follow by enforcing \cref{eq:AK-delta} $\le \eps^*/2$ and \cref{eq:CS-EM-delta} $\le \eps^*/2$, taking logarithms for \cref{eq:AK-delta} and rearranging the second.
\end{proof}

\begin{remark}
\citet[Lemma 5 \& Corollary~4]{low_hamiltonian_2019} analyze the truncated Dyson series assuming differentiability (or Lipschitz continuity) of $t\mapsto H(t)$ to bound the time-discretization (quadrature) error with a linear-in-$h$ rate. In contrast, \cref{lem:Ik-Qk,thm:TDS-path} require only Hölder continuity \cref{eq:LambdaCalpha}, and the quadrature error scales as $h^\alpha$ with $\alpha\in(0,1]$. Thus:
\begin{itemize}
\item When $H$ is merely Hölder-$\alpha$ with $\alpha<1$ (e.g., $\alpha<1/2$ for OU-driven terms), the rate degrades from $h^1$ to $h^\alpha$, but the truncated Dyson framework remains valid.
\item The truncation tail $\mathcal{E}_K$ remains unchanged and depends only on $\Lambda T$.
\end{itemize}
\end{remark}

\end{document}